\documentclass[11pt]{article}

\usepackage{amsmath,amssymb,amsfonts,xfrac}
\usepackage{ae,aecompl}

\newcommand{\aBL}{\textsc{Balance-Load}}
\newcommand{\aRP}{\textsc{Randomized-Per\-mu\-ta\-tions}}
\newcommand{\aDP}{\textsc{De\-ter\-min\-is\-tic-Per\-mu\-ta\-tions}}
\newcommand{\aMDP}{\textsc{Modified-Det-Perm}}
\newcommand{\aEP}{\textsc{Effort-Priority}}
\newcommand{\aDMY}{\textsc{DMY}}
\newcommand{\PartOne}{\textsc{Part-One}}
\newcommand{\PartTwo}{\textsc{Part-Two}}
\newcommand{\PartThree}{\textsc{Part-Three}}
\newcommand{\DA}{\textit{Do-All}}

\newcommand{\Item}{\B\item}
\newcommand{\qed}{\hfill $\square$ \smallbreak}
\newcommand{\Paragraph}[1]{\BBB\paragraph{#1}}

\newcommand{\cC}{{\mathcal C}}
\newcommand{\cD}{{\mathcal D}}
\newcommand{\cE}{{\mathcal E}}

\newcommand{\cI}{{\mathcal I}}
\newcommand{\cJ}{{\mathcal J}}

\newcommand{\cM}{{\mathcal M}}

\newcommand{\cO}{{\mathcal O}}

\newcommand{\cS}{{\mathcal S}}
\newcommand{\cT}{{\mathcal T}}

\newcommand{\cW}{{\mathcal W}}

\newcommand{\FF}{\vspace*{\medskipamount}}
\newcommand{\FFF}{\vspace*{\bigskipamount}}
\newcommand{\B}{\vspace*{-\smallskipamount}}

\newcommand{\BBB}{\vspace*{-\bigskipamount}}

\newenvironment{proof}{\noindent{\bf Proof:}}{\qed}

\newcounter{innerlistcounter}
\newenvironment{innerlist}{
\begin{list}{
	{\bf \alph{innerlistcounter}.}}{
  	\usecounter{innerlistcounter}
   	\itemindent0em 
   	\leftmargin0em 
   	\setlength{\rightmargin}{\leftmargin}
   	\itemsep1ex }}{
\end{list}}

\newtheorem{theorem}{Theorem}
\newtheorem{corollary}{Corollary}
\newtheorem{definition}{Definition}
\newtheorem{lemma}{Lemma}
\newtheorem{fact}{Fact}
\newtheorem{proposition}{Proposition}

\newlength{\pagewidth}
\newlength{\captionwidth}
\begin{document}

\baselineskip 	3ex
\parskip 		1ex

\title{		Doing-it-All with Bounded Work and Communication 
					\footnotemark[1]\vfill}

\author{	Bogdan S. Chlebus \footnotemark[2]  \and
		Leszek G\k asieniec \footnotemark[3] \and
		Dariusz R. Kowalski \footnotemark[3] \and
		Alexander A. Schwarzmann \footnotemark[4]}

\footnotetext[1]{The results of this paper were announced in a preliminary form in~\cite{ChlebusGKS-DISC02} and published in their final version as~\cite{ChlebusGKS-IC2017}.
The work of the first author was supported by NSF Grants 0310503 and 1016847.
The work of the third author was supported by Polish National Science Center UMO-2015/17/B/ST6/01897.
The work of the fourth author was supported by NSF Grants 0311368 and 1017232.
}

\footnotetext[2]{Department of Computer Science and Engineering, University of Colorado Denver, Denver, Colorado 80217, USA.
}

\footnotetext[3]{Department of Computer Science, University of Liverpool, Liverpool L69~3BX, UK.}

\footnotetext[4]{	Department of Computer Science and Engineering, University of Connecticut, Storrs, Connecticut 06269, USA.
}

\date{}

\maketitle

\vfill


\begin{abstract}
We consider the \DA\ problem, where  $p$ cooperating processors need to complete $t$ similar and independent tasks
in an adversarial setting. Here we deal with a synchronous message passing system 
with processors that are subject to crash failures. 
Efficiency of algorithms in this setting is measured in terms of \emph{work} complexity 
(also known as total available processor steps) and \emph{communication} complexity 
(total number of point-to-point messages). 
When work and communication are considered to be comparable resources, then the overall efficiency is meaningfully expressed in terms of \emph{effort} defined as \emph{work} + \emph{communication}.
We develop and analyze a constructive algorithm that has work $\cO( t +   p \log p\, (\sqrt{p\log p}+\sqrt{t\log t}\, )  )$  and a nonconstructive algorithm that has work $\cO(t +p \log^2 p)$.
The latter result is close to the lower bound $\Omega(t + p \log p/ \log \log p)$ on work.
The effort of each of these algorithms is proportional to its work when the number of crashes is bounded above by $c\,p$,  for some positive constant $c < 1$.
We also present a nonconstructive algorithm that has effort $\cO(t + p ^{1.77})$.

\vfill

\end{abstract}

\vfill

\thispagestyle{empty}

\setcounter{page}{0}

\newpage

\section{Introduction}

Performing a collection of tasks by processors prone to failures is among the fundamental problems in fault-tolerant distributed computing.
We consider the problem called \DA, where the processors cooperate on tasks that are similar, independent, and idempotent. 
Simple instances of this include checking all the points in a large solution space, attempting either to generate a witness or to refute its existence, and scheduling collections of tasks admitting ``at least once'' execution semantics.
We consider the synchronous setting with crash-prone processors
that communicate via point-to-point messages, 
where $t$ tasks need to be performed by $p$  processors subject to $f$ crash failures,
provided at least one processor does not crash (i.e., $f \le p-1$).

Synchronous message-passing solutions for \DA\ were first developed by Dwork et al.~\cite{DworkHW98} who estimated the performance of their algorithms in terms of \emph{effort} defined as the sum of task-oriented work and communication. 
\emph{Task-oriented work} counts only the processing steps expended on performing tasks and 
it discounts any steps spent idling or waiting for messages. 
\emph{Communication} costs are measured as the message complexity, calculated as the total number of point-to-point messages. 
Note that the time of algorithm executions may be very large, for example, when $f$ can be $p-1$, time
may be at least linear in $t$ for the cases where one remaining processor must perform all tasks.
Thus time is not normally used to describe the efficiency of \DA{} algorithms when large number of failures is allowed. 
We also measure algorithm performance in terms of \emph{effort}, except that we use a more conservative approach.
Instead of task-oriented work, we include \emph{available processor steps} complexity (or \emph{total work}) defined by Kanellakis and Shvartsman~\cite{KanellakisS92}  that accounts for all steps taken by each processor, including idling, until the processor either terminates the computation or crashes.
Thus we define the \emph{effort} of an algorithm to be $\cW + \cM$, where $\cW$ is its total work complexity and $\cM$ is its message complexity.
Other prior research focused on developing \DA\ algorithms that are efficient in terms of work, 
then dealing with communication efficiency as a secondary goal, e.g., using the \emph{lexicographic complexity}~\cite{PriscoMY94}.

Trade-offs  between work and communication in solutions of  \DA\  are to be expected, indeed, 
communication improves coordination among processors
with the potential of reducing redundant work
in unreliable systems.
There are two direct ways in which a processor can get to know that a certain task is complete
in a message-passing system:  
the processor can either perform the task, or it can receive a message  
that the task was completed by another processor. 
Note that \DA{} can be solved without any communication: 
simply have each processor perform each task.
The work of such an algorithm is $\cO(p\cdot t)$.
On the other hand, an algorithm may cause each processor to always share its knowledge with all other processors in an attempt to reduce work. 
This may result in efficient work, but the communication complexity 
of such an algorithm is $\Omega(p \cdot t)$: each time a processor completes a task, 
it sends a message to all other processors.
Thus it is highly desirable to develop algorithms for which both work is $o(p \cdot t)$
{and} communication is $o(p \cdot t)$.
This makes it meaningful to balance work and communication, and to consider them as comparable resources.
These observations motivate the use of the quantity $\cW + \cM$ as a unifying performance metric.

The performance bounds of our algorithms are expressed in terms of three parameters:
the number of processors~$p$, the number of tasks~$t$,  
and the number of crashes~$f$ that may occur in the course of an execution.
Our algorithms place no constraints on the relationship between $t$ and $p$;
these parameters are independent in our algorithms and the complexity bounds.
The only restriction on $f$ is that is that at least one processor does not crash,
i.e., $f<p$.

We say that a parameter is \emph{known} when it can be used in the code of an algorithm.
The parameters $p$ and $t$ are always known.
When the parameter $f$ appears in performance bounds, then this indicates that the algorithm
is designed to optimize performance for the number of crashes that is at most~$f$.
When $f<p$ is used as a parameter then $f$ is known.
When $f$ is not used in the code of an algorithm then it is only known that $f\le p-1$.
It so happens that if $f$ appears in an algorithm in this paper, then the corresponding
communication complexity depends on $f$, whereas the bounds on the work complexity involve only $p$ and~$t$.

We consider an adversarial setting, in which a nefarious adversary causes processors to crash.
An adversary is \emph{$f$-bounded} if $f<p$ is an upper bound on the number of crashes in any execution.
When stating a performance bound for a known $f<p$ we normally state that the bound holds 
``for an $f$-bounded adversary.''
The $(p-1)$-bounded adversary is called \emph{unbounded}.
An $f$-bounded adversary  is called \emph{linearly bounded} if $f\le c\,p$, for some positive constant $c<1$.
Our algorithms always solve \DA\ when exposed to the unbounded adversary, but their message complexity 
may  be especially efficient when the adversary is linearly-bounded. 

We call a deterministic algorithm \emph{constructive} if its code can be 
produced by a deterministic sequential algorithm in time polynomial in~$t$ and~$p$. 
This is in contrast with \emph{nonconstructive} algorithms that may rely on 
combinatorial objects that are only known to exist.
Methodologically, starting with a constructive algorithm, we trade constructiveness for 
better effort bounds in producing nonconstructive algorithms. 

We aim for algorithmic  solutions for \DA\ that attain good effort complexity, 
rather than seeking just work-efficient solutions.
Here a key challenge, besides tolerating crashes and 
controlling work, is to ensure that communication costs do not exceed work complexity.
Whereas any two processors can communicate in any step of computation, 
in our algorithms we limit communication by allowing messaging to take place
over certain constant-degree subnetworks.
To this end, we use constructive graphs with good ``expansion'' properties, and this contributes 
to the emerging understanding of how expansion-related properties of the underlying 
communication schemes can be used to improve fault tolerance and efficiency.


\Paragraph{Our results.}

We present a new way of structuring algorithms for the $p$-processor, $t$-task \DA\ problem that allows for both work \emph{and} communication to be controlled in the presence of adaptive adversaries.
We give a generic algorithm for performing work in systems with  crash-prone processors, 
and we parameterize it by $(i)$ task-assignment rules and $(ii)$ virtual overlay graphs superimposed on the underlying communication medium. 
We now detail our contributions.

\begin{itemize}
\Item[\sf I.]
We present a deterministic constructive algorithm, called \aBL, that uses a balancing task allocation policy (Section~\ref{sec:constructive}).
This algorithm solves  \DA{} in any execution with at least one non-crashed processor, and its performance is tuned to  a known upper bound $f$ on the number of crashes, where $f<p$.
The algorithm's work is $\cW=\cO( t +   p \log p\, (\sqrt{p\log p}+\sqrt{t\log t}\, ) )$, which does not depend on $f$,  while the message complexity does depend on~$f$.
When the adversary is additionally constrained to be linearly-bounded, the message complexity of the algorithm is  $\cM=\cO(\cW)$.
\end{itemize}

No prior algorithms using point-to-point messaging 
attained total work (available processor steps) 
that is both $o(p^2)$ and $o(t^2)$ against the $f$-bounded adversary, for any known $f<p$.
By using embedded graphs whose properties depend on $f$, each processor in our algorithm sends $\cO(1)$ messages in each round in the case of linearly-bounded adversaries,
making communication comparable to work.
(Section~\ref{sec:graphs} discusses the use of embedded graphs in communication.)

\begin{itemize}
\Item[\sf II.]
We develop a deterministic algorithm, called \aDP, that is more efficient than algorithm \aBL, but is nonconstructive (Section~\ref{sec:algorithm-optimized-for-work}).
This algorithm solves  \DA{} in any execution with at least one non-crashed processor, and its performance is tuned to  a known upper bound $f$ on the number of crashes, where $f<p$. 
This algorithm has work $\cW=\cO(t+p\log^2 p)$, which does not depend on~$f$, while the message complexity of the algorithm does depend on~$f$.
When the adversary is additionally constrained to be linearly-bounded, the message complexity of the algorithm is  $\cM=~\cO(\cW)$.
\end{itemize}

We note that the upper bound on work for this algorithm differs from the lower bound  $\Omega(t+p\log p/\log\log p)$ given 
in \cite{ChlebusK-RSA04,GeorgiouRS04} by the small factor $\log p \log\log p$.

\begin{itemize}
\Item[\sf III.]
We give a deterministic nonconstructive algorithm, called \aEP,
that is effort-efficient against the unbounded adversary (Section~\ref{sec:effort}).
The effort of this algorithm is $\cW+\cM=\cO (t+p^{1.77})$, for any unknown $f<p$.
\end{itemize}

This algorithm is obtained by combining algorithm \aDP\ with the algorithm 
of De~Prisco et al.~\cite{PriscoMY94} (although the authors did not consider the effort efficiency,
this was the first algorithm that obtained effort that is $\cO(t+p^2)$ against the unbounded adversary).


\newcommand{\RB}{\raisebox{3ex}{~}}
\newcommand{\LB}{\raisebox{-1.5ex}{~}}

\begin{table}[tp]
\begin{center}
\begin{tabular}{|@{\hskip -1pt}c@{\hskip 4pt} |l |l |}

\hline
\RB \LB
Paper & Work $\cW$ and message $\cM$ complexities
& Remarks \\
\hline
\hline
\RB \LB
\cite{DworkHW98}& $\cW=\cO(t+p)$& Task-oriented work  \\
\LB
(1992) & $\cM=\cO(p\sqrt{p})$&  \\
\hline
\RB \LB
\cite{PriscoMY94} & $\cW=\cO(t+p^2)$&\\
\LB
(1994) & $\cM=\cO(p^2)$&\\
\hline
\RB \LB
\cite{GalilMY95} & $\cW=\cO(t+p^2)$ & Known $f$ and $\varepsilon$   \\
\LB
(1995) & $\cM=\cO(f p^{\varepsilon} + p\log p)$ & \\
\hline
\RB \LB
\cite{ChlebusPS01} & $\cW=\cO\bigl(t\log f + p\frac{\log p\log f}{\log\log p}\bigr)$ & 
Reliable broadcast,  \vspace*{-1.5mm} \\
\RB \LB
(1997) & $\cM=\cO(t + p\frac{\log p}{\log\log p} + fp)$ & no restarts\\
\hline
\RB \LB
\cite{GeorgiouRS04} & $\cW=\cO\bigl(t\log f +p\log p \log f / \log \frac{p}{f} \bigr)$&Improved analysis for~\cite{ChlebusPS01}\\
\LB
(2004) & $\cM=\cO\bigl(t+p\log p / \log \frac{p}{f}+fp\bigr)$& above when $f \le p/\log p$\\
\hline
\RB \LB
\cite{ChlebusPS01} & $\cW=\cO(t\log f + p\log p\log f)$ &  Reliable broadcast; restated \\
\LB
(1997) & $\cM=\cO(t+p\log p +f p)$& for less than $p$ restarts\\
\hline
\RB\LB
\cite{ChlebusGKS-JDA08} & $\cW=\cO\bigl(t+p \frac{\log^2 p}{\log\log p}\bigr)$ & Randomized solution,\\
\RB \LB
(2002) & $\cM=\cO\bigl(\bigl(\frac{p}{p-f}\bigr)^{3.31}\,\cW\, \bigr)$ & with known $f$\\
\hline
\RB \LB
This paper & $\cW=\cO(t +   p \log p\, (\sqrt{p\log p}+\sqrt{t\log t}\, ) )$ & Known $f$ \\
\LB
(2002) & $\cM=\cW=\cO\bigl(\bigl(\frac{p}{p-f}\bigr)^{3.31}\, \cW\, \bigr)$ &\\
\hline
\RB \LB
This paper& $\cW=\cO(t+p \log^2 p)$ &Nonconstructive solution, \\
\LB
(2002) & $\cM=\cO\bigl(\bigl(\frac{p}{p-f}\bigr)^{3.31}\, \cW\, \bigr)$ & with known $f$ \\
\hline
\RB \LB
This paper& $\cW=\cO\bigl(t+p^{1.77}\bigr)$ & Nonconstructive solution,\\
\LB
(2002) & $\cM=\cO\bigl(t+p^{1.77}\bigr)$&optimized for \emph{effort} \\
\hline
\RB \LB
\cite{GeorgiouKS05} & $\cW=\cO\bigl(t+p^{1+\varepsilon}\bigr)$&Nonconstructive solution, \\
\LB
(2003) & $\cM=\cO\bigl(t+p^{1+\varepsilon}\bigr)$&optimized for \emph{effort}, $\varepsilon$  known\\
\hline
\end{tabular}
\parbox{\pagewidth}{\FF\caption{\label{summary} 
A summary of known solutions for \DA\ in the message-passing model.
The year of the initial  announcements of the results at a conference is given below each citation (in parenthesis).
Algorithms are deterministic and constructive, unless stated otherwise.
Broadcast is not assumed to be reliable, unless stated otherwise.
The upper bound $f$ on the number of crashes does not appear in the code, unless it is stated in the remarks that this quantity is known. 
The complexity is given in terms of work $\cW$ and message $\cM$ complexities.
Work is called task-oriented when idling and waiting are not accounted for, but each instance of performing  of a task contributes one unit of work.
}}
\end{center}
\end{table}


\Paragraph{Previous and related work.}

Dwork et al.~\cite{DworkHW98} were the first to study the \DA\ problem in the message-passing setting, assessing efficiency in terms of \emph{effort}.
De~Prisco et al.~\cite{PriscoMY94} were the first to use the available processor steps as the work measure for message-passing algorithms. 
They gave an algorithm with work $\cO(t+(f+1)p)$ and communication $\cO((f+1)p)$.
Galil et al.~\cite{GalilMY95} gave an algorithm with better communication cost, $\cO(fp^\varepsilon+\min\{f+1,\log p\}p)$, 
for any $\varepsilon>0$. 
A popular algorithmic paradigm used in prior research was to disseminate 
local knowledge among groups of processors by implementing \emph{coordinators}, who first collect and then spread information.
Chlebus et al.~\cite{ChlebusPS01} developed algorithms based on aggressive coordination, where the number of appointed coordinators grows exponentially following crashes of all previously appointed coordinators.
Their algorithms rely on \emph{reliable multicast}, where if a sender crashes during a multicast, 
then either none or all of the messages are delivered to the non-faulty recipients.
One of their algorithms has work $\cO((t + p\log p/\log\log p)\log f)$  and communication $\cO(t + p\log p/\log\log p + fp)$; 
these bounds were improved by Georgiou et al.~\cite{GeorgiouRS04} for $f\le p/\log p$.
Recent work by Davtyan et al.~\cite{DPGS14} explores a way to use an unreliable broadcast and presents an experimental study.
Another algorithm in~\cite{ChlebusPS01} incorporates restarted processors. 
It is the only known algorithm able to deal with restarts efficiently; it has work $\cO((t + p\log p + f)\cdot \min\{\log p,\log f\})$ and its message complexity is $\cO(t+p\log p+ f p)$.
In~\cite{ChlebusGKS-JDA08}, we gave a randomized solution with work $\cO(t+p \frac{\log^2 p}{\log\log p})$ for the $f$-bounded adversary, for any known $f<p$.
This paper considers the same generic algorithm and uses the same overlay graphs as in the current paper, but with a randomized rule for  tasks selection by processors.
Georgiou et al.~\cite{GeorgiouKS05} subsequently developed an algorithm with effort $\cO(t+p^{1+\varepsilon})$ against the unbounded adversary, for any known~$\varepsilon$, using an approach based on gossiping.
A summary of  results for \DA\ related to the results of this paper is given in Table~\ref{summary}.
The book~\cite{GeorgiouS-book08} by Georgiou and Shvartsman gives a representative account of the topics in \emph{Do-All} computing.

Table~\ref{summary} does not present results dealing with specialized adversaries and 
models of inter-processor communication;
we review such work here.
Chlebus and Kowalski~\cite{ChlebusK-RSA04} studied \DA\ with crashes in the presence of a weakly-adaptive linearly-bounded adversary, giving
a randomized algorithm with $\cO\big(t + p \big(1+\log^*\big(\frac{t\log p}{p}\big)\big)\big)$ 
expected effort.
This bound is provably better than for any deterministic algorithm subject to the $\Omega(t+p\log p/\log\log p)$ lower bound \cite{ChlebusK-RSA04, GeorgiouRS04} on work.
They also present a deterministic algorithm  that schedules tasks by  balancing them perfectly and performs $\cO(t+p\log p/\log\log p)$ work against the linearly-bounded adversary;
this implies that $\Theta(t+p\log p/\log\log p)$ is precisely the optimum  work in such a setting.
Kowalski and Shvartsman~\cite{KowalskiS05} studied the \DA\ problem in the message-passing model with restricted asynchrony, where every message delay is at most~$d$,  for some value~$d$ unknown to the algorithm.
They showed that  $\Omega(t+p\,d \log_d p)$ is a lower bound on the expected work and developed  a deterministic algorithm with work $\cO((t+pd)\log p)$. 
Georgiou et al.~\cite{GeorgiouRS04} considered an iterated version  of \DA.
In~\cite{GeorgiouRS05} the same authors studied a problem, called \emph{Omni-Do}, that is an on-line version of \DA\ in asynchronous systems with partitionable networks. 
They gave a randomized algorithm achieving optimal competitive ratio against oblivious adversaries. 
Fern\' andez et al.~\cite{FernandezGRS05} considered \DA\  with Byzantine failures.
Chlebus et al.~\cite{ChlebusKL-DC06} and Clementi et al.~\cite{ClementiMS02} studied \DA\ 
in a model with communication over a multiple-access channel.
The former~\cite{ChlebusKL-DC06} considered the impact of collision detection and randomization  on the complexity of protocols.
The latter~\cite{ClementiMS02} gave tight bounds on work and time for such deterministic algorithms against $f$-bounded adversaries.

The use of \emph{gossip} emerged as a paradigm for solving  \DA\ and related problems
 in crash-prone message-passing systems after the conference 
 version~\cite{ChlebusGKS-DISC02} of this paper was published.
It was first applied by Chlebus and Kowalski~\cite{ChlebusK-JCSS06} 
to obtain a message-efficient solution to consensus in synchronous message passing.
This paper reformulated gossiping in a crash-prone environment as a problem in which  each processor starts with a rumor to be ideally learned by all processors.
Here, after a processor $v$ crashes, any other processor must either learn $v$'s rumor 
or the fact that $v$  crashed.
That paper proposed solutions to gossiping with $\cO(p\text{ polylog } p)$ message
complexity for $p$ processors and suitably restricted adversaries, and with $\cO(p^{1.77})$ message complexity for unbounded adversaries, 
using the same underlying communication overlay 
topologies  as in the current paper.
Georgiou et al.~\cite{GeorgiouKS05} showed how to obtain  
$\cO(p^{1+\varepsilon})$ communication for the unbounded adversary 
while maintaining suitably good time performance, for any fixed $\varepsilon>0$.
With this, they gave an algorithm for  \DA\ with  $\cO(t+p^{1+\varepsilon})$  effort
against the unbounded adversary, again for any fixed $\varepsilon>0$.
Kowalski et al.~\cite{KowalskiMS-ICDCS05} showed how 
to obtain a \emph{constructive} solution 
for gossiping with crashes with  $\cO(p^{1+\varepsilon})$ communication
for the unbounded adversary, for any fixed $\varepsilon>0$.
Chlebus and Kowalski~\cite{ChlebusK-DISC06} developed a constructive gossiping solution  with time $\cO(\text{polylog } p)$  while sending  $\cO(p\text{ polylog } p )$ point-to-point messages, and demonstrated its efficiency when applied to consensus.
Georgiou et al.~\cite{GeorgiouGGK13} studied the impact of adversarial models, 
both oblivious and adaptive, on 
the efficiency of gossiping in asynchronous message passing.

Problems related to \DA\  were studied in distributed systems with shared read/write registers.
Bridgland and Watro~\cite{BridglandW87} considered tasks in asynchronous systems 
with crash-prone processors under the constraint that each processor 
can perform at most one task.
Another related problem is the problem of collecting all values originally stored in
a set of registers. Here each process proceeds by performing reads and storing the collected values in its own register 
(each register has the capacity to store all needed values).
A read from register $x$ followed by a write to register $y$ adds 
the contents of register $x$ to the contents of register~$y$.
Since each register storing a value needs to be read at least once, this resembles 
the problem of performing each task at least once.
A solution to the problem of collecting values was first used by 
Saks et al.~\cite{SaksSW91} in their consensus algorithm.
Chlebus et al.~\cite{ChlebusKS-STOC04} showed that $n$ processors in an asynchronous system  can collect $n$ values using $\cO(n)$ registers in $\cO(n\text{ polylog } n)$ total work.
\emph{Write-All}, the problem introduced by Kanellakis and Shvartsman~\cite{KanellakisS92}, 
is about writing at least once to each  register in a collection.
Here writing  to a  register is viewed as a task in the same sense as in \DA.
Chlebus and Kowalski~\cite{ChlebusK-STOC05} showed that $n$ processors 
can write to each of $\cO(n)$  registers using $\cO(n\text{ polylog }n)$ total work.
Alistarh et al.~\cite{AlistarhBGG12} studied \DA\ in 
an asynchronous shared-memory model and gave
  a deterministic non-constructive algorithm that performs $t$ tasks on $p$ processors with work $\cO(t+p\text{ polylog}(p+t))$.
Kentros et al.~\cite{KentrosKNS09} introduced the \emph{At-Most-Once} problem where
the goal is for a set of processors to perform as many tasks as possible
provided that each task is performed at most once.
They studied this problem for shared memory with asynchronous crash-prone processors,
where efficiency is measured by  \emph{effectiveness} that factors in the number of tasks completed.
They gave a lower bound on the number of performed tasks, 
showing that it was impossible to perform all  tasks, 
and  gave an algorithm with effectiveness close to the lower bound.
Kentros and Kiayias~\cite{KentrosK13} solve \emph{At-Most-Once} 
 with improved effectiveness.
Kentros~et al.~\cite{KentrosKK12} introduced the \emph{Strong-At-Most-Once} problem, 
in which all tasks must be performed in the absence of crashes, 
and showed that it has consensus number~$2$.
Censor-Hillel~\cite{Hillel10} used a randomized wait-free 
solution to multi-valued consensus  as a building block in an algorithm for \emph{At-Most-Once} of optimal effectiveness.

In recent years, Internet supercomputing has become an increasingly popular 
means for harnessing the power of a vast number of interconnected computers. 
With this come the challenges of marshaling distributed resources and dealing with failures. 
Traditional centralized approaches employ a master processor and many worker processors 
that execute a collection of tasks on behalf of the master.
Despite the simplicity and advantages of centralized schemes, the master processor is 
a performance bottleneck and a single point of failure. Additionally, a phenomenon of increasing
concern is that workers may return incorrect results. Thus completely decentralized
fault-tolerant solutions are of interest. Representative works in this area include
the papers by Davtyan et al.~\cite{DavtyanKS-OPODIS2011,DavtyanKRS15}, also formulated for synchronous, crash-prone, message-passing settings, but with different
failure models than in this work. 
We also note that the goal in those works is for all non-crashed processors
to learn the results of all tasks, while in our work it is sufficient for the 
processors to know that the tasks have been performed.

The solutions to \DA\ given in~\cite{PriscoMY94} and~\cite{GalilMY95} employed a group membership routine (called checkpointing) that enables all non-crashed processors to obtain the same knowledge about the set of the remaining processors.
In the current paper, we also use a solution to a membership problem, which is integrated into the algorithm described in Section~\ref{sec:effort}.
Group membership abstractions were already proposed by Birman and Joseph~\cite{BirmanJ87}, and intensely researched since then.
Solutions to problems related to that of group membership are discussed in the books by Cachin et al.~\cite{CachinGR2011} and Birman~\cite{Birman2012}.

Expanders are small-degree graphs that have strong connectivity-related properties.
The ``expansion'' properties used to define expanders usually mean either that ``small'' sets of nodes have ``many'' neighbors, or that any two ``large'' disjoint sets of nodes are connected by at least one edge.
The former approach is discussed by Chung~\cite{Chung-book97}.
The latter is considered by Pippenger~\cite{Pippenger87}, who defined an \emph{$a$-expander} to be a graph where, for any two disjoint sets $X_1$ and~$X_2$ of nodes of size $a$ each, there is an edge $(x_1,x_2)$
in the graph, where $x_1\in X_1$ and $x_2\in X_2$.

There is a substantial body of research on fault-tolerance in networks with expansion properties. 
Here we cite the most relevant work.
Alon and Chung~\cite{AlonC88} showed that, for every $0<\varepsilon< 1$ and every integer $m>0$, there are constructive graphs with $\cO(m/\varepsilon)$ vertices and maximum degree $1/\varepsilon^2$ that have the property that after removing any $1-\varepsilon$ fraction of 
 vertices or edges the remaining graph contains a path of length~$m$.
 They showed this by using the expanders of Lubotzky et al.~\cite{LubotzkyPS88} with suitably chosen parameters.
 A protocol is said to achieve an $h(f)$-agreement if, in each execution with at most $f$ processors failures,  at least $p-h(f)$ non-faulty processors eventually decide on a common value. 
A protocol achieves an almost-everywhere agreement if it achieves $h(f)$-agreement for some function $h(f)\le \mu f$, where the constant~$\mu$ is independent of $f$ and of the size of the network.
Upfal~\cite{Upfal94} showed how an almost-everywhere agreement can be achieved with a linear number of faults, extending the result by Dwork et al.~\cite{DworkPPU88}. 
This result uses the same family of expanders from~\cite{LubotzkyPS88}, 
and it is shown that 
if one removes from the graph any subset of vertices of cardinality that
is a constant fraction of the total number of vertices,  the diameter of
the remaining linear-size graph is  logarithmic.
The graph construction in Section~\ref{sec:graphs}, in which we consider powers of Ramanujan graphs,  is an extension of that approach.

Magnifying the expansion properties of a graph by taking its power was already suggested 
by Margulis~\cite{Margulis73},  and used, for instance by Diks and  Pelc~\cite{DiksP00}, 
in the context of broadcasting and fault diagnosis. 
Goldberg et al.~\cite{GoldbergMP94} developed a deterministic distributed algorithm to reconfigure an $n$-input  multi-butterfly network with faulty switches so that it still can route any permutation between some sets of~$n-\cO(f)$ inputs and outputs, for any number of~$f$ faulty switches.
Chlebus et al.~\cite{ChlebusKS-STOC04} showed that $a$-expander networks have the property that for any set of~$3a$ non-faulty nodes there are $2a$ of these nodes that can communicate  in $\cO(\log a)$ time only among themselves.

Linearly-bounded adversaries have been used in other contexts, e.g., 
Lamport et al.~\cite{LamportSP82} showed that the consensus problem for Byzantine faults in synchronous settings can be solved against $f$-bounded adversaries for $f<p/3$.
Diks and Pelc~\cite{DiksP00} studied deterministic broadcast and fault diagnosis in complete networks, with processor faults controlled by  linearly-bounded adversaries.
Kontogiannis et al.~\cite{KontogiannisPSY00} studied a fault-prone 
bulk-synchronous parallel computer model, where processor faults are controlled by a variant of oblivious linearly-bounded adversaries.


\Paragraph{Document structure.}

We define  models and efficiency measures in Section~\ref{sec:model}.
We define  communication schemes based on  overlay graphs and give their  properties  in Section~\ref{sec:graphs}. 
Section~\ref{sec:design-algorithms} contains our generic algorithm 
and its analysis.
In the sequel we instantiate the generic algorithm and analyze 
efficiencies of the resulting derivations.
Section~\ref{sec:constructive} presents a constructive algorithm based on load balancing.
Section~\ref{sec:algorithm-optimized-for-work} describes an algorithm optimized for work.
Section~\ref{sec:effort} gives an algorithm optimized for effort.
We conclude  in Section~\ref{sec:discussion}.

\section{Models and Technical Preliminaries} 

\label{sec:model}

We consider a  system of~$p$ processors with unique processor identifiers in the interval~$[1,p]$.
Each processor knows the value of~$p$, and $p$ may appear in the code of the algorithms.


\Paragraph{Communication.}

Processors communicate by sending messages.
The size of a message is assumed to be sufficiently large to contain the local state of a processor, as determined by the executed distributed algorithm. 
Communication is modeled at the ``transport layer'' level of abstraction, meaning that the network appears to the processors to be completely connected. 
Thus we abstract away the physical network, where they may or may not be direct links between processors, with routing provided by the ``network layer.''


\Paragraph{Synchrony.}

The system is synchronous, where the local steps of all processors are governed by a global clock.
An execution is structured as a sequence of  rounds.
A \emph{round} is defined to be a fixed number of clock cycles that is sufficiently large but minimal 
for the following to occur: 
(1) a processor can send messages to several other processors in one round,
(2) messages sent during a round are delivered and received later in that round, 
(3) the duration of a round is sufficient to perform the needed local processing.


\Paragraph{Failures and adversaries.}

Processors fail by crashing.
We denote by $f$ an upper bound on the number of failures that may occur in an execution.
We do not consider processor restarts and so $f$ never exceeds~$p$.
Multicast messaging is \emph{not} assumed to be atomic with respect to failures: 
if a processor crashes while attempting to multicast a message, then some arbitrary subset of the destinations may receive the message.
We assume that no messages are lost or corrupted in transit.

The occurrence of failures is governed by adversarial models that determine upper bounds on the number of processor failures and the specific subsets of processors that may crash at some instant.
An adversary is called \emph{$f$-bounded} when at most $f$ crashes may occur in an execution.
An adversary is  \emph{linearly bounded} when it is $f$-bounded and~$f$ satisfies $f\le cp$, for a positive constant $c<1$, which means that $p-f=\Omega(p)$.
The $(p-1)$-bounded adversary is also called \emph{unbounded}.
The unbounded adversary is subject only to the weakest restriction that at least one processor does not crash in any execution.
Adversaries are always \emph{adaptive} in this paper, in the sense that the decisions about the failures of specific processors are made on-line.
The adversaries are subject only to the upper bound on the number of crashes allowed to occur in an execution, unless stated otherwise (as in Section~\ref{sec:restricted-adversaries}).


\Paragraph{The \DA{} problem and tasks.}

The \DA{} problem is to perform a given set of $t$ tasks in the presence of an adversary.
The tasks have unique identifiers in the interval~$[1,t]$. 
Each processor knows the value of~$t$, and $t$ may appear in the algorithm code.
Any processor can perform any task when given the tasks's identifier. 
Tasks are \emph{similar} in terms of time needed to perform one; we assume that any task can
be performed within a single round and that one round is needed to perform one task by one processor.
Tasks are \emph{idempotent}, obeying at-most-once execution semantics: any task can be performed multiple times and concurrently. 
Tasks are also \emph{independent}: they can be performed in any order.


\Paragraph{Correctness and termination.}

A processor may either voluntarily stop working  in a given execution,
or it may be forced to do so by an adversary that causes it to fail.
In the former case, we say that the processor \emph{halts}, and in the latter case the processor \emph{crashes}.
Processors may halt in different rounds.
Halted processors are considered non-faulty.
We say that an algorithm \emph{solves the \DA\ problem} against a given adversary, 
if the following two conditions are satisfied in any execution consistent with the adversary:
\begin{enumerate}
\Item
Each task is eventually performed by some processor (or processors). 
\Item
\label{cor}
Each processor eventually halts, unless it crashes. 
\end{enumerate}
Our algorithms solve \DA\  in the presence of the unbounded adversary.
Observe that this form of correctness implies that when a processor~$v$ halts, 
then all the tasks must have been performed.
This is because if a non-faulty processor halts when not all tasks are complete, 
then the unbounded adversary can immediately crash all other processors,
and thus there will remain tasks that have not been completed.

We say that an algorithm \emph{terminates} in a round if this is the first round
where each processor either halted or crashed.
It is not a requirement that algorithms terminate
by having all correct processors halt simultaneously.


\Paragraph{Performance metrics.}

We consider work  and communication as metrics of performance of algorithms. 
\emph{Work} $\cW$ is the maximum over all executions of  the sum, for all processors, of the number of rounds a processor performs in an execution until the processor crashes or the algorithm terminates, whichever occurs first. 
This means that we count the available processor steps~\cite{KanellakisS92}. 
In particular, halted processors continue contributing to work until the algorithm terminates.
\emph{Communication~$\cM$} is the maximum over all executions of the total number of point-to-point messages  sent in an  execution;
multicasts are accounted for in terms of the point-to-point messages according to the number of destinations for each multicast.
We define \emph{effort} as $\cE=\cW+\cM$.

\section{Graphs and Communication} 

\label{sec:graphs}

To achieve our efficiency goals we control the communication costs of our algorithms
by allowing processors to send messages only to a selected subset of processors during a round.
In this section, we specify graphs used by the algorithms to limit communication. 
Each such graph forms a conceptual overlay network over the assumed complete network.

We consider simple graphs, that is, undirected graphs with at most one edge between any two nodes.
A graph $G=(V,E)$ is defined in terms of its set of nodes~$V$ and the set of edges~$E$.
We denote by $|G|$ the number of nodes in the graph.
A graph $H=(V',E')$ is a \emph{subgraph of~$G=(V,E)$}, 
denoted by $H\subseteq G$, when $V'\subseteq V$ and $E'\subseteq E$.
A subgraph~$H=(V',E')$ of~$G=(V,E)$ is an \emph{induced subgraph of~$G$} when, 
for any  pair of nodes $x$ and $y$ in~$V'$, if  $(x,y)$ is an edge in~$E$ then it is also an edge in~$E'$.
An induced subgraph~$H=(V',E')$ is said to be \emph{induced by~$V'$}, 
since the set of nodes $V'$ determines the edges in~$E'$.

If $(x,y)$ is an edge in a graph, we say that node~$x$  is a \emph{neighbor of~$y$}
(and $y$ is a neighbor of $x$).
For a subset $V_0\subseteq V$,  we define the set of \emph{neighbors of~$V_0$}, denoted $N_G(V_0)$, to be the set of all nodes of~$G$ that have neighbors in~$V_0$.
Extending the neighborhood notation, we denote by $N^k_G(V_0)$, for a positive integer $k$, to be 
the set of all nodes in~$G$ connected to at least one node in~$V_0$ by a path of length at most~$k$.
The graph $G^k = (V,E')$ is called the \emph{$k^{th}$ power of graph~$G$} if 
for any pair of nodes $u$ and $v$ in~$V$, $(u,v)$ is an edge in~$E'$ 
if and only if there is a path between $u$ and~$v$ in~$G$ of length at most~$k$.
Observe that for any $V_0\subseteq V$ we have $N^k_G(V_0) = N_{G^k}(V_0)$.

We model each processor as a node, and any pair of processors that are to communicate regularly is modeled as an edge of an \emph{overlay graph}.
For $p$ processors and a bound $f<p$ on the number of crashes, the suitable overlay graph is denoted by~$G(p,f)$.
Each crashed or halted processor is removed from this graph, 
causing the remaining portion of the overlay graph to evolve through a sequence of subgraphs, 
until the graph vanishes altogether when the algorithm terminates.

Definitions~\ref{def:compactness} and~\ref{def:subgraph-property} below are parameterized by three constants $\alpha$, $\beta$ and~$\gamma$, where $\alpha$ is real such that $0<\alpha<1$, and $\beta$ and $\gamma$ are positive integers, and also by the number of failures~$f$ that are tolerated in an execution; we will choose
the specific constants $\alpha$, $\beta$ and~$\gamma$ later.
We use the following notation for logarithms: $\lg x$ means $\log_2 x$, 
while $\log x$ is used when the constant base of the logarithm does not matter, as in the asymptotic notation $\cO(\log n)$.


\begin{definition}
\label{def:compactness}

Let $f$ be an upper bound on the number of crashes.
We say that a subgraph~$H\subseteq G(p,f)$ is \textbf{compact} if $|H|\ge \alpha(p-f)$ and the diameter of~$H$ is at most $\beta \lg p+\gamma$.
\qed
\end{definition}

For a graph~$G$, a \emph{subgraph function $P$ on $G$} assigns an induced subgraph $P(G_0)\subseteq G$ to any induced subgraph~$G_0\subseteq G$.


\begin{definition}
\label{def:subgraph-property}

Graph $G(p,f)$ has \textbf{the subgraph property for~$f$} if there exists a subgraph function~$P$ on $G(p,f)$ that satisfies the following conditions for induced subgraphs $G_0$ and $G_1$, each of at least $p-f$ nodes:
\begin{enumerate}
\item
$P(G_0)\subseteq G_0$; 
\item
$|P(G_0)|\ge\alpha |G_0|$;
\item
The diameter of~$P(G_0)$ is at most $\beta\lg p +\gamma$;
\item
If $G_0\subseteq G_1$ then $P(G_0)\subseteq P(G_1)$.
\qed
\end{enumerate}
\end{definition}


\Paragraph{Graphs with good expansion properties.}

We denote by $L(n,d)$ the constructive Ramanujan graph on $n$ nodes and of node degree~$d$ as given by Lubotzky et al.~\cite{LubotzkyPS88}, for positive integers $n$ and~$d$.
A self-contained exposition of this construction and a discussion of the expansion-type properties of the graphs are  given by Davidoff et al.~\cite{DavidoffSV-book03}; Theorem~4.2.2 therein is relevant in particular.
Any number~$d$ such that $d -1$ is an odd prime can be chosen as a node degree of~$L(n,d)$. 
The number of nodes~$n$ is either of the form $r(r^2-1)/2$ or of the form $r(r^2-1)$, where $r>2\sqrt{d-1}$ is an odd prime, depending on whether or not the prime $d -1$ is a quadratic residue modulo~$r$, respectively.
We will use a fixed  node degree denoted  by~$\Delta_0$. 

Given a number $p$ of processors, we set $x$ to the smallest positive integer with the property that $x(x^2-1)/2\ge p$.
Let $r_p$ be a prime number between $x$ and~$2x$ that must exist by Chebyshev's Theorem on the distribution of prime numbers.
Consider the corresponding Ramanujan graph of node degree $\Delta_0$ and of either $n_p=r_{p} (r_{p}^{2}-1)/2$ or $n_p=r_p(r_p^2-1)$ nodes, assuming additionally that $n_p$ is sufficiently large. 
This graph $L(n_p,\Delta_0)$ has the properties that $n_p\ge p$ and $n_p=\Theta(p)$. 
If $n_p=p$ then nothing else needs to be done to model the $n_p$ processors.
Otherwise, we let each processor simulate~$\cO(1)$ nodes, so that $n_p=\Theta(p)$ nodes of $L(n_p,\Delta_0)$ are simulated.
The simulating set of $p$ processors are prone to crashes, and when such a simulating processor crashes,  then this results in $\Theta(1)$ crashes of the simulated nodes. 
For the simulation to provide the given results, it is sufficient  for the presented constructions to have the property that replacing~$p$ by~$c p$ and $f$ by~$c f$, where $c>1$ denotes the number of processors simulated by one simulating process, preserves the validity of the used arguments and their conclusions. 
It is actually only Lemmas~\ref{lem:degree-power-of-Lp}, \ref{lem:subgraph-property} and~\ref{lem:disjoint-are-connected} in this Section that directly refer to the properties of the  graphs~$L(p,\Delta_0)$, and they pass the test, by inspection.  
This allows us to simplify the notation, so that henceforth we assume that~$n_p=p$ and refer to the graphs~$L(p,\Delta_0)$ by~$L(p)$.

The expansion of a graph can be estimated by the eigenvalues of its adjacency matrix.
The matrix is symmetric, hence its eigenvalues are real.
Let the eigenvalues be denoted by 
\[
\lambda_0\ge \lambda_1\ge\ldots\ge \lambda_{p-1}\ , 
\]
for a graph of~$p$ vertices.
The largest eigenvalue $\lambda_0$ of a regular graph is equal to its degree,
which in the case of the graphs~$L(p)$ is~$\Delta_0$.
The graphs~$L(p)$ have an additional property that their second largest in absolute value
eigenvalue  $\lambda=\max\{|\lambda_1|,|\lambda_{p-1}|\}$ satisfies the inequality 
\begin{equation}
\label{eqn:ramanujan}
\lambda\le 2\sqrt{\Delta_0 -1}\ ,
\end{equation} 
as shown in~\cite{DavidoffSV-book03,LubotzkyPS88};
this is the property that defines \emph{Ramanujan Graphs}
among regular connected graphs of the node degree~$\Delta_0$.

We will refer to the following two facts known from the literature.


\begin{fact}[\cite{Tanner84}] 
\label{fact:tanner}

For a regular graph~$G$ with $p$ nodes and a subset of nodes $V_0$, the size of~$N_G(V_0)$ can be estimated by  the two largest eigenvalues of~$G$ as follows:
\[
|N_G(V_0)| \ge 
\frac{\lambda_0^2 |V_0|}{\lambda^2+(\lambda_0^2-\lambda^2)|V_0|/p}\ .
\]
\end{fact}


\begin{fact}[\cite{Upfal94}]
\label{fact:upfal}

If $\Delta_0=74$ is used in the specification of graph~$L(p)$ then there exists a function~$P'$ such that if $L_0$ is a subgraph of~$L(p)$ of size at least $71p/72$, then the subgraph~$P'(L_0)$ of~$L_0$ is of size at least $|L_0|/6$ and its diameter is at most $30\lg p$. 
\end{fact}

In terms of Definition~2, Fact~\ref{fact:upfal} means that if $\Delta_0=74$ is used in the construction of graph~$L(p)$, and if $f\le p/72$, then $L(p)$ has the subgraph property for~$f$ with $\alpha=1/6$, $\beta=30$ and~$\gamma=0$.
Next we will show how to extend this fact so that the restriction $f\le p/72$ can be abandoned, for appropriately adjusted constants $\alpha$, $\beta$ and $\gamma$. 
To this end, we use the overlay graphs~$L(p)^j$, for suitably large~$j$,  that have stronger connectivity properties  than those of~$L(p)$.
Observe that, for a given exponent~$j$ in this construction, the maximum node degree~$\Delta$ of~$L(p)^j$ is at most~$\Delta_0^j$.
We will press into service the function $P'$ from Fact~\ref{fact:upfal} several times in the sequel.


\Paragraph{The specification of the overlay graphs.}

The definitions of graph~$L(p)$ and the notions of compactness and subgraph property are parametrized by the constants $\Delta_0$, $\alpha$, $\beta$, and~$\gamma$.
From now on,  we fix these constants to the values $\Delta_0=74$, $\alpha=1/7$, $\beta=30$, and~$\gamma=2$.
Let $N(R)$ denote $N_{L(p)}(R)$, for any set $R$ of nodes of~$L(p)$.

The overlay graphs~$G(p,f)$ are defined as $L(p)^\ell$ for a suitable exponent~$\ell$.
If $f\le p/72$, then we may take $\ell=1$ and rely on Fact~\ref{fact:upfal}.
Let us suppose that $f>p/72$.
Let $R$ be a set of fewer than $71p/72$ nodes.
We estimate the size of the neighborhood~$N(R)$ by the inequality of Fact~\ref{fact:tanner} and  property~\eqref{eqn:ramanujan} by which $L(p)$ is a Ramanujan graph, to obtain the following bounds:
\begin{equation}
\label{eqn:specialized-tanner}
|N(R)|
\ge 
\frac{\Delta_0^2}{\lambda^2+(\Delta_0^2-\lambda^2) \frac{|R|}{p}}|R| 
\ge  
\frac{5476}{292+5184\cdot\frac{|R|}{p}}|R|
\ .
\end{equation}
If the size $|R|$ is at most $p/50$, then $|N(R)|\ge \rho |R|$ holds for the constant~$\rho = 27/2$ that is obtained by substituting $p/50$ for~$|R|$ in the estimate~\eqref{eqn:specialized-tanner}.
If $|R|$ is between $p/50$ and~$71 p/72$, then $|N(R)|\ge~\!\!\rho_1|R|$ for the constant $\rho_1 = 1.013$ that is obtained by substituting $71 p/ 72$ for~$|R|$ in the estimate~\eqref{eqn:specialized-tanner}.

A benefit brought in by replacing the estimate~\eqref{eqn:specialized-tanner} by the products $\rho |R|$ and $\rho_1 |R|$ is to be able to express the estimates on the sizes of neighborhoods of the form $N^i(R)$ in a power $L(p)^i$ of graph~$L(p)$ by the powers of $\rho$ and~$\rho_1$, with the corresponding exponents suitably related to~$i$.

We want to take a power~$k$ of~$L(p)$ that is sufficiently large to expand any neighborhood of set~$R$ of at least $p-f$ nodes to~$71p/72$~nodes, then use function~$P'$ from Fact~\ref{fact:upfal}, and finally rely on the exponent~$\ell$ to be sufficiently large to translate this action into the suitable properties of graph~$L(p)^\ell$.
Let us define the three positive integers $r$, $k$, and~$\ell$  as follows:
\begin{itemize}
\item[\rm a)] Let $r$ be the smallest positive integer such that $(p-f)\rho^{r}> p/50$.

\item[\rm b)] Given this~$r$, let $k>r$ be the smallest integer that is large enough so that the inequality  $(p-f)\rho^{r}(\rho_1)^{k-r}> 71p/72$ holds.

\item[\rm c)] Set  $\ell=2k+1$.
\label{eqn:l=2+1}
\end{itemize}
The existence of~$r$ is by the choice of~$\rho = 27/2$, while $k$ is well defined by the choice of~$\rho_1$, as specified just after the estimate~\eqref{eqn:specialized-tanner}, which means that graph~$L(p)^\ell$ is well defined.
We let $G(p,f)$ denote the graph~$L(p)^\ell$.

Given a set of nodes $R$ in~$G(p,f)$, we may estimate the size of the set $N^k (R)$ in~$G(p,f)$ as follows. 
If $R$ is a set of nodes such that $|R|\ge p-f$ then the following inequality holds:
\begin{equation}
\label{eqn:exponent}
|N^k (R)| > \frac{71}{72}\, p \ .
\end{equation}
This is because parts a) and b) of the specification of $L(p)^\ell$ provide 
\[
|R|\rho^{r}(\rho_1)^{k-r}\ge (p-f)\rho^{r}(\rho_1)^{k-r}\ .
\]
If $R$ is a set of nodes such that $|R|<p-f$, then  the following inequality holds:
\begin{equation}
\label{eqn:generalized}
|N^k (R)| > \frac{|R|}{p-f}\cdot \frac{71}{72}\, p \ .
\end{equation}
This is because in such a case we have that 
\[
|R| \rho^{r} (\rho_1)^{k-r} =  \frac{|R|}{p-f} (p-f) \rho^{r} (\rho_1)^{k-r}> \frac{|R|}{p-f} \cdot \frac{71p}{72}
\ ,
\]
by the inequality used in part b) of the construction of the overlay graph.


\Paragraph{Three properties of the overlay graphs.}

Now we discuss three properties of the overlay graphs $G(p,f)$, as defined above in the form~$L(p)^\ell$ for a suitable~$\ell$, that will be key for applications in algorithm design in subsequent sections.


\begin{lemma}
\label{lem:degree-power-of-Lp}

For any $f<p$,  the maximum degree of $G(p,f)$ is $\cO\bigl(\bigl(\frac{p}{p-f}\bigr)^{3.31}\bigr)$.
\end{lemma}

\begin{proof}
We estimate $r$ and $k-r$ as used in the specification of $G(p,f)=L(p)^\ell$, where $\ell=2k+1$.
It follows from part a) of the construction that
\begin{equation}
\label{eqn:aux1}
\log_\rho\frac{p}{50(p-f)} < r \le \log_\rho\frac{p}{50(p-f)} +1\ .
\end{equation}
Similarly, since the inequality
\[
(p-f)\rho^{r}\cdot \rho_1^{k-r-1}\le \frac{71p}{72}
\]
holds by part b) in the construction, and  the inequality
\[
\frac{1}{(p-f)\rho^r}<\frac{50}{p}
\]
holds by part a), we obtain that
\begin{equation}
\label{eqn:aux2}
\rho_1^{k-r-1}\le \frac{71p}{72} \cdot \frac{1}{(p-f)\rho^r} 
<\frac{71\cdot 50}{72}\ .
\end{equation}
Combining the estimates \eqref{eqn:aux1} and~\eqref{eqn:aux2} yields the following bound on~$k$:
\begin{eqnarray*}
k&=&r+(k-r)\\
&\le& \log_\rho\frac{p}{50(p-f)} +1 + \log_{\rho_1} \frac{71\cdot 50}{72} + 1 \\
&=&
\log_\rho \frac{p}{p-f} + \Bigl(2-\log_\rho 50 + \log_{\rho_1} \frac{71\cdot 50}{72}\Bigr)\\
&=& \log_\rho \frac{p}{p-f} + c \ ,
\end{eqnarray*}
where $c$ is a constant ($c= 2-\log_\rho 50 + \log_{\rho_1} \frac{71\cdot 50}{72}$). 
This leads to the following bound on $\Delta$, which is the maximum degree of~$L^{\ell}(p)$:
\begin{eqnarray*}
\Delta&\le& \Delta_0^{2k+1}\\
&\le& \Delta_0^{2\log_\rho \frac{p}{p-f} + 2c +1}\\
&=&\cO\bigl(\Delta_0^{2\log_\rho \frac{p}{p-f}}\bigr)\\
&=&
\cO\Bigl(\Bigl(\frac{p}{p-f}\Bigr)^{2\log_\rho \Delta_0}\Bigr) \ .
\end{eqnarray*}
For $\rho = 27/2$ and $\Delta_0=74$, the exponent $2\log_\rho \Delta_0$ in the bound on the maximum degree  is less than~$3.31$, by direct inspection.
\end{proof}

Let $f<p$ and let $R$ denote a set of nodes of the graph~$G(p,f)=L(p)^\ell$ of at least $p-f$ elements. 

\begin{definition}\label{def:P}
The \textbf{subgraph function $P(R)$} gives the subgraph of~$L(p)^\ell$ induced by the vertices in~$N^k (P'(N^k (R)))\cap R$, where $\ell=2k+1$ and $P'$ is the subgraph function from Fact~\ref{fact:upfal}.
\end{definition}


\begin{lemma}
\label{lem:subgraph-property}

For any $f<p$,  graph $G(p,f)$ has the subgraph property with $P(R)$ as the subgraph function.
\end{lemma}

\begin{proof}
We verify the four parts of the definition of the subgraph property.

Part~1)  states that $P(R)\subseteq R$: this follows directly from the definition of~$P(R)$.

Part 2) states that $|P(R)|\ge|R|/7$.
Suppose, to the contrary, that $|R\setminus P(R)|> 6|R|/7$. 
We proceed by making three preliminary observations.
The first one is that the following inequality holds:
\begin{equation}
\label{eqn:first-observation}
|N^k (R\setminus P(R))| > \frac{6}{7} \cdot \frac{71}{72} \, p \ .
\end{equation}
To show \eqref{eqn:first-observation}, we consider two cases.
If $|R\setminus P(R)|\ge p-f$ then estimate~\eqref{eqn:exponent} suffices.
Otherwise, when $|R\setminus P(R)| < p-f$, then we resort to~\eqref{eqn:generalized} and additionally verify that
\[
\frac{|R\setminus P(R)|}{p-f}\ge \frac{|R\setminus P(R)|}{|R|} > \frac{6}{7}\ .
\]
The second observation is that the following inequalities hold
\begin{equation}
\label{eqn:second-observation}
|P'(N^k(R))|\ge \frac{|N^k (R)|}{6} \ge \frac{1}{6}\cdot\frac{71}{72}\, p\ , 
\end{equation}
which follows from bound \eqref{eqn:exponent}  and Fact~\ref{fact:upfal}. 
Finally, observe that the inequality
\begin{equation}
\label{eqn:third-observation}
\frac{6}{7}\cdot\frac{71}{72} + \frac{1}{6}\cdot\frac{71}{72}>1 
\end{equation}
holds, which can be verified directly. 

The inequalities~\eqref{eqn:first-observation},  \eqref{eqn:second-observation} and~\eqref{eqn:third-observation} together imply that the intersection of~$N^k (R\setminus P(R))$ with $P'(N^k (R))$ is nonempty, because the sum of the sizes of sets is larger than~$p$, which is the total number of nodes.
Let~$v$ be a node that belongs to both $N^k (R\setminus P(R))$ and~$P'(N^k (R))$.
There is a path in~$L(p)$ of length at most~$k$ from~$v$ to some  $w$ in $R\setminus P(R)$. 
According to the specification $P(R)= N^k (P'(N^k (R)))\cap R$, this means that $w$ is in~$P(R)$ because the node~$v$ is in~$P'(N^k (R))$.
This results in a contradiction, since the sets~$R\setminus P(R)$ and~$P(R)$ are disjoint.
The contradiction yields that $|P(R)|\ge |R|/7$.

Part~3) states that the diameter of~$P(R)$ is at most $30\lg p +2$.
Let $v_1$ and~$v_2$ be two vertices in the subgraph~$P(R)$ of~$L(p)^\ell$. 

Suppose first that both $v_1$ and~$v_2$ are in~$P'(N^k (R))$.
By Fact~\ref{fact:upfal}, there is a path from node~$v_1$ to node~$v_2$  of length at most $30\lg p$ that traverses nodes in $P'(N^k (R))$ via edges in~$L(p)$. 
Let $\langle w_1=v_{1},w_2,\ldots ,w_m=v_2\rangle$ be the sequence of consecutive nodes on this path. 
We show that there is a path $\langle w'_1=v_{1}, w'_2,\ldots ,w'_m=v_2\rangle$ between $v_1$ and~$v_2$ in the subgraph~$P(R)$ of~$L(p)^\ell$ that has the following property: for $1\le i<m$, either $w'_{i+1}=w_{i+1}$ or $w'_{i+1}$ is a neighbor of~$w_{i+1}$ in~$L(p)^k$. 
The construction is by induction on the length of the path. 
The base of induction is for the path of length zero, which holds by the assumption $w'_1=v_{1}$.
Suppose that we have found an initial segment of this path up to a node $w'_i$, where $1\le i<m-1$. 
The inductive step is accomplished for each of the following four possible cases: 
\begin{itemize}
\Item[\rm a)] 
The case when $w_i=w'_i$ and~$w_{i+1}$ is in~$P(R)$:

Set $w'_{i+1}=w_{i+1}$. 
The nodes $w'_i$ and~$w'_{i+1}$ are neighbors in~$L(p)$  and so also in~$L(p)^\ell$.

\Item[\rm b)] 
The case when $w_i=w'_i$ and~$w_{i+1}$ is not in~$P(R)$:

Hence $w_{i+1}\in P'(N^k(R))\setminus P(R)$. 
Let $C$ denote $P'(N^k (R))\setminus R$. 
We have the following equalities:
\[
P'(N^k(R))\setminus P(R) = 
P'(N^k(R))\setminus (N^k(P'(N^k(R)))\cap R) =
P'(N^k(R))\setminus R = C \ ,
\] 
because 
\[
P'(N^k(R)) \cap R \subseteq P'(N^k(R))
\]
and so the inclusion
\[
P'(N^k(R)) \cap R \subseteq N^k(P'(N^k(R))
\] 
holds.
This means that $w_{i+1}\in C$.

Let $w'_{i+1}$ be any neighbor of~$w_{i+1}$ that is in the subgraph~$(N^k (C)\setminus C)\cap R$ of~$L(p)^k$, and hence in~$P(R)$; such a node exists since otherwise $N^k(w_{i+1})\cap R$ would be empty, which contradicts the fact that 
\[
w_{i+1}\in C\subseteq  P'(N^k(R))\subseteq N^k(R) \ .
\]
The nodes $w'_i$ and~$w'_{i+1}$ are neighbors in~$L(p)^\ell$, as their distance in~$L(p)$ is at most $1+k<\ell$; namely, the distance in~$L(p)$ between $w'_i$ and $w_{i+1}$ is $1$ and the distance in~$L(p)$ between $w_{i+1}$ and $w'_{i+1}$ is at most~$k$.

\Item[\rm c)] 
The case when $w_i\neq w'_i$ and~$w_{i+1}$ is in~$P(R)$:

Take $w_{i+1}$ to be $w'_{i+1}$.
The distance between $w'_i$ and~$w'_{i+1}=w_{{i+1}}$ in~$L(p)$  is $k+1<\ell$, since the node $w_i$ is a neighbor of~$w_{i+1}$ in~$L(p)$ and, by the inductive assumption, the node $w_i$ is a neighbor of~$w'_i$ in~$L(p)^k$.
Hence there is an edge between $w'_i$ and~$w'_{i+1}$ in~$L(p)^\ell$. 

\Item[\rm d)] 
The case when $w_i\neq w'_i$ and~$w_{i+1}$ is not in~$P(R)$:

Let $C$ denote $P'(N^k (R))\setminus R$. 
By the inductive assumption, the node~$w_i$ is a neighbor of~$w'_i$ in~$L(p)^k$ and $w_{i+1}\in C$.
Let $w'_{i+1}$ be any neighbor of~$w_{i+1}$ in the subgraph $(N^k (C)\setminus C)\cap R\subseteq P(R)$ of~$L(p)^k$; such a node exists since otherwise $N^k(w_{i+1})\cap R$ would be empty, which contradicts the fact that $w_{i+1}\in P'(N^k(R))\subseteq N^k(R)$.
The distance between $w'_i$ and~$w'_{i+1}$ in~$L(p)$ is at most $\ell=k+1+k$, since the distance between $w'_i$ and $w_i$ is at most $k$, the distance between $w_i$ and $w_{i+1}$ is~$1$, and the distance between $w_{i+1}$ and $w'_{i+1}$ is at most $k$. 
Therefore there is an edge between $w'_i$ and $w_i$ in~$L(p)^\ell$. 
\end{itemize}
Note that if a path $\langle v_1=w'_1,w'_2,\ldots ,w'_{m-1}\rangle$ has been built then the node~$w'_{m-1}$ is in $N^k (\{w_{m-1}\})$, by the inductive assumption. 
Hence the distance in~$L(p)$ between the nodes~$w'_{m-1}$ and~$v_2$ is at most $k+1$. 
This completes the inductive construction, and hence the case when both $v_1$ and~$v_2$ are in~$P'(N^k (R))$. 

The case when either $v_1$ or~$v_2$ is not in~$P'(N^k (R))$ follows from the observation that both these nodes have neighbors in~$P'(N^k (R))$ in~$L(p)^k$, by the inclusion $P(R)\subseteq N^k (P'(N^k (R)))$.
Therefore the distance between them in the subgraph~$P(R)$ of~$L(p)^\ell$ is at most $2+30\lg p$. 

Part~4)  states that $P(R_1)\subseteq P(R_2)$ for any induced subgraphs $R_1\subseteq R_2$ of~$L^\ell(p)$, each of size at least $p-f$.

Let $V_1$ be the set of nodes of~$R_1$ and $V_2$ the set of nodes of~$R_2$.
Clearly $N^k (V_1)\subseteq N^k (V_2)$.
Moreover, $|N^k (V_2)|\ge |N^k (V_1)|\ge 71p/72$ by~\eqref{eqn:exponent}.
It follows that $P'(N^k (V_1))\subseteq P'(N^k (V_2))$, by Fact~\ref{fact:upfal}, which implies 
\[
N^k(P'(N^k (V_1)))\subseteq N^k(P'(N^k (V_2))) 
\]
and next 
\[
N^k(P'(N^k (V_1))) \cap V_1\subseteq N^k(P'(N^k (V_2)))\cap V_2 \ .
\]
Since $P(R_1)$ is the subgraph of~$L^\ell(p)$ induced by the nodes in $N^k(P'(N^k (V_1))) \cap V_1$, and $P(R_2)$ is the subgraph of~$L^\ell(p)$ induced by the nodes in $N^k(P'(N^k (V_2))) \cap V_2$, the inclusion $P(R_1)\subseteq P(R_2)$ follows.
\end{proof} 

As a part of the statement of Lemma~\ref{lem:subgraph-property}, the diameter of a compact graph is at most $30 \lg p + 2$; we will use the notation~$g(p)$ as a shorthand for $30 \lg p + 2$.

We say that two disjoint sets of nodes $A$ and $B$ of a graph are \emph{connected by an edge} if there exist $v\in A$ and $w\in B$ such that $(v,w)$ is an edge.

\begin{lemma}
\label{lem:disjoint-are-connected}

Any two disjoint sets  of nodes of graph $G(p,f)=L(p)^\ell$, such that each contains at least $(p-f)/7$ nodes, are connected by an edge.
\end{lemma}

\begin{proof}
Let $A$ and $B$ be two disjoint sets of nodes such that each has at least $ (p-f)/7$ elements.
By the specification of graph $L(p)^\ell$,  it is sufficient to show that there is a path in $L(p)$, of length at most~$\ell$, from some node $v\in A$ to some node $w\in B$.
In what follows in this proof, edges and neighborhoods refer
to graph~$L(p)$. 

Recall that $\ell=2k+1$
and the expansion factor for sets of size at most $p=50$ is  $\rho = 27/2$.
We show that the intersection of $N^{k}(A)$ and $N^{k}(B)$ is nonempty.
Let $A'=N(A)$. 
The following inequalities hold:
\[
|A'|\ge\min\{ \rho |A|,\lfloor p/50\rfloor +1\}\ge\min\{p-f,\lfloor p/50 \rfloor+1\}
\ ,
\]
because $|A|\ge (p-f)/7$ and $\frac{1}{7}\cdot \frac{27}{2}>1$.
The inequality $|N^r(A')| > p/50$ holds by the specifications of $r$ and~$\rho$. 
The size of the set
\[
N^k(A)=N^{k-1}(A')=N^{k-r-1}(N^r(A'))
\] 
is  at least $(p-f)\rho^r (\rho_1)^{k-r-1}$, by the definitions of $k$ and~$\rho_1$.
The number $k$ was selected such that the following inequality holds
\[
(p-f)\rho^r (\rho_1)^{k-r}> \frac{71p}{72}
\ ,
\] 
for some integer~$r>0$.
This implies the following inequalities:
\begin{equation}
\label{eqn:bigger-than}
(p-f)\rho^r (\rho_1)^{k-r-1}> \frac{71p}{72\rho_1} > \frac{p}{2}
\ ,
\end{equation}
because  $2\cdot \frac{71}{72}>1.013=\rho_1$.
We conclude that set $N^k(A)$ has more than $p/2$ nodes. 

By the same argument, we have that $|N^k(B)|>p/2$. 
It follows that there exists a node $x\in N^k(A)\cap N^k(B)$.
This node $x$ is of distance at most $k$ from some node $v\in A$ and also of distance at most $k$ from some $w\in B$.
Therefore, there exists a path in $L(p)$ that starts at  $v\in A$, passes through $x$, and ends at~$w\in B$, whose total length is at most $2k<\ell$.
\end{proof}

\section{Algorithm Design}

\label{sec:design-algorithms}

In this section we present the design template for our algorithms formulated as  \emph{the generic algorithm}.
This algorithm design leaves the following open: 
1)~the assignment of tasks for processors to perform, and 
2)~the overlay graphs used for sending messages.
The generic algorithm solves the \DA\ problem against the unbounded adversary, as we show in Proposition~\ref{proposition:correctness-generic}.
In the subsequent sections we deal with specific algorithms as instantiations of the generic algorithm.
The purpose of the instantiations is to achieve suitable efficiency.

\subsection{The generic algorithm}

\label{sec:generic-algorithm}

The generic algorithm is  given in Figure~\ref{fig:algorithm}, and it contains the initialization followed by the \texttt{repeat} loop.
The activity of each processor consists of two stages.
In the first stage, the processor works to perform all tasks and to share its knowledge with other processors.
This is implemented by procedure \texttt{Main} given in Figure~\ref{fig:procedure-main}.
After the processor learns that there are no outstanding tasks, 
it switches to the second stage, where it propagates its knowledge to all processors.
This is implemented by procedure \texttt{Closing} given in Figure~\ref{fig:procedure-closing}.
The processor halts when it learns that every processor knows that there are no outstanding tasks.


\begin{figure}[t]
\rule{\textwidth}{0.75pt}

Generic Algorithm
\begin{center}
\begin{minipage}{\pagewidth}
\begin{description}
\item[{\it /* Initialization */}] ~ 

  \texttt{Tasks}$_{v} \leftarrow$ the sorted list of all  tasks 

  \texttt{Processors}$_{v}$,  \texttt{Busy}$_{v}  \leftarrow$ the sorted list of all  processors

 \texttt{Phase}$_{v} \leftarrow \texttt{Main}$ {\it /*  \texttt{Phase}$_{v}$ is set to point to procedure \texttt{Main} */}

 $\texttt{Done}_{v} \leftarrow \texttt{false}$

\item[{\it /* Algorithm code */}] ~

\texttt{repeat} call procedure  \texttt{Phase}$_{v}$ 
\end{description}
\end{minipage}
\FFF

\rule{\textwidth}{0.75pt}

\parbox{\captionwidth}{\caption{\label{fig:algorithm}
The generic algorithm code for processor~$v$; the algorithm starts with the initialization, 
followed by iterating either procedure \texttt{Main} or procedure \texttt{Closing}
based on the value of \texttt{Phase}$_{v}$.}}
\end{center}
\end{figure}


\begin{figure}[t]
\rule{\textwidth}{0.75pt}

Procedure \texttt{Main}

\begin{center}
\begin{minipage}{\pagewidth}
\begin{itemize}
\Item[\sc Round 1:] receive messages
\Item[\sc Round 2:] perform local computation:
\begin{innerlist}
\Item
update the private lists using the messages just received:
\begin{enumerate}
\item
remove processor identifier $x$ from $\texttt{Processors}_v$ if either $x$ is missing from some \texttt{Processors} list received,
or $x$ is a neighbor of~$v$ and no message from $x$ was received
\item
remove processor identifier $x$ from \texttt{Busy}$_v$ if either $x$ is missing from some \texttt{Busy} list received, 
or $x$ is not in \texttt{Processors}$_v$
\item
remove task identifier $y$ from \texttt{Tasks}$_v$ if $y$ is missing from some \texttt{Tasks} list received 
\end{enumerate}

\Item
if \texttt{Tasks}$_v$ is nonempty then 
\begin{enumerate}
\item
select a task from \texttt{Tasks}$_v$ by applying  \texttt{Se\-lec\-tion\_Rule}
\item
perform the selected task
\item
remove the selected task from \texttt{Tasks}$_v$
\end{enumerate}
else $\texttt{Done}_v \leftarrow \texttt{true}$

\Item
if $\texttt{Busy}_v \neq \texttt{Processors}_v$ then $\texttt{Done}_v \leftarrow \texttt{true}$ 

\Item
if \texttt{Stop} signal received then $\texttt{Done}_v \leftarrow \texttt{true}$

\Item 
if $\texttt{Done}_v = \texttt{true}$ then 
\begin{enumerate}
\item
remove $v$ from \texttt{Busy}$_v$
\item
set \texttt{Tasks}$_v$ to an empty list
\item
$\texttt{Phase}_v \leftarrow \texttt{Closing}$ {\it /*  \texttt{Phase}$_{v}$ is set to point to procedure \texttt{Closing} */}
\end{enumerate}   
\end{innerlist}
\B
\Item[\sc Round 3:] 
multicast a message containing \texttt{Tasks}$_v$, \texttt{Processors}$_v$, and \texttt{Busy}$_v$ to those neighbors in the overlay graph that are in \texttt{Processors}$_v$
\end{itemize}
\end{minipage}
\FFF

\rule{\textwidth}{0.75pt}

\parbox{\captionwidth}{
\caption{\label{fig:procedure-main}
Procedure \texttt{Main} of the generic algorithm; the code for processor~$v$.
The code is parameterized by the procedure \texttt{Se\-lec\-tion\_Rule} that is used
to select an item from a list.}}
\end{center}
\end{figure}

A processor is called \emph{active} if it neither halted nor crashed.
An active processor  is called \emph{busy} when it executes procedure \texttt{Main}, and so is still ``busy performing tasks.''
A processor executing procedure \texttt{Closing} is no longer busy performing tasks but it remains active.
An active processor that is not busy may halt for two reasons in the course of executing procedure \texttt{Closing}.
(1) The processor halts after it learns that there are no busy processors, 
meaning that all active processors know that there are no outstanding tasks.
(2) The processor may halt after it learns that there are ``too few" active processors 
located sufficiently close in the overlay communication network to cooperate with productively; 
this is a technical condition designed to avoid scenarios when isolated islands of processors remain active, thereby unnecessarily contributing to the work and message complexities.


\begin{figure}[t!]
\rule{\textwidth}{0.75pt}

Procedure \texttt{Closing}

\begin{center}
\begin{minipage}{\pagewidth}

\begin{itemize}
\Item[\sc Round 1:] receive messages
\Item[\sc Round 2:] perform local computation:
\begin{innerlist}
\Item
update the private lists using the messages just received:
\begin{enumerate}
\item
remove processor identifier $x$ from $\texttt{Processors}_v$ if either $x$ is missing 
from some  \texttt{Processors} list received, 
or $x$ is a neighbor of~$v$ and no message from~$x$ was received
\item
remove processor identifier $x$ from $\texttt{Busy}_v$ if either $x$ is missing from some 
\texttt{Busy} list received, 
or $x$ is not in $\texttt{Processors}_v$
\end{enumerate}

\Item
if $v$ does not consider itself compact then \texttt{halt} after this phase is over
\Item
if \texttt{Busy$_v$} is nonempty then 
\begin{enumerate}
\item
select a processor from \texttt{Busy$_v$} by applying \textsf{Selection\_Rule}
\item
set \textsf{Selected\_Processor$_v$} to the selected processor
\item 
remove \textsf{Selected\_Processor$_v$} from \texttt{Busy$_v$}
\end{enumerate}
else 
\begin{enumerate}
\item[4.]
set \textsf{Selected\_Processor$_v$} to $v$
\item[5.]
\texttt{halt} after this phase is over
\end{enumerate}
\end{innerlist}
\B
\Item[\sc Round 3:] send messages:
\begin{innerlist}
\Item
send \texttt{Stop} signal  to \textsf{Selected\_Pro\-cessor$_v$}

\Item
multicast a message containing \texttt{Tasks}$_v$, \texttt{Processors}$_v$, and \texttt{Busy}$_v$ to those neighbors in the overlay graph that are in \texttt{Processors$_v$} 
\end{innerlist}
\end{itemize}
\end{minipage}

\FFF
\rule{\textwidth}{0.75pt}

\parbox{\captionwidth}{\caption{\label{fig:procedure-closing}
Procedure \texttt{Closing} of the generic algorithm; the code for processor~$v$.
The property of a processor to consider itself compact is defined in terms of the subgraph of the overlay graph induced by the nodes in $\texttt{Processors}_v$.
The processor~$v$ halts after completing the full three rounds of the last phase it performs, and so in particular after sending all the messages in Round~3 of the phase.
The code is parameterized by the procedure \texttt{Se\-lec\-tion\_Rule} that is used
to select an item from a list.}}
\end{center}
\end{figure}


\Paragraph{Local states and information propagation.}

Each processor~$v$ maintains its local state consisting of three ordered lists and additional  variables.
The list \texttt{Tasks$_v$} contains the identifiers of tasks assumed to be outstanding.
The list \texttt{Processors$_v$} contains the identifiers of processors assumed to be active.
The list \texttt{Busy$_v$} contains the identifiers of processors assumed to be busy performing tasks. 
These lists are sorted in the order of the identifiers contained in them. 
The position of an item in a list is called this item's \emph{rank} in the list.
Items may be removed from lists; this affects the ranks of the remaining items.

During the initialization, processor~$v$ sets the list \texttt{Tasks}$_v$ to the sorted list of all task identifiers.
Both lists \texttt{Busy}$_v$ and \texttt{Processors}$_v$  are initialized to the sorted list of all processor identifiers.

Processors share their knowledge by sending their \texttt{Tasks}, \texttt{Processors}, and \texttt{Busy} lists to their neighbors in the current overlay graph. 
Each active processor~$v$ updates its lists \texttt{Tasks$_v$}, \texttt{Processors$_v$} and \texttt{Busy$_v$} 
after it receives messages containing such lists. 
If an item is not on the list received from some neighbor, then the processor removes the item from the corresponding private list. 
The processor removes from its  \texttt{Tasks}$_v$ list the tasks that it knows to be complete based on the received information.
Similarly, the processor removes from its  \texttt{Processors}$_v$ list any processors that crashed or halted,  and finally, it removes from its \texttt{Busy}$_v$ list the processors that know that there are no outstanding tasks.
If no message is received by $v$ from a neighboring processor~$u$, 
then $u$ is removed from the list \texttt{Processors$_v$}, as this indicates that $u$ is no longer active. 
In this way the changes in the lists are propagated through flooding broadcasts over the overlay graph.

The knowledge contained in the private lists may be out of date in a given round. 
In the case of the lists \texttt{Processors$_v$} and \texttt{Busy$_v$}, the sets of processor identifiers
on these lists may be the supersets of the actual sets of active and busy processors, respectively. 
This happens if some processors are no longer active or no longer busy, but processor~$v$ has not learned this yet. 
Similarly, the set of task identifiers on the list \texttt{Tasks$_v$} may be a superset of the actual set of outstanding tasks, 
since some tasks may have been performed, but processor~$v$ does not know this yet.

Each processor maintains the following additional variables.
The variable \texttt{Phase$_v$} points to the procedure to be performed by $v$ in the next iteration of the main loop;
this procedure is either \texttt{Main} or \texttt{Closing}.
Initially \texttt{Phase}$_v$ is set to point to \texttt{Main}.
The variable \textsf{Selected\_Processor$_v$} is used to store the identifier of the processor to whom the \texttt{Stop} signal is sent by $v$ in the current execution of the procedure \texttt{Closing}. 
The boolean variable \texttt{Done$_v$} facilitates switching processor~$v$ from executing \texttt{Main} to executing \texttt{Closing} in the next iteration.
The  variable \texttt{Done$_v$} is initially set to false, and it remains false until~$v$ learns that all tasks are complete.

\texttt{Done$_v$} is set to true in the following three situations:
\begin{enumerate}
\item[(1)] 
When the list \texttt{Tasks}$_v$ becomes empty: this means that all tasks are complete.
\item[(2)] 
When $v$ receives the \texttt{Stop} signal from some processor~$u$: this means that $u$ knows that all tasks are complete and so it switches to executing \texttt{Closing}.
\item[(3)]
When~$v$ determines that the lists $\texttt{Busy}_v$ and $\texttt{Processors}_v$ are different: 
this happens when there is an active processor in $\texttt{Processors}_v$ that is missing from $\texttt{Busy}_v$, meaning that this processor will no longer execute \texttt{Main} as it switched to executing \texttt{Closing}.
\end{enumerate}
Once \texttt{Done$_v$} is set to true, processor~$v$ removes its identifier~$v$ from \texttt{Busy}$_v$  so that the knowledge that $v$ is no longer busy is propagated via messages carrying copies of lists.

We refer to each iteration of the repeat loop, where either \texttt{Main} or \texttt{Closing} is invoked, as a \emph{phase}.
Correspondingly, we call a phase either \emph{main} or \emph{closing}, depending on which procedure is invoked.
Each phase consists of three \emph{rounds}: 1)~receiving messages,  2)~performing local computation, and 3)~multicasting messages.
The goal of main phases is to complete all tasks.
Pseudocode for procedure \texttt{Main}, and so for each main phase, is given in Figure~\ref{fig:procedure-main}.
A processor whose list \texttt{Tasks} becomes empty switches to executing procedure \texttt{Closing}.
Now the goal is to inform all active processors that all tasks are complete.
This is accomplished in a manner similar to that of performing tasks, in that informing a processor is now treated as a ``task.''
To perform a ``task'' of this kind, first an item from the \texttt{Busy} list is selected, next the \texttt{Stop} signal is sent to it, and finally the item is removed from the list.
A processor halts when its  \texttt{Busy} list  becomes empty.
Pseudocode for procedure \texttt{Closing}, and so for each closing phase, is given in Figure~\ref{fig:procedure-closing}.

When a processor needs to select an item from a list, it does so according to a procedure
that implements a selection rule denoted by the name \textsf{Selection\_Rule} in Figures~\ref{fig:procedure-main} and~\ref{fig:procedure-closing}.

We instantiate the generic algorithm by specifying the following:
\begin{enumerate}
\item[1)]
the overlay graphs governing sending messages to neighbors and what it means for a processor to consider itself compact, and 
\item[2)]
the rules used to select tasks to be performed in a round.
\end{enumerate}
We use the overlay graphs $G(p,f)$ as specified in Section~\ref{sec:graphs}, for any $f<p$.


\begin{lemma}
\label{lem:communication-from-work}

If $\cW$ is the work complexity of the generic algorithm that uses the overlay graph $G(p,f)$, for any $f < p$, then its communication complexity is $\cO\Bigl(\Bigl(\frac{p}{p-f}\Bigr)^{3.31}\, \cW\Bigr)$.
\end{lemma}

\begin{proof}
Let $\Delta$  be the maximum degree of~$G(p,f)$.
Since work is the number of rounds a processor is active, the algorithm sends up to~$\Delta$ messages for each unit of work.
It follows that the communication cost is $\cM=\cO(\Delta\cdot \cW )$. 
The maximum degree~$\Delta$ of the graph~$G(p,f)$ is $\cO\bigl(\bigl(\frac{p}{p-f}\bigr)^{3.31}\bigr)$,  by Lemma~\ref{lem:degree-power-of-Lp}.
\end{proof}


\Paragraph{Compactness.}

As stipulated in Section~\ref{sec:graphs}, a subgraph of the overlay graph is compact if its
size is at least~$(p-f)/7$ and its diameter  is at most~$g(p)=30\lg p +2$.
Let the \emph{range} of a processor~$v$ be the subgraph of the overlay graph that contains each non-faulty processor  whose distance from~$v$  in the overlay graph is at most~$g(p)$.
A processor  is said to be \emph{compact} if its range includes at least $(p-f)/7$ nodes.
Any processor~$v$ in a compact subgraph~$H$ is compact since $H$ is included in the range of~$v$.

The processors estimate distances to other processors in the overlay graph using their local views of this graph.
The distances may change during an execution due to crashes and halts, 
so it is possible for a graph to stop being compact at some point in time.
Each processor~$v$ computes the subgraph of the overlay graph $G(p,f)$ induced by the processors in \texttt{Processors}$_v$, that is, 
by the processors that $v$ still considers being active.
Next $v$ computes the distances from each node of~$G(p,f)$ to itself in this subgraph.
If the estimated size of the range of processor~$v$ is at least~$(p-f)/7$, then $v$ is said to \emph{consider itself compact}.
This property is used as a criterion for halting, see Figure~\ref{fig:procedure-closing}.
Observe that, if a processor~$v$ is compact, but sufficiently many processors in the range of~$v$ 
have halted, thus no longer being active, then $v$ may consider itself not compact at this point.


\Paragraph{Correctness of the generic algorithm.}

Recall that the algorithm starts with initialization, followed by the repeat loop which, depending on the value of $\texttt{Phase}_v$,  invokes either procedure \texttt{Main} or procedure \texttt{Closing} (see Figure~\ref{fig:algorithm}).


\begin{proposition} 
\label{proposition:correctness-generic}

Any instantiation of the generic algorithm solves the \DA\ problem against  the unbounded adversary.
\end{proposition}

\begin{proof} 
First we show that if a task is not in  \texttt{Tasks}$_v$ of processor~$v$ in some round, then the task is complete.
Note that processor~$v$ removes a task from \texttt{Tasks}$_v$  only if either $v$ performs the task itself 
(line 2.b.3 in Figure~\ref{fig:procedure-main}), or $v$ receives a copy of a \texttt{Tasks} list without this task in it 
(line 2.a.3 in Figure~\ref{fig:procedure-main}).
The proof is by induction on the length of the path traversed by a sequence of messages that 
brings the information that the task is not on the list \texttt{Tasks}$_u$ of some processor~$u$.

A processor may halt only while executing a closing phase.
In order to enter this phase, the private variable \texttt{Done} needs to be set to true
in the pseudocode of Figure~\ref{fig:procedure-main}.
We show that if  \texttt{Done}$_v$ is true at some $v$, then \texttt{Tasks}$_u$ is empty at some~$u$; this in turn implies that all tasks are complete.
Inspecting the code of \texttt{Main} in Figure~\ref{fig:procedure-main}, 
we see that there are only three ways for \texttt{Done}$_v$ to be set to true.
(1) When \texttt{Tasks}$_v$ becomes empty (line 2.b  in Figure~\ref{fig:procedure-main});
this provides the base of induction.
(2) When $v$ receives \texttt{Stop} signal (line 2.d), and 
(3) when some processor on \texttt{Processors} list is missing from the \texttt{Busy} list (line 2.c) because it was missing from some received \texttt{Busy} list. 
Processor~$v$ removes its identifier from \texttt{Busy}$_v$ 
(line 2.e.1 in Figure~\ref{fig:procedure-main}) and sends \texttt{Stop} signal (line 3.a in Figure~\ref{fig:procedure-closing})  only when its variable \texttt{Done}$_v$ is true.
This provides the inductive step, since each message that triggers setting \texttt{Done} to true increments by one the length of the path originated at a processor whose \texttt{Tasks} list became empty.

We need to show that the correctness and termination conditions are satisfied in any execution.
Every non-faulty processor performs a task in a main phase, therefore a processor performs this phase at most $t$ times.
Similarly, every non-faulty processor removes at least one processor from its \texttt{Busy} list in each iteration of the closing phase, 
except possibly in its last iteration (if processor~$v$ halts at line~2.b in Figure~\ref{fig:procedure-closing}, 
then it does not remove a processor from \texttt{Busy}$_v$), hence a processor performs this phase at most $p$ times.
Therefore each processor halts by round $t+p$, unless it crashes earlier.
There is at least one processor that never crashes according to the definition of the unbounded adversary.
This processor halts eventually.
If a processor~$v$ halts, it happens in the closing phase, so \texttt{Done}$_v$ is true.
This means that all the tasks are complete.
\end{proof}


\Paragraph{Task selection rules.}

What remains to specify to make the generic algorithm fully instantiated is the task selection rules.
We will consider four such rules.
The resulting instantiations are named as follows:
\begin{itemize}
\Item
\aBL: 
a constructive deterministic algorithm based on load balancing; 
\Item
\aRP: 
a constructive randomized algorithm using random permutations to select tasks; 
\Item
\aDP: 
a nonconstructive deterministic algorithm using permutations that are part of the code to select tasks; 
\Item
\aEP:
a hybrid deterministic algorithm that uses a modification of algorithm \aDP\ combined with algorithm \aDMY\ from~\cite{PriscoMY94}.
\end{itemize}
The purpose of considering the randomized algorithm \aRP\ is only to prove the existence of permutations that make \aDP\ efficient.

\subsection{Epochs}

\label{sec:epoch-core-processors}

We partition an execution of the generic algorithm into disjoint \emph{epochs}, denoted $\cE_i$,  for $i\ge 0$.
Epoch~$\cE_0$ denotes the initialization.
An epoch $\cE_i$, for $i\ge 1$, is defined to be a segment of $g(p)=30\lg p+2$ consecutive phases; the index~$i$ in~$\cE_i$ denotes the $i$th such segment of phases.
Epochs include the phases occurring after  termination.

We use $K_i$, for $i\ge 0$, to denote the subgraph of the overlay graph~$G(p,f)$ induced 
by the nodes that are non-faulty  at the beginning of the epoch~$\cE_i$.
Similarly, we use $G_i$, for $i\ge 0$, to denote the subgraph of the overlay graph~$G(p,f)$ induced 
by the nodes that are non-faulty through the end of the epoch~$\cE_i$.
This means that epoch~$\cE_i$ begins with $K_i$ as the set of non-faulty nodes and ends with~$G_i$ as the set of non-faulty nodes.
We have that $|K_i|\ge p-f$ and $|G_i|\ge p-f$ and that the equalities $G_i = K_{i+1}$ and the inclusions $G_i \subseteq K_i $ hold for $i\ge 0$.

Some epochs $\cE_i$ may have the property that among the processors that start the epoch as active
(these make the set~$K_i$) relatively many crash in the course of the epoch.
The effect is such that those processors that remain non-faulty through the end of the epoch 
(these constitute the set~$G_i$) make a small fraction of the original set~$K_i$.
Formally, an epoch $\cE_i$ is defined to be \emph{stormy} when $|G_{i}| < |K_i |/2$. 
An epoch~$\cE_i$ is \emph{calm} if it is not stormy, that is, when $|G_{i}| \ge |K_i |/2$.
In the analysis of work performance in the following sections, we are mostly concerned with calm epochs, because stormy epochs are taken care of by the following Lemma.


\begin{lemma}
\label{lem:stormy-epochs}

Stormy epochs contribute $\cO(p\log p)$ to work.
\end{lemma}

\begin{proof}
An epoch lasts for $g(p)=\cO(\log p)$ rounds.
Each processor contributes $\cO(\log p)$ to the number of the available processor steps in an epoch.
The sum, taken over all stormy epochs, of the numbers of processors beginning the epoch as non-faulty, is $\cO(p)$, because it is the sum of a geometrically decreasing sequence.
Therefore, the total work over stormy epochs can be estimated as the product of $\cO(p)$ and $\cO(\log p)$.
\end{proof}


\Paragraph{Core processors.}

We define a sequence of subgraphs $ H_i$ in~$G(p,f)$, for $i\ge 0$, by stipulating that $H_i=P(G_i)$, where $P$ is the subgraph function of the overlay graph~$G(p,f)$ from Definition~\ref{def:P} in Section~\ref{sec:graphs}.
Lemma~\ref{lem:subgraph-property} in that section states that the construction satisfies the requirements of Definition~\ref{def:subgraph-property} of subgraph property. 
Graph~$H_i$ is called \emph{the core graph for~$\cE_i$}, and the processors that belong to~$H_i$ are \emph{the core processors for~$\cE_i$}.
When a core processor halts, it remains core, as we do not consider halted processors to be faulty.

The graphs $K_i$, $G_i$ and $H_i$,  for $i\ge 0$, are uniquely determined by how crashes occur in an execution.
It follows by the definition of subgraph property that each $H_i$ is compact and that $H_{i+1}\subseteq H_{i}$ and $|H_i|\ge |G_i|/7\ge (p-f)/7$, for $i\ge 0$.
These properties guarantee that some processors remain core throughout the whole execution.

The diameter of~$H_i$, for $i\ge 0$, is at most~$g(p)$, which is the reason why epoch is defined to be a segment of~$g(p)$ consecutive phases.
This definition of epoch guarantees that any epoch is of a sufficient duration for all the core processors to propagate their knowledge among themselves by flooding across the overlay graph.
Here the knowledge of processors can be restricted to mean the contents of their lists, as this is what the algorithms uses.
When a processor $v\in H_i$, that is core for an epoch~$\cE_i$, halts in~$\cE_i$, then $v$ cannot participate in such flooding.
The next lemma helps to argue about flooding in such situations.


\begin{lemma}
\label{lem:core-compact-forever}

If a processor $v$ that is core for an epoch $\cE_i$ halts in that epoch, then, in the last phase before $v$ halts, each of the lists \texttt{Busy}$_v$ and \texttt{Tasks}$_v$ is empty.
\end{lemma}

\begin{proof}
First we argue about the list \texttt{Busy}$_v$.
A processor can halt only while executing procedure \texttt{Closing}.
By the pseudocode of this procedure  in Figure~\ref{fig:procedure-closing}, there are two possible triggers to halt: one is when the list \texttt{Busy} is empty and another is when the processor does not consider itself compact.
We argue that the list \texttt{Busy}$_v$ is empty in the phase just after which $v$ halts by induction on this phase's number.

First, let us consider the base of induction.
Any processor in~$H_i$ is compact in~$\cE_i$, because $H_i$ is included in the $v$'s range.
Therefore a processor in $H_i$ considers itself compact as long as no processor in~$H_i$ has halted yet.
The base of induction holds because when a processor $v$ halts and no other core processor halted before then $v$ considers itself compact and so its list \texttt{Busy}$_v$ is empty.

Next consider the inductive step. 
Let $v$ halt in a phase such that any other core processor, say, $x$ that halted before $v$ had its list \texttt{Busy}$_x$ empty.
Observe that this empty list \texttt{Busy}$_x$ got forwarded to the $x$'s neighbors just before $x$ halted.
This is by the pseudocode in Figure~\ref{fig:procedure-closing}, which specifies that a processor halts after completing a phase, and so after sending all the messages in that phase. 
An alternative for the list \texttt{Busy}$_v$ to be empty, as a trigger making $v$ halt, is that $v$ does not consider itself compact.
Let us consider this possibility to show that it will not occur.
For~$v$ not to consider itself compact requires removing sufficiently many \emph{core} processors from its list \texttt{Processors}$_v$.
A processor~$w$ is removed from the lists \texttt{Processors} of its neighbors for the first time after a phase with no communication from~$w$, and later this information is reflected in the contents of the forwarded lists \texttt{Processors}.
When $w$ is core that halted before $v$ then the empty list \texttt{Busy}$_w$ was received by the $w$'s neighbors before $w$ stopped communicating with them, as we already argued. 
Therefore the information about the empty \texttt{Busy} lists of the core processors that halted before~$v$ reaches~$v$ before the information that would trigger removing these processors from the list \texttt{Processors}$_v$.
When $v$ obtains the first such a message with an empty  \texttt{Busy} list, then $v$ makes its  list \texttt{Busy}$_v$ empty in the same phase.
The list \texttt{Busy}$_v$ stays empty, once set to be such, until $v$ halts, because entries from this list may be deleted but are never added after initialization.
This completes the proof of the inductive step.

Next we reason about the list \texttt{Tasks}$_v$.
A transition from \texttt{Main} to \texttt{Closing} occurs after line 2.e.3 of Figure~\ref{fig:procedure-main} is executed.
This line is triggered by the condition $\texttt{Done}_v = \texttt{true}$ in line 2.e
that also causes list \texttt{Tasks} to be emptied in line 2.e.2.
It follows that any processor executing procedure \texttt{Closing} has its list \texttt{Tasks} empty.
A processor can halt only while executing procedure \texttt{Closing}, so its list \texttt{Tasks} is empty when it halts.
\end{proof}

The properties of epochs stated as Lemmas~\ref{lem:one-epoch-switch} and~\ref{lem:inequality-ui-by-si} below  are shown by referring to Lemma~\ref{lem:core-compact-forever}.


\begin{lemma}
\label{lem:one-epoch-switch}

If a processor that is core in epoch~$\cE_i$ switches to closing phases by the beginning of~$\cE_{i}$, then every  core processor in~$\cE_i$ switches to closing phases by the end of epoch~$\cE_{i}$, for $i\ge 1$.
\end{lemma}

\begin{proof}
Let $v$ be a processor in~$H_{i}$ that  switches to closing phases prior to the beginning of~$\cE_{i}$. 
Processor~$v$  removes its identifier from \texttt{Busy}$_v$ in the round when it switches to closing, per line~2.e.1 in Figure~\ref{fig:procedure-main}.
Next $v$ sends its lists to the neighbors, resulting in each of them removing $v$ from their respective \texttt{Busy} lists,  per line 2.a.2  in Figures~\ref{fig:procedure-main} and line 2.a.2 in Figure~\ref{fig:procedure-closing}.
The fact that $v$ is missing in a \texttt{Busy} list propagates through all core  processors  by flooding within one epoch.
Node~$v$ stays in the lists \texttt{Processors} throughout epoch $\cE_i$, since $v$ belongs to~$G_i$.
Suppose a processor~$u$ receives a copy of \texttt{Busy} list without~$v$ in it in~$\cE_i$  and $u$ is still executing main phases.
The lists $\texttt{Busy}_u$ and $\texttt{Processors}_u$ are compared in line~2.c  in~Figure~\ref{fig:procedure-main}.
The fact that $\texttt{Busy}_u \neq \texttt{Processors}_u$ causes $u$ to set \texttt{Done}$_u$ to~\texttt{true} and next to switch to closing phases in line~2.e.3 in~Figure~\ref{fig:procedure-main}.

If such a flooding chain, as referred to above, fails, it is because some core processor~$w$ in the chain has already halted.
By Lemma~\ref{lem:core-compact-forever}, when a core processor~$w$ halts then its list \texttt{Busy}$_w$ is empty.
Copies of this empty list \texttt{Busy}$_w$ are forwarded to the $w$'s neighbors in the phase just before $w$ halts, by the pseudocode in Figure~\ref{fig:procedure-closing}.
Such forwarding of the empty \texttt{Busy} list to the neighbors concludes flooding from the $w$'s perspective, because this is the ultimate information to be forwarded in a flooding chain, as  the lists  \texttt{Busy} may shrink but they never grow. 
It is sufficient to forward an empty \texttt{Busy} list to the neighbors only once. 
This means that a halted processor does not disrupt a flooding chain of copies of \texttt{Busy} list.
\end{proof}


\Paragraph{Counting outstanding tasks.}

We make use of the following notation for sets: 
 $T_{v,i}$ is a set that contains all items in \texttt{Tasks$_v$} at the end of epoch~$\cE_i$, for $i\ge 0$.
When an epoch's index~$i$ is understood from the context, then we may simply write~$T_v$ for~$T_{v,i}$.
When  considering how much progress has been made by the core processors,
we use the following shorthand notation: $U_i$~stands for the set $\bigcup_{v\in H_i} T_{v,i}$, and $S_i$ denotes  $\bigcap_{v\in H_i}T_{v,i}$, for $i\ge 0$.
We use the following notations for numbers denoting sizes of sets: $u_i=|U_i|$ and $s_i=|S_i|$, for $i\ge 0$. 
For convenience, we let $s_{-1}=u_{-1}=t$.

Epoch $\cE_i$ begins with~$u_{i-1}$ tasks such that each of them belongs to at least one \texttt{Tasks} list of a core processor, and it ends with~$u_{i}$  tasks such that each of them belongs to at least one \texttt{Tasks} list of a core processor.
Similarly, epoch $\cE_i$ begins with the \texttt{Tasks} list of the core processors such that each list has at least~$s_{i-1}$ tasks in it, and it ends with \texttt{Tasks} list of the core processors such that each list has at least~$s_{i}$ tasks in it.


\begin{lemma}
\label{lem:inequality-ui-by-si}

The inequality $u_{i}\le s_{i-1}$ holds, for any $i\ge 0$.
\end{lemma}

\begin{proof}
We argue that $U_{i}\subseteq S_{i-1}$.
It is sufficient to show that $T_{v,i}\subseteq T_{w,i-1}$, for any two processors~$v$ and $w$ in~$H_{i}$.
If a task is not in  \texttt{Tasks}$_w$ at the beginning of~$\cE_{i}$, then this information will propagate through all of  $H_{i}$ through flooding by the end of this epoch;
this results in processor~$v$ removing the task from its list \texttt{Tasks}$_v$ in the epoch, 
unless the task was removed earlier.

If such flooding fails, it is because some core processor~$w$ in the chain have already halted.
By Lemma~\ref{lem:core-compact-forever}, when a core processor~$w$ halts then its list \texttt{Task}$_w$ is empty.
Copies of this empty list \texttt{Task}$_w$ are forwarded to the $w$'s neighbors in the phase just before $w$ halts, by the pseudocode in Figure~\ref{fig:procedure-closing}.
Such forwarding of the empty \texttt{Task} list to the neighbors concludes flooding from the $w$'s perspective, because this is the ultimate information to be forwarded in a flooding chain, since  the lists \texttt{Task} may only become smaller
(they never grow). 
It is sufficient to forward an empty \texttt{Tasks} list to the neighbors only once. 
This means that a halted processor does not disrupt a flooding chain of copies of \texttt{Tasks} lists.
\end{proof}

The sequences $\langle u_i\rangle_{i\ge 0}$ and $\langle s_i\rangle_{i\ge 0}$ are nonincreasing, and $s_i\le u_i$, for $i\ge 0$, directly from their definitions.
These properties combined with Lemma~\ref{lem:inequality-ui-by-si} allow us to use either one of the sequences $\langle u_i\rangle_{i\ge 0}$ and $\langle s_i\rangle_{i\ge 0}$ when measuring progress in the number of tasks completed.


\Paragraph{Main, mixed, and closing epochs.}

Recall that phases of an execution of an instantiation of the generic algorithm are partitioned into main and closing phases.
This partitioning of phases is independent for each processor, and is determined by which corresponding procedure the processor performs in a phase.
In parallel to having this partitioning of phases, we also partition the epochs of an execution into the related categories, for two of which we use the same terms.
Epochs are partitioned into the following groups, which make three  disjoint contiguous segments of an execution: 
\begin{enumerate}
\item[(1)] 
\emph{main epochs} are those  in which all core processors are busy, 
\item[(2)] 
 \emph{mixed epochs} are those  in which some core processors are busy while some are not, and 
\item[(3)] 
 \emph{closing epochs} are those  in which no core processors are busy.
\end{enumerate}
For calm main epochs~$\cE_i$, it is sufficient to estimate total work by the contribution from core processors, as is shown in the following lemma.


\begin{lemma}
\label{lem:core-processors-suffice}

If $w_i$ is the amount of work accrued by the core processors during a calm main epoch~$\cE_i$,  then the total work accrued during~$\cE_i$  is~$\cO(w_i)$.
\end{lemma}

\begin{proof}
No processor that is core in epoch~$\cE_i$ crashes in~$\cE_i$, as $H_i$ is determined by $G_i$, which consists of the nodes that do not crash by the end of~$\cE_i$.
No processor that is core in epoch~$\cE_i$ halts in~$\cE_i$, by the pseudocode in Figure~\ref{fig:procedure-main}.
Therefore, every node in graph $H_i$  performs work throughout the whole epoch~$\cE_i$.

Graph~$H_i$ is defined as the compact subgraph of~$G_i$ of the form $H_i=P(G_i)$,  where~$P$ is a subgraph function by Lemma~\ref{lem:subgraph-property}, so that $|G_i|\le 7 |H_i|$.
We combine this fact with the inequality $|K_{i}| \le 2 |G_i |$, given by the assumption, to obtain $|K_i| \le 14 |H_i |$.

This implies that the total work accrued during~$\cE_i$  is at most $14 w_i$.
\end{proof}

Next, we formulate a lemma which is an analogue of Lemma~\ref{lem:core-processors-suffice} for closing phases.
The difference we need to handle is that processors may halt voluntarily in closing epochs.

\begin{lemma}
\label{lem:core-closing-work}

Suppose that $\cE_i$ and~$\cE_{i+1}$ are both calm closing epochs.
Let $w_i$ be the amount of work accrued by the core processors during~$\cE_i$.
If every core processor that empties its \texttt{Busy} list during~$\cE_i$ does
not stay core for~$\cE_{i+1}$,  
then the total work accrued during~$\cE_i$  is~$\cO(w_i)$.
\end{lemma}

\begin{proof}
If no core processors empty their \texttt{Busy} lists during~$\cE_i$, then the same argument as in the proof of Lemma~\ref{lem:core-processors-suffice} applies, because no processor halts.
Otherwise, let some core  processors empty their \texttt{Busy} lists during~$\cE_i$.
By the assumption, all these processors do not stay core for the next epoch~$\cE_{i+1}$.
Therefore, the  processors that are core in~$\cE_{i+1}$ do not halt during~$\cE_i$, and so contribute to work through the end of~$\cE_i$. 
This property allows to compare the total work in epoch~$\cE_i$  to the work performed in this epoch by  processors that are core in the next epoch~$\cE_{i+1}$.

The number  of core processors that do stay core for~$\cE_{i+1}$ is  
\[
|H_{i+1}| \ge |G_{i+1}|/ 7 \ge |K_{i+1}|/14 = |G_i|/14
\ ,
\]
as epoch~$\cE_{i+1}$ is calm.
It follows that $|G_i|\le 14 |H_{i+1}|\le 14 |H_{i}|$.
Epoch $\cE_i$ is calm, in that the inequality $|K_i|\le 2 |G_i|$ holds, and so $|K_i|\le 28 |H_i|$.
Therefore the total work accrued during~$\cE_i$  is at most $28 w_i$.
\end{proof}

Our next goal is to estimate work in closing epochs by relating it to work accrued in main epochs.


\begin{lemma}
\label{lem:closing-is-as-efficient-as-main}

Suppose that an instantiation of the generic algorithm is such that if $p$ processors aim to perform $p$ tasks, then this is accomplished with $B(p)$ work by core processors in main epochs, for some function~$B(p)$.
Then the total work in closing epochs of such an instantiation with $p$ processors is $\cO(B(p)+ p\log p)$.
\end{lemma}

\begin{proof}
Stormy epochs contribute $\cO(p \log p)$ to work, by Lemma~\ref{lem:stormy-epochs}.
We include this amount as a component of the bound we seek to justify, and in what follows consider calm epochs only.
The following notation is used, when $\cE_i$ is a calm epoch, then $\cE_i'$, $\cE_i''$ and $\cE_i'''$ are the three immediately following calm epochs. 

We break calm closing epochs into two categories and consider them one by one.

The first category consists of these closing epochs~$\cE_i$ during which  every core processor that empties its \texttt{Bus}y list during $\cE_i$ does not stay core for the next epoch~$\cE_{i}'$. 
\texttt{Busy} lists in procedure \texttt{Closing} are analogous to \texttt{Tasks} lists in procedure \texttt{Main}.
Observe that sending a \texttt{Stop} signal to a processor~$u$ in a closing epoch can be interpreted as ``performing task~$u$.''
This follows by examining the pseudocodes in Figures~\ref{fig:procedure-main} and~\ref{fig:procedure-closing}: first performing a task and then removing this task from the list \texttt{Tasks} in procedure \texttt{Main} corresponds to first sending a 
\texttt{Stop} signal to a processor and then removing this processor from the list \texttt{Busy} in procedure \texttt{Closing}.
It follows that the activity of the system during such closing epochs can be interpreted  as having at most~$p$ processors perform at most $p$ tasks during main epochs.
The core processors perform at most~$B(p)$ work during these epochs~$\cE_i$, by definition of~$B(p)$.
The amount of work performed by all processors during these epochs is~$\cO(B(p))$, by  Lemma~\ref{lem:core-closing-work}.

The second category of calm closing epochs comprises epochs that start from the first epoch~$\cE_k$ during which some core processor~$v$ empties its  \texttt{Busy} list and stays core for the epoch~$\cE_{k}'$.
Processor~$v$ knows, from the emptiness of its  \texttt{Busy} list,  
that no processors are busy performing tasks, so it gets ready to halt.
But before stopping,
$v$ sends a ``time-to-halt'' notification to all its neighbors.
More precisely, processor~$v$ sends copies of empty list \texttt{Busy}$_v$ to its neighbors during this phase, by the pseudocode in Figure~\ref{fig:procedure-closing}.
We argue next that all processors core for $\cE_{k}'$ halt by the end of~$\cE_{k}'$.

The argument we employ is similar to the one used in the proof of Lemma~\ref{lem:one-epoch-switch}.
The information about empty  \texttt{Busy} lists is propagated by flooding.
Namely,  a processor that is about to halt first forwards this information to its neighbors.
The duration of an epoch is long enough for this  information to propagate successfully among the core processors,  if the flooding starts by a core processor in the first phase of an epoch.
If processor $v$ did not halt by the beginning of epoch $\cE_{k}'$, but its \texttt{Busy} list was empty at the beginning of epoch $\cE_{k}'$,  then $v$ would start disseminating
a ``time-to-halt'' notification  in the first phase of $\cE_{k}'$, and this notification would reach all processors core for~$\cE_{k}'$ during~$\cE_{k}'$.
If $v$ halts before the first phase of epoch~$\cE_{k}'$,  then $v$ has already sent its ``time-to-halt'' notification at the time when epoch $\cE_{k}'$ begins, so the effect of flooding is the same, in that by the end of epoch~$\cE_{k}'$ all processors core for~$\cE_{k}'$ halt.

Next we consider processors that are not core for epoch~$\cE_{k}'$.
Each of these processors either stays compact throughout epoch~$\cE_{k}''$ or it does not.
The case of a processor that stops being compact at some point in epoch~$\cE_{k}''$ is straightforward, because such a processor also stops considering itself compact by the end of~$\cE_{k}'''$, and then it halts. 

Let $w$ be a processor $w$ that is not core for epoch~$\cE_{k}'$ but it 
 compact throughout~$\cE_{k}''$.
If $w$'s range~$R$ at the end of~$\cE_{k}''$, includes a processor $z_1$ 
that is core for~$\cE_{k}'$,  then $w$ receives  ``time-to-halt'' notification from~$z_1$ by the end of~$\cE_{k}'$, and so $w$ halts by the end of~$\cE_{k}''$.
So, suppose that $R$ does not include a processor core for~$\cE_{k}'$.
The set~$R$ and the set of core processors for~$\cE_{k}'$ each contains at least $(p-f)/7$ nodes, by the definitions of compactness and core processors.
Since these two sets are disjoint, Lemma~\ref{lem:disjoint-are-connected} implies
that there exists a processor~$z_2$ in~$R$ that is connected by an edge to some processor core for~$\cE_{k}'$.
A processor core for~$\cE_{k}'$ sends a ``time-to-halt'' notification to processor~$z_2$ during~$\cE_{k}'$.
Processor~$z_2$ stays operational through~$\cE_{k}''$, 
so it receives this notification and forwards it to its neighbors.
Finally, processor~$w$ receives a  ``time-to-halt'' notification from~$z_2$  by the end of~$\cE_{k}''$, and so it halts during ~$\cE_{k}'''$ at the latest.

It follows that all processors halt by the end of epoch~$\cE_{k}'''$.
These four epochs $\cE_k$, $\cE_{k}'$, $\cE_{k}''$  and~$\cE_{k}'''$ contribute 
a total $\cO(p\log p)$ to the closing work.
\end{proof}

\section{A Constructive Algorithm}

\label{sec:constructive}

We obtain a specific algorithm by instantiating the generic algorithm with a task-selection rule 
that assigns tasks to processors with the aim of balancing the load among the processors.
The parameter~$f<p$ determines the overlay graph~$G(p,f)$ used for communication, as defined in Section~\ref{sec:graphs} and discussed in Section~\ref{sec:design-algorithms}.
We call the resulting constructive algorithm \aBL. 

The precise selection rule employed in algorithm \aBL\ is as follows.
Consider processor~$v$ and let  $k>0$ be the number of items in \texttt{Tasks}$_v$ list
in a certain round, that is $k=|\texttt{Tasks}_v|$.
We let~$r(v)$ stand for $\bigl\lceil \frac{k \, v}{p} \bigr\rceil$.
The positions on the list are numbered starting with~$1$.
Processor~$v$ selects the item at position $r(v)$ in its list of tasks to perform in this round; 
this selection is well-defined since $k\cdot \frac{v}{p}\le k$.
Number $r(v)$ is the largest integer~$j$ satisfying the inequalities 
\begin{equation}
\label{eqn:ranks}
k\cdot \frac{v}{p}-1< j <k\cdot\frac{v}{p}+1\ .
\end{equation}
Observe  that $v$ as a number is the rank of~$v$ on list \texttt{Processors$_v$} after the initialization.
As the execution proceeds, the rank of~$v$ in this list may change due to crashes,  but this does not affect the choices of tasks by~$v$, as $v$ does not use the size of \texttt{Processors$_v$} when deciding which task to perform.

Note also that the distances between the positions of tasks selected by different processors in an epoch need to be at least~$p g(p)$ if tasks are not to be duplicated in this epoch.
For this to be possible during the initial epochs, $t$ needs to be at least~$p^2 g(p)$, when the position of a selected task determined as $r(v)=\bigl\lceil \frac{k \, v}{p} \bigr\rceil$, which is equivalent to $\sqrt{t}\ge p \sqrt{g(p)}$.
In view of this, we may expect expressions involving $\sqrt{t}$ to appear in a bound on work of algorithm \aBL{} (this is indeed that case, as we state later in Theorem~\ref{thm:work-balance-load};
Lemma~\ref{lem:no-overlap-in-tasks} below covers the epochs during which the work performed by core processors is not duplicated and so such epochs contribute $\cO(t)$ to a bound on work).
In the other case, when $\sqrt{t} < p \sqrt{g(p)}$, duplication of work can occur in any epoch.
In such a situation, computing work is more involved as progress during an epoch in performing tasks depends on how many crashes occurred during the epoch and how much progress has been made in the previous epoch.
(Later in this section we address the involved cases leading to Theorem~\ref{thm:work-balance-load}.)

Next we formally analyze the algorithm's performance.
Recall the following notation: $U_i=\bigcup_{v\in H_i} T_{v,i}$ and $u_i=|U_i|$, for $i\ge 0$, and similarly $S_i=\bigcap_{v\in H_i}T_{v,i}$ and $s_i=|S_i|$, for $i\ge 0$.


\begin{lemma}
\label{lem:no-overlap-in-tasks}

If $u_{i-1}\ge 11p^2g(p)$ then $u_{i-1}-u_{i}\ge |H_{i}|g(p)$, for $i\ge 1$. 
\end{lemma}

\begin{proof}
We consider two cases, based on the value of $u_{i-1}-s_{i-1}$.

\noindent
Case 1: 
\begin{equation}
\label{eq-3}
u_{i-1}-s_{i-1}\ge pg(p)
. 
\end{equation}
In this case, we combine \eqref{eq-3} with Lemma~\ref{lem:inequality-ui-by-si} to obtain the estimates
\[
u_{i-1}-u_{i} \ge  u_{i-1}-s_{i-1} \ge pg(p) \ge |H_{i}|g(p)\ . 
\]
Case 2: $u_{i-1}-s_{i-1}<pg(p)$. 
For any processor~$v$, we calculate the decrease in the size of the list \texttt{Tasks$_v$} during  epoch~$\cE_{i}$.
At most $u_{i-1}-s_{i-1}<pg(p)$ tasks can be removed from this list due to discrepancy of the knowledge of processor~$v$ and the respective knowledge by other processors in the range of~$v$  about the outstanding tasks.
Clearly, at most $g(p)$ tasks are performed by processor~$v$ during  epoch~$\cE_{i}$.  

We want to show that during  epoch~$\cE_{i}$ each task is performed by at most one processor, 
as this property implies that the number of outstanding tasks decreases by at least $|H_{i}|g(p)$ during epoch~$\cE_{i}$.
To this end, consider the ranks of tasks in the ordered set $U_{i-1}=\bigcup_{v\in H_{i-1}}T_{v,i-1}$. 
In the first phase of  epoch~$\cE_{i}$, processor~$v$ performs a task whose rank 
in \texttt{Tasks$_v$} is~$r(v)$; let $w(v)$ be the rank of this task in~$U_{i-1}$. 
Consider any other task performed by~$v$ during~$\cE_i$ and let $w'(v)$ be its rank in~$U_{i-1}$. 
List \texttt{Tasks$_v$} is a dynamic data structure and the ranks of its elements may change during~$\cE_i$, 
while the rank of a task with respect to~$U_{i-1}$ stays the same during~$\cE_i$. 
Let $T_{v}$ denote the set~$T_{v,i-1}$.
We can combine~\eqref{eqn:ranks}, for $k=|T_v|$, the number of elements in \texttt{Tasks}$_v$ 
at the beginning of~$\cE_i$, with the fact that \eqref{eq-3} does not hold to obtain that
\begin{equation}
\label{eqn:w}
|T_v|\cdot\frac{v}{p}-1
<w(v)<
(|T_v|+pg(p)) \cdot \frac{v}{p}+1\ .
\end{equation}
The list \texttt{Tasks$_v$} may shed at most $2pg(p)$ items during  epoch~$\cE_i$; this implies
\begin{equation}
\label{eqn:wprime}
(|T_v|-2pg(p)) \cdot \frac{v}{p}-1
<w'(v)<
(|T_v|+pg(p)) \cdot \frac{v}{p}+1\ . 
\end{equation}
By comparing the lower and upper boundaries of the ranges~\eqref{eqn:w} and~\eqref{eqn:wprime}, we obtain that
\begin{equation}
\label{eq-4}
|w'(v)-w(v)| < 3pg(p)+2 \ . 
\end{equation}
For each processor~$v$, we define $B_v\subseteq U_{i-1}$ to be the set of tasks whose ranks differ from~$w(v)$ by at most $3pg(p)+2$.
It follows from \eqref{eq-4} that all  tasks performed by~$v$ during~$\cE_i$ are in~$B_v$.

Consider two different processors ${v_1}$ and~${v_2}$, where $v_1>v_2$. 
We show that their corresponding sets $B_{v_1}$ and~$B_{v_2}$ are disjoint. 
We first estimate the distance in~$U_{i-1}$ of the tasks they both perform in the first phase of~$\cE_i$. 
We use the following inequalities 
\begin{equation}
\label{eqn:aa}
r(v_1)>|T_{v_1}|\cdot \frac{v_1}{p}-1
{\rm ~ ~ ~ ~ ~ ~ and ~ ~ ~ ~ ~ ~}
r(v_2)<|T_{v_2}|\cdot\frac{v_2}{p}+1 
\end{equation}
to estimate the ranks in the relevant \texttt{Tasks} lists. 
We then combine \eqref{eqn:aa} with the assumption that~\eqref{eq-3} does not hold to obtain the following estimates of the ranks in~$U_{i-1}$: 
\[
w(v_1)>|T_{v_1}|\cdot\frac{v_1}{p}-1
{\rm ~ ~ ~ ~ ~ ~ and ~ ~ ~ ~ ~ ~}
w(v_2) \leq r(v_2) + u_{i-1} - s_{i-1} < |T_{v_2}|\cdot\frac{v_2}{p}+1+pg(p)\ . 
\]
This in turn leads to the following estimate:
\begin{eqnarray*}
w(v_1)-w(v_2)
&>& 
|T_{v_1}|\cdot\frac{v_1}{p}-1-|T_{v_2}|\cdot\frac{v_2}{p}-1-pg(p) \\
&\ge & 
(|T_{v_2}|-pg(p))\cdot\frac{v_1}{p}-|T_{v_2}|\cdot\frac{v_2}{p}-2-pg(p)\\
& = &
|T_{v_2}|\cdot\frac{v_1-v_2}{p}-(v_1+p)g(p)-2 \\
&\ge &
(u_{i-1}-pg(p))\cdot\frac{v_1-v_2}{p}-(v_1+p)g(p)-2 \\
&\ge & 
10pg(p)-2pg(p)-2  \\
& > & 
7pg(p)
\ . 
\end{eqnarray*}
Here we used the assumption that $u_{i-1}\ge 11p^2g(p)$ and the fact that $p \geq v_1 > v_2$. 
Notice that $7 p g(p) > 2 [3pg(p)+2]$; together with~\eqref{eq-4} this shows that $B_{v_1}$ and~$B_{v_2}$ are disjoint.
\end{proof}


\begin{lemma}
\label{lem:x-as-parameter}

Let $\cE_{k}$ be an  epoch and consider another epoch $\cE_i$ such that $i >k$. 
If all processors in~$H_{i}$ are in their main phases  during~$\cE_i$ and $x$ is a number such that $u_{i-1} \ge x \geq 7$, then $u_{i-1}-u_{i}$ is either at least $\min\bigl\{\frac{x}{14p}\;, 1\bigr\} |H_{i}|$  or at least $\frac{u_{i-1}}{x}$.   
\end{lemma}

\begin{proof}
Suppose that the inequality
\begin{equation}
\label{eqn:u(i-1)-ui}
u_{i-1}-u_{i}<\frac{u_{i-1}}{x} 
\end{equation}
holds.
We need to show that $u_{i-1}-u_{i}$ is either at least $\frac{x}{14p}|H_{i}|$ or at least~$|H_{i}|$. 

Let ${v_1}$ and~${v_2}$ be the identifiers of some two processors in~$H_{i}$ that perform the same task in the first phase of~$\cE_i$, where $v_1 > v_2$; let $r(v_1)$ and~$r(v_2)$ be the ranks of this task in \texttt{Tasks$_{v_1}$} and \texttt{Tasks$_{v_2}$}, respectively, in this phase. 
We first show that \eqref{eqn:u(i-1)-ui} implies 
\begin{equation}
\label{eqn:2-star}
v_1-v_2<\frac{7p}{x} \ .
\end{equation}
Lemma~\ref{lem:inequality-ui-by-si} gives the estimate $u_{i-1}-s_{i-1} \le u_{i-1}-u_{i}$.
We  combine this with \eqref{eqn:u(i-1)-ui} to obtain the inequality $u_{i-1}-s_{i-1} < \frac{u_{i-1}}{x}$, which yields
\begin{equation}
\label{eq-3*}
s_{i-1} > u_{i-1}\Bigl(1-\frac{1}{x}\Bigr)  \ .
\end{equation}
Let $T_{v_1}=T_{v_1,i-1}$ and $T_{v_2}=T_{v_2,i-1}$. 
Substituting the estimates $s_{i-1}\le |T_{v_j}|$ and $|T_{v_j}|\le u_{i-1}$ in~\eqref{eq-3*}, one obtains
\begin{equation}
\label{eqn:T_v_j}
u_{i-1}\Bigl(1-\frac{1}{x}\Bigr)\le |T_{v_j}|\le u_{i-1} \ , 
\end{equation} 
for both $j=1$ and $j=2$, which next yields 
\begin{equation}
\label{eq-4*}
\max\{|T_{v_1}\setminus T_{v_2}|, |T_{v_2}\setminus T_{v_1}|\}\le \frac{u_{i-1}}{x} \ . 
\end{equation}
Inequality \eqref{eq-4*} together with the estimate $ |r(v_1)-r(v_2)| \le  \max\{|T_{v_1}\setminus T_{v_2}|, |T_{v_2}\setminus T_{v_1}|\}$ imply
\begin{equation}
\label{eq-5*}
|r(v_1)-r(v_2)| \le \frac{u_{i-1}}{x} \ , 
\end{equation}
as $r(v_1)$ and~$r(v_2)$ are the ranks of the same task in \texttt{Tasks$_{v_1}$} and \texttt{Tasks$_{v_2}$}, respectively. 
We bound the difference $v_1-v_2$ from above using \eqref{eqn:ranks} as follows:
\begin{equation}
\label{eq-6*}
v_1-v_2<\frac{p}{|T_{v_1}|}\Bigl(r(v_1)+1\Bigr)-\frac{p}{|T_{v_2}|}\Bigl(r(v_2)-1\Bigr)
\ .
\end{equation}
We continue bounding $v_1-v_2$ from above by substituting the estimates given by \eqref{eqn:T_v_j} and \eqref{eq-5*} into the right-hand side of the inequality~\eqref{eq-6*}: 
\begin{align}
\label{eq-7*}
v_1-v_2
&< p\cdot 
\Bigl(\frac{r(v_2)+u_i/x+1}{|T_{v_1}|}-\frac{r(v_2)-1}{|T_{v_1}|+u_i/x}
\Bigr)
\nonumber \\ 
&=
p\cdot \frac{r(v_2)u_{i-1}/x+(u_{i-1}/x)^2+u_{i-1}/x+(u_{i-1}/x)|T_{v_1}|+2|T_{v_1}|}
{|T_{v_1}|(|T_{v_1}|+u_{i-1}/x)} \nonumber \\ 
&\le p\cdot 
\frac{u_{i-1}^2}{x}\cdot\frac{r(v_2)/u_{i-1}+1/x+1/u_{i-1}+|T_{v_1}|/u_{i-1}+2x|T_{v_1}|/u_{i-1}^2}
{u_{i-1}(1-1/x)(u_{i-1}(1-1/x)+u_{i-1}/x)} \ . 
\end{align}
Note that the numbers $r(v_2)$,  $|T_{v_1}|$ and $x$ are at least~1, and at most~$u_{i-1}$, each. 
We obtain the following upper bound
\[
v_1-v_2 
\le
p\cdot\frac{u_{i-1}^2}{x}\cdot\frac{6}{u_{i-1}^2(1-1/x)}
=
\frac{6p}{x-1}
\le 
\frac{7p}{x}\ ,
\]
by estimating the expression~\eqref{eq-7*} and using the fact that $x\ge 7$.
This completes the proof of~\eqref{eqn:2-star}. 

The concluding argument is as follows.
Suppose first that $7p < x$. Then at least $|H_{i}|$ processors perform different tasks.
Otherwise, if $7p \ge x$, then at least 
\[
\frac{| H_{i} | }{ \lfloor 7p/x\rfloor +1 } \leq
\frac{|H_{i}|}{2\cdot 7p/x}= |H_{i}|\cdot\frac{ x}{14p}
\]
processors in~$H_{i}$ perform different tasks. 
\end{proof}

For any execution, the indices $i$ of main epochs $\cE_i$ form a contiguous interval, call it~$\cI$.
For the epochs $\cE_i$ where~$i\in \cI$, the processors  in~$H_i$ are busy executing main phases. 
We first define two subsets $I_1$ and $I_2$ of~$\cI$ as follows:

Set~$I_1$ contains all $i$ from $\cI $ for which $\cE_i$ is stormy, that is, $|G_{i}| < |K_i |/2$.

Set~$I_2$ contains all $i$ from $\cI$ with the property $u_{i-1}\ge 11p^2g(p)$.

Let $k+1$ be the smallest among all epoch indices $i$ such that $u_{i}<11p^2g(p)$. 
We will use Lemma~\ref{lem:x-as-parameter} for $x$ equal to~$\sqrt{\frac{u_{k}}{\lg u_{k}}}$; in what follows we use $x$ as a shorthand for $\sqrt{\frac{u_{k}}{\lg u_{k}}}$. 
Let us define three subsets $I_3$, $I_4$, and $I_5$ of $\cI \setminus (I_1\cup I_2)$ as follows:

Set $I_3$ contains all $i$ from $\cI \setminus (I_1\cup I_2)$ such that $u_{i-1} < x$. 

Set $I_4$ contains all $i$ from $\cI \setminus (I_1\cup I_2)$ such that $u_{i-1} \ge x$ and $u_{i-1}-u_{i}\ge \frac{u_{i-1}}{x}$.

Set $I_5$ contains all $i$ from $\cI \setminus (I_1\cup I_2)$ such that $u_{i-1} \ge x$ and
\begin{equation}
\label{eqn:min-decrement}
u_{i-1}-u_{i}\ge\min \Bigl\{\frac{x}{14p},1\Bigr\}|H_{i}|\ . 
\end{equation}

It follows from Lemma~\ref{lem:x-as-parameter} that $I_3\cup I_4\cup I_5=\cI\setminus (I_1\cup I_2)$, and hence $I_1\cup I_2\cup I_3\cup I_4 \cup I_5=\cI$.

The following theorem summarizes the properties of algorithm \aBL.


\begin{theorem}
\label{thm:work-balance-load}

Algorithm \aBL\ is a constructive deterministic solution for the \DA\ problem with $p$ processors and $t$ tasks that can be instantiated, for any known $f<p$, to have work $\cO(  t +   p \log p\, (\sqrt{p\log p}+\sqrt{t\log t}\, )  )$.
\end{theorem}

\begin{proof}
We estimate the amount of work in main, mixed and closing epochs separately.
We first consider main epochs, whose indices are in~$\cI$.
We consider the contribution of each set from the cover $I_1\cup I_2\cup I_3\cup I_4 \cup I_5$ of the interval $\cI$ one by one. 

The epochs with indices in $I_1$ contribute $\cO(p\log p)$ to work, by Lemma~\ref{lem:stormy-epochs}.
To estimate the amount of work during epochs~$\cE_i$ for $i\in \cI\setminus I_1$ that are calm, we first calculate the contribution by the respective core  processors.

The work performed by the core processors, during epochs $\cE_i$ whose indices $i$ are in~$I_2$, is~$\cO(t)$, as during these epochs the number of the outstanding tasks decreases proportionately to the number of executed available processor steps, 
per Lemma~\ref{lem:no-overlap-in-tasks}.

The work performed by the core processors, during the epochs whose indices are in~$I_3$, is $\cO\bigl( p\log p\,\sqrt{\frac{t}{\log t}}\,\bigr)$, since $x\le \sqrt{\frac{t}{\log t}}$ and an epoch consists of $g(p)=\cO(\log p)$ phases.

If $i\in I_4$, then $u_{i}\le u_{i-1}(1-\frac{1}{x})$, so each such epoch contributes to decreasing $u_{i-1}$ by a factor of at least $1-\frac{1}{x}$. 
The number of such epochs $\cE_i$ is $\cO(x \lg u_{k})$, because 
\[
u_{k}\Bigl(1-\frac{1}{x}\Bigr)^{x\lg u_{k}}=\cO(1)\ .
\]
It follows that 
\[
|I_4|
=
\cO(x \lg u_{k})
=
\cO\Bigl(\sqrt{\frac{u_{k}}{\lg u_{k}}} \lg u_{k}\Bigr)
=
\cO(\sqrt{u_{k}\log u_{k}})
=
\cO(\sqrt{t\log t})\ ,
\] 
as $u_{k}\le t$. 
Therefore the work accrued by the core processors, during the epochs whose indices~$i$ are in~$I_4$, is $\cO(p\log p\,\sqrt{t\log t})$.

Next we consider the work by the core processors during the epochs whose indices are in~$I_5$.
We estimate the sum of the decrements $u_{i-1} - u_{i}$ for $i\in I_5$ by the telescoping series: 
\[
\sum_{i\in I_5} (u_{i-1} - u_{i})
\le
\sum_{i\geq k+1} (u_{i-1} - u_{i})
\le u_k\ .
\]
After combining this with~\eqref{eqn:min-decrement}, we obtain that
\[
u_{k}\ge \sum_{i\in I_5} (u_{i-1} - u_{i}) 
\ge 
\sum_{i\in I_5} \min \Bigl\{\frac{x}{14p},1\Bigr\}|H_{i}|\ .
\]
This implies that
\begin{equation}
\label{eqn:bb}
\sum_{i\in I_5} |H_{i}| 
\le 
u_{k} \max\Bigl\{\frac{14p}{x},1\Bigr\}
\le
\max\Bigl\{\frac{14p\,u_{k}\sqrt{\lg u_{k}}}{\sqrt{u_{k}}},u_{k}\Bigr\}
\ .
\end{equation}
Since $u_{k}$ is less than $11p^2 g(p)$ and at most~$t$, we obtain that 
\[
u_{k}  <  p\sqrt{11 u_{k}g(p)}  \le  p\sqrt{11 \, t \, g(p) }.
\]
This makes the right-hand side of~\eqref{eqn:bb} to be at most
\begin{equation}
\label{eqn:cc}
\max\{14p\,\sqrt{u_{k}\lg u_{k}},u_{k}\}
\le 
\max\{14p\,\sqrt{t\lg t},p\sqrt{11\, t\, g(p)}\}
\ .
\end{equation}
By \eqref{eqn:bb} and \eqref{eqn:cc}, the work contributed by the processors in the core graphs, 
during epochs $\cE_i$ for~$i\in I_5$, is 
\begin{eqnarray*}
\sum_{i\in I_5} |H_{i}| g(p)
&\le&
\, p \, g(p) \max\bigl\{14\,\sqrt{t\lg t} \, ,\sqrt{11 \, t \, g(p)}\bigr\}\\
&=&
\cO\bigl(p \log p\, \bigl(\sqrt{t\log t}+\sqrt{t\log p}\, \bigr) \bigr)\ .
\end{eqnarray*}

The work contributed by the core processors during the epochs $\cE_i$ with $i\in \cI$, is obtained by summing up the contributions of $I_1$, $I_2$, $I_3$, $I_4$, and~$I_5$.
We obtain  the following bound on work performed by core processors during main epochs:
\begin{gather}
\nonumber
\cO(p\log p) +
\cO(t) + 
\cO\Bigl(p \log p\, \sqrt{\frac{t}{\log t}}\,\Bigr) +
\cO\bigl(p \log p\, \sqrt{t\log t}\bigr) +
\cO\bigl(p \log p\, \bigl(\sqrt{t\log t}+\sqrt{t\log p}\, \bigr) \bigr) \\
\label{eqn:bound-work}
=
\cO \bigl( t  + p \log p\, \bigl(\sqrt{t\log t}+\sqrt{t\log p}\, \bigr)  \bigr) 
\ .
\end{gather}
By Lemma~\ref{lem:core-processors-suffice}, bound \eqref{eqn:bound-work} is also an estimate of total work in calm main epochs.

Two consecutive mixed epochs are sufficient for every core processor to learn that there are no outstanding tasks, by Lemma~\ref{lem:one-epoch-switch}. 
These mixed epochs contribute only $\cO(p\log p)$ to work.

To estimate the work during closing epochs, we apply Lemma~\ref{lem:closing-is-as-efficient-as-main}.
The bound on this work is obtained by substituting $p$ for~$t$ in~\eqref{eqn:bound-work} and adding $\cO(p\log p)$, yielding
\[
\cO\bigl(p + p\log p+ p \log p\,\sqrt{p\log p}\bigr)=\cO\bigl((p\log p)^{3/2}\bigr)\ .
\]   
This bound on work in closing epochs,  combined with $\cO(p\log p)$ work in mixed epochs, and with bound~\eqref{eqn:bound-work} on work in main epochs, gives
\[        
\cO\bigl( t+ p\log p+ p\log p\, \bigl(\sqrt{t\log t}+\sqrt{t\log p}\, \bigr) + (p\log p)^{3/2} \bigr) = 
\cO\bigl(t+p \log p\, \bigl(\sqrt{p\log p} + \sqrt{t\log t}\,\bigr)\bigr)
\]
as the bound on total work.
\end{proof}

Note that the bound  in Theorem~\ref{thm:work-balance-load} does not depend on~$f$, 
even though the instantiations of the generic algorithm use overlay graphs of degrees determined by~$f$, as spelled out in Lemma~\ref{lem:degree-power-of-Lp}.
The degrees of overlay graphs affect communication, as stipulated in Lemma~\ref{lem:communication-from-work}.


\begin{corollary}
\label{cor1}

Algorithm \aBL\ is a constructive deterministic solution for the \DA\ problem with $p$ processors and $t$ tasks that can be instantiated, for a number of crashes $f$ bounded by $f\le cp$, for any known constant $0<c < 1$, to have effort $\cO( t+p \log p\, (\sqrt{p\log p} + \sqrt{t\log t}\,) )$.   
\end{corollary}

\begin{proof}
Since the adaptive known adversary is restricted to be linearly bounded,
the assumption about failures implies $p-f\ge (1-c)p$.
This makes the bound on communication given in Lemma~\ref{lem:communication-from-work} to become $\cM=\cO(\cW)$.
The bound on work is taken from Theorem~\ref{thm:work-balance-load}.
We then combine the bounds on work~$\cW$ and communication~$\cM$ to obtain a bound on effort $\cE=\cW+\cM$.
\end{proof}

\section{An Algorithm Optimized for Work} 

\label{sec:algorithm-optimized-for-work}

We obtain a specific algorithm by instantiating the generic algorithm with a task selection rule such that the processors choose items from lists in the order determined by their private permutations.
We call the resulting constructive algorithm \aDP. 
The difference between \aBL\  in Section~\ref{sec:constructive} and \aDP\  given in this section is that the former is constructive, whereas the latter is nonconstructive, as it resorts to a selection rule that uses nonconstructive permutations.

All our algorithms in this section and in Section~\ref{sec:constructive} use the same constructive overlay graphs for communication: the parameter~$f<p$ determines the overlay graph~$G(p,f)$, as defined in Section~\ref{sec:graphs}.
These graphs are discussed in Section~\ref{sec:design-algorithms}; in particular, Lemma~\ref{lem:communication-from-work} determines how the amount of communication is related to the accrued work.

We consider two ways for establishing the suitable permutations used for task selection.
First, we consider a randomized algorithm, called \aRP, where processors use private permutations generated randomly at the beginning of an execution.
(The purpose of considering the randomized algorithm \aRP\ is only to prove the existence of certain permutations that are then used in a deterministic algorithm.)
Next, we consider the deterministic algorithm \aDP, in which the code is parameterized by private permutations assigned to  processors. 
We show that algorithm \aDP\ can be equipped with suitable permutations that make it more efficient than the deterministic algorithm \aBL.

We now specify how permutations are used to select items from lists.
Each processor~$v$ has two private permutations: $\pi_1$ over the set $\{1, \ldots,p\}$, 
and~$\pi_2$ over the set $\{1,\ldots,11p^2g(p)\}$. 
Processor~$v$ permutes its list \texttt{Busy$_v$} according to the permutation~$\pi_1$.
To explain how the permutation~$\pi_2$ is used, we consider two cases.

The first case is for $t\ge 11p^2g(p)$.
Processors select tasks to perform according to a load balancing rule.
This rule is the same as the one used in algorithm \aBL, as long as the size of \texttt{Tasks$_v$} is greater than~$11p^2g(p)$. 
As soon as exactly $11p^2g(p)$ tasks remain in the list \texttt{Tasks$_v$}, processor~$v$ permutes \texttt{Tasks$_v$} according to the permutation~$\pi_2$.

In the second case we have $t<11p^2g(p)$.
Here the list \texttt{Tasks$_v$} consisting of $t$ entries is rearranged by the permutation~$\pi_2$ as if it consisted of $11p^2g(p)$ entries.
More precisely this means the following: we pad the list \texttt{Tasks$_v$} to  the  length $11p^2g(p)$ by appending ``dummy'' items, then permute the resulting list, and finally remove the dummy items to compact the list to its original size.

Each time processor~$v$ needs to select an item from \texttt{Tasks$_v$}, 
it selects the first item from the rearranged list, then removes it from the list.
If~$v$ needs to select a processor from \texttt{Busy$_v$} in a closing phase, 
then it selects the first processor from the rearranged list, then removes it from the list.

The randomized algorithm \aRP\ starts by each processor selecting two permutations uniformly at random and independently across all  processors. 
The deterministic algorithm \aDP\ equips each processor with its individual pair of permutations. 
We show the existence of permutations that guarantee work $\cO( t+p\log^2 p)$ of \aDP.
A technical challenge we encounter here is to demonstrate that the complexity of \aDP\ can be made suitably small for some permutations. 
Algorithm \aRP\ is used only as an intermediate step to facilitate the analysis of the deterministic counterpart. 
The idea is to estimate the probability that algorithm \aRP\ deviates from its expected performance, in order to next apply the probabilistic method to argue that there exists a family of permutations that achieves a comparable performance in the worst case.

We use the following Chernoff bound on the probability that a sum of independent random variables deviates from its expectation:
For $0<r<1$, let $X_1,\ldots, X_n$ be a sequence of independent Bernoulli trials with $\Pr(X_j=1)=r$ and~$\Pr(X_j=0)=1-r$, for each~$1\le j\le n$.
If $S=\sum_{j=1}^n X_j$ then, for any $0<\varepsilon<1$, the following inequality holds \cite{MitzenmacherUpfal-book05}:
\begin{equation}
\label{eqn:chernoff}
\Pr ( S\le (1-\varepsilon)nr ) \le \exp(-nr\varepsilon^2/2) \ .
\end{equation}

\subsection{Extended epochs}

\label{sec:extended-epochs}

Recall that an epoch~$\cE_i$ begins with $K_i$ as the set of non-faulty nodes, that it ends with~$G_i$ as the set of non-faulty nodes, 
and that the core graph~$H_i$ is a compact subgraph of~$G_i$, for integer~$i\ge 0$.
We define an \emph{extended epoch $\cD_j$} of an execution, for integer $j\ge 0$, to be a contiguous segment 
of (regular) epochs of the execution, subject to the following additional restrictions.

The first extended epoch~$\cD_0$ consists of the first regular epoch~$\cE_0$, which denotes initialization.
Suppose  that~$\cE_k$ is the first regular epoch of an extended epoch~$\cD_j$, for some $j\ge 0$.
If either $|G_{k}| < |K_k|/2$  or $|H_k|g(p) \ge u_{k-1}$ hold, then the extended epoch~$\cD_{j}$ consists of only this one regular epoch~$\cE_k$.
Otherwise, when both $|G_{k}| \ge |K_k|/2$ and $|H_k|g(p) < u_{k-1}$ hold, then the extended epoch~$\cD_j$ consists of a contiguous segment of epochs $\langle\cE_k, \ldots, \cE_\ell\rangle$ such that $\ell\ge k$ is the largest integer for which both $|G_{\ell}| \ge |K_k|/2$ and $(\ell-k+1)\,|H_{\ell}|\,g(p) < u_{k-1}$ hold.

The number~$\ell$ is well defined because the quantity $(\ell-k+1)(p-f)g(p)$ is unbounded as a function of~$\ell$.
When the extended epoch~$\cD_j$ is the segment $\langle\cE_k, \ldots, \cE_\ell\rangle$ of regular epochs, then the next extended epoch $\cD_{j+1}$ begins with the regular epoch~$\cE_{\ell+1}$.

We define the number~$m(j)$ to be the index of the first regular epoch in~$\cD_j$, and, similarly,  the number~$l(j)$ to be the index of the last regular epoch in~$\cD_j$, so that $\cD_j=\langle \cE_{m(j)},\ldots, \cE_{l(j)}\rangle$.
We also define $m(-1)=-1$ and $l(-1)=-1$ for convenience of notation,  letting $u_{m(-1)}=t$ and~$u_{l(-1)}=t$.


\Paragraph{Stormy extended epochs.}

If $|G_{l(j)}| \ge |K_{m(j)}|/2$, then the extended epoch $\cD_j$ is called \emph{calm}, otherwise, when the inequality $|G_{l(j)}| < |K_{m(j)}|/2$ holds, then this extended epoch $\cD_j$ is \emph{stormy}.

\begin{lemma}
\label{lem:stormy-extended-epoch}

An extended epoch is stormy if and only of it consists of one stormy regular epoch.
\end{lemma}

\begin{proof}
Let us consider the first regular epoch $\cE_k$ of an extended epoch~$\cD_j$, so that $m(j)=k$.
If epoch~$\cE_k$ is stormy then  $|G_{k}| < |K_k|/2$ and so, by the definition of an extended epoch, $\cD_j$ consists of only the epoch~$\cE_k$, and also $\cD_j$ is clearly a stormy extended epoch.
If $\cE_k$ is calm then the inequality $|G_{k}| \ge |K_k|/2$ holds.
Now $\cD_j$ consists of a contiguous segment of epochs $\langle\cE_m(j), \ldots, \cE_l(j)\rangle$ such that both $|G_{l(j)}| \ge |K_{m(j)}|/2$ and $(l(j)-m(j)+1)\,|H_{l(j)}|\,g(p) < u_{m(j)-1}$ hold, so $D_j$ is calm.
\end{proof}


\begin{lemma}
\label{lem:work-stormy-extended-epochs}

Work performed during stormy extended epochs is $\cO(p\log p)$.
\end{lemma}

\begin{proof}
Each stormy extended epoch consists of one  stormy regular epoch, by Lemma~\ref{lem:stormy-extended-epoch}.
The stormy regular epochs, that make also stormy extended epochs, contribute $\cO(p\log p)$ to work, by Lemma~\ref{lem:stormy-epochs}.
\end{proof}


\Paragraph{Core processors and main extended epochs.}

We define $Z_j$ to denote $H_{l(j)}$, for an extended epoch~$\cD_j$.
The subgraph $Z_j$ is called the \emph{core subgraph for~$\cD_j$}, and the processors in $Z_j$ are \emph{the core processors for~$\cD_j$}.

Extended epochs in which all core processors are busy are called \emph{main extended epochs}.


\begin{lemma}
\label{lem:graphs-Z-suffice}

If $w_j$ work is accrued by the core processors during a calm main extended epoch~$\cD_j$, then the total work accrued during $\cD_j$  is~$\cO(w_j)$.
\end{lemma}

\begin{proof}
No core processor halts in the extended epoch~$\cD_j$, as this is possible only in the course of executing \texttt{Closing}, by the pseudocode in Figures~\ref{fig:procedure-main} and~\ref{fig:procedure-closing}.
Therefore graph $Z_j$ stays intact throughout the extending epoch~$\cD_j$.

Graph~$Z_j$ is defined as the compact subgraph of~$G_{l(j)}$ of the form $Z_j=P(G_{l(j)})$,  where~$P$ is a subgraph function by Lemma~\ref{lem:subgraph-property}.
It follows that $|G_{l(j)}|\le 7 |H_j|$.
We can combine this fact with the inequality $|K_{m(j)}| < 2 |G_{l(j)} |$, given by the assumption, to obtain $|K_{m(j)}| < 14 |Z_j |$.
Therefore the total work accrued during~$\cD_j$  is less than $14 w_j$.
\end{proof}

We denote by $c_j$ the number of selections of tasks made by the core processors (those in~$Z_j$)  during the extended epoch~$\cD_j$;  
this number~$c_j$ is called \emph{the core number for~$\cD_j$}.
We have that  
\begin{equation}
\label{eqn:core-number}
c_j=(l(j)-m(j)+1)\,|Z_j|\,g(p)\ ,
\end{equation}
as there are $l(j)-m(j)+1$ regular epochs in~$\cD_j$, each taking $g(p)$ rounds, with one  selection of a task by a processor in~$Z_j$ per round.
An extended epoch $\cD_j$ begins with $u_{l(j-1)}$ tasks still present in some lists of the processors, and ends with $s_{l(j)}$ tasks that occur in every list.

\begin{lemma}
\label{lem:work-extended-epoch}
The work accrued during one extended epoch by processors that stay core through its end is $\cO(t+p\log p)$.
\end{lemma}

\begin{proof}
We consider two cases.  
The first case holds when  $\cD_j$ consists of one epoch.
In this case, let $\cE_k$ be the first regular epoch of an extended epoch~$\cD_j$.
Then the amount of work accrued by the processors that are core for $\cE_k$ is $\cO(p\log p)$, because an epoch lasts $\cO(\log p)$ rounds.
The other case is when there are other epochs that follow~$\cE_k$; let $\cE_\ell$ be the last epoch of~$\cD_j$.
Then the work accrued during~$\cD_j$ by processors that are core for $\cE_\ell$ is smaller than $u_{k-1}\le t$, by the definition of an extended epoch.
\end{proof}

We partition the extended epochs of an execution into two categories as follows.
One of them consists of extended epochs~$\cD_j$ with $c_j < u_{l(j-1)}$, for $j\ge 0$;
such extended epochs are called \emph{task-rich} extended epochs.
Another category consists of extended epochs~$\cD_j$ with $c_j \ge  u_{l(j-1)}$, for~$j\ge 0$;
such extended epochs are called \emph{task-poor} extended epochs.


\begin{lemma}
\label{lem:few-tasks-extended-epoch}

An extended epoch that is task-poor consists of one regular epoch.
\end{lemma}

\begin{proof}
Consider an extended epoch $\cD_j = \langle \cE_{m(j)},\ldots, \cE_{l(j)}\rangle$.
If it is task-poor then this means $|H_{l(j)}| \, g(p) \ge u_{l(j-1)}$, because $Z_j=H_{l(j)}$.
Therefore, by the definition of extended epoch, extended epoch~$\cD_j$ consists of only the regular epoch~$\cE_{l(j)}=\cE_{m(j)}$.
\end{proof}


\Paragraph{Productive extended epochs.}

An extended epoch~$\cD_j$  is called \emph{productive} when the inequality
\begin{equation}
\label{eqn:def-productive}
u_{l(j-1)}-s_{l(j)}\ge \min\{u_{l(j-1)},c_j\}/4
\end{equation} 
holds at the completion of~$\cD_j$.
The right-hand side of~\eqref{eqn:def-productive} can be simplified when we know if the extended epoch is task-poor or task-rich.

Task-rich extended epochs $\cD_j$ are defined by  $c_j < u_{l(j-1)}$, therefore the inequality
\begin{equation}
\label{eqn:many-tasks-extended}
\min\{u_{l(j-1)},c_j\}/4 \ge  c_j/4
\end{equation}
holds for such~$j$.

Task-poor extended epochs $\cD_j$ are defined by $c_j \ge u_{l(j-1)}$, therefore the inequality 
\begin{equation}
\label{eqn:few-tasks}
\min\{u_{l(j-1)},c_j\}/4 \ge u_{l(j-1)}/4
\end{equation}
holds for such~$j$.


\begin{lemma}
\label{lem:J-4}

There are $\cO(\log p)$ productive task-poor main extended epochs, in any execution of the generic algorithm.
\end{lemma}

\begin{proof}
We use the sequence $\langle u_i \rangle_{i\ge 0}$ as reflecting progress in completing tasks (see the discussion just after Lemma~\ref{lem:inequality-ui-by-si}).
If $\cD_j$ is a productive task-poor extended epochs then we can combine~\eqref{eqn:def-productive} with \eqref{eqn:few-tasks} to obtain that  
\begin{equation}
\label{eqn:decrement-J4}
u_{l(j-1)}-u_{m(j+1)}\ge u_{l(j-1)} / 4\ ,
\end{equation} 
as $u_{m(j+1)}\le s_{l(j)}$.
Any such extended epochs $\cD_j$ consist of only one regular epoch~$\cE_{m(j)}=\cE_{\ell(j)}$, by Lemma~\ref{lem:few-tasks-extended-epoch}. 
The inequality~\eqref{eqn:decrement-J4}  means that two consecutive regular epochs $\cE_{m(j)}$ and $\cE_{m(j)+1}$ contribute to decreasing the number of outstanding tasks by at least  $u_{m(j)} / 4\le u_{l(j-1)}/4$.
The number of tasks in the first productive task-poor main extended epoch is less than $11p^2g(p)$, so there are  $\cO(\log (11p^2g(p))=\cO(\log p)$ productive task-poor main extended epochs.
\end{proof}

For any execution, the indices $j$ of main extended epochs $\cD_j$ form a contiguous interval, call it~$\cJ$.
We first identify two subsets of~$\cJ$ as follows:

Set $J_1$ consists of these $j\in \cJ$ for which $\cD_j$ is stormy, that is, $|G_{l(j)}| < |K_{m(j)} |/2$.

Set $J_2$  consists of these $j\in \cJ$  for which $u_{\ell(j-1)} \ge 11p^2 g(p)$. 

There is at most one extended epoch $\cD_j$ such that $u_{m(j)} \ge 11p^2g(p)$ and $u_{l(j)} < 11p^2g(p)$.
We let $J_3$ be the singleton set $\{j\}$, when such a $\cD_j$ exists, otherwise $J_3=\emptyset$.

The extended epochs $\cD_j$ for $j\in \cJ \setminus (J_1\cup J_2\cup J_3)$ are called \emph{key}.
We define two subsets $J_4$ and~$J_5$  of the set of indices of the key extended  epochs $\cJ \setminus (J_1\cup J_2\cup J_3)$ as follows:

Set $J_4$ consists of these $j \in \cJ \setminus (J_1\cup J_2\cup J_3)$ for which the key extended epoch~$\cD_j$ is a task-rich extended epoch.

Set $J_5$ consists of these $j\in \cJ \setminus (J_1\cup J_2\cup J_3)$ for which the key extended epoch~$\cD_j$ is a task-poor extended epoch.

Note that $\cJ= J_1\cup J_2\cup J_3\cup J_4\cup J_5$.


\begin{lemma}
\label{lem:key-epochs-productive}

For any known $f<p$, the work accrued in any execution of the generic algorithm is $\cO (t +p\log^2 p)$ during the main extended epochs~$\cD_j$, provided that all the key extended epochs among them are productive.
\end{lemma}

\begin{proof}
Stormy extended epochs, those with indices in $J_1$, contribute $\cO(p\log p)$ to work, by Lemma~\ref{lem:work-stormy-extended-epochs}.

The amount of work during extended epochs~$\cD_j$, for $j\in \cJ\setminus J_1$, that are calm, can be estimated by calculating the contribution by the respective core  processors, by Lemma~\ref{lem:graphs-Z-suffice}.
Observe that the work performed by the processors in~$Z_{j}$ during $\cD_j$ is~$c_j$, where $c_j$ is the core number of~$\cD_j$ as defined by equation~\eqref{eqn:core-number}.

Let us consider the work contributed by the core processors for the extended epochs with indices in~$J_2$ and such that all their regular epochs $\cE_k$ satisfy the inequality $u_k \ge 11p^2g(p)$. This work is~$\cO(t)$ by Lemma~\ref{lem:no-overlap-in-tasks}.

There is at most one extended epoch~$\cD_j$ such that $u_{m(j)} \ge 11p^2g(p)$ and $u_{l(j)} < 11p^2g(p)$;  and if it exists then $J_3=\{ j\}$.
The work accrued during any one extended epoch by processors core for this extended epoch is~$\cO(t+p\log p)$, by Lemma~\ref{lem:work-extended-epoch}; this applies in particular to~$\cD_j$.

Next, let us estimate the amount of work by the core processors during task-rich key extended epochs~$\cD_j$; their indices form the set~$J_4$. 
These extended epochs are productive, by the assumption.
We combine \eqref{eqn:def-productive} and~\eqref{eqn:many-tasks-extended} with $u_{m(j+1)}\le s_{l(j)}$ to obtain  
\begin{equation}
\label{eqn:many-tasks}
u_{l(j-1)}-u_{m(j+1)}\ge c_j/4\ .
\end{equation} 
We use the sequence $\langle u_i \rangle_{i\ge 0}$ as reflecting progress in completing tasks (see the discussion just after Lemma~\ref{lem:inequality-ui-by-si}).
With this interpretation of the numbers $u_{l(j-1)}$ and $u_{m(j+1)}$, we see that  bound \eqref{eqn:many-tasks} means that the progress in completing tasks during $\cD_j$ and $\cD_{j+1}$, assuming $\cD_j$ is task-rich, is proportional to $c_j$, which is the work expended by the processors that are core for~$\cD_j$. 
It follows that the work performed by core processors during productive task-rich extended epochs is~$\cO(t)$.

As the last step, we estimate the amount of work by the core processors during task-poor key extended epochs~$\cD_j$; their indices form the set~$J_5$. 
There are  $\cO(\log p)$ such extended epochs $\cD_j$, by Lemma~\ref{lem:J-4}, as each of them is productive, by the assumption.
Therefore task-poor extended epochs  contribute $\cO(p \log ^2 p)$  work, since each such extended epoch consists of a single regular epoch, by Lemma~\ref{lem:few-tasks-extended-epoch}.

We obtain  $\cO (t +p\log^2 p)$ as the bound on work by summing up the contributions of~$\cD_j$, for all~$j\in \cJ$.
\end{proof}

\subsection{Restricted adversaries and randomization}

\label{sec:restricted-adversaries}

We expect that if permutations are assigned to the processors in a random and independent manner then work will be efficient with positive probability for crashes happening at arbitrary times.
If so, this will be sufficient to claim the existence of permutations that can be fixed in the code of the algorithm so that work will be efficient in a deterministic instantiation of the generic algorithm against the general adaptive adversary.

We consider restricted adversaries as technical means for approximating the behavior of the general adaptive adversaries in the probabilistic analysis of a randomized instantiation of the generic algorithm.
Random selections of tasks by core processors make the dominating contribution to progress in performing tasks in the course of an extended epoch.
These selections are made independently, which facilitates the probabilistic analysis.
For such analysis to hold, we need to be certain that the processors that make the selections have not crashed yet.
To this end, we simply specify what processors are core at the beginning of the extended epoch, so that these processors will not crash through the end of the extended epoch.
This is presented as if the adversary were constrained not to be able to crash these particular processors in this execution.
This is only a conceptual constraint, because ultimately we will consider all the possible sets of processors that do not crash through the end of an extended epoch.
The adaptive adversary is restricted by the number of crashes anyway, so in retrospect we can interpret any specific execution, in which some processors end up as operational at the end of an extended epoch, as if the adversary were restricted  at the outset of the extended epoch to spare  these processors.


\Paragraph{Bars on crashes.}

Recall that we denote by~$K_{m(j)}$ the set of processors that are non-faulty at the beginning of the extended epoch~$\cD_j$, and by~$G_{l(j)}$ the set of non-faulty processors at the end of~$\cD_j$. 

A set of processors $B\subseteq K_{m(j)}$ is said to be a \emph{bar for crashes for~$\cD_j$} if  the adversary is constrained at the beginning of~$\cD_j$ by the requirement that $B=G_{l(j)}$.

This definition means that the processes in the set~$B$ cannot fail and that all the processes in the set~$K_{m(j)}\setminus B$ must fail during the extended epoch~$\cD_j$. 
Alternatively, given a set of processors~$B$, we could consider extended epochs~$\cD_j$ that happen to result in~$B=G_{l(j)}$, rather than present this as a constraint on the adversary.
This is because we consider all possible bars on crashes in the proof of Lemma~\ref{lem:probability-f-chain}, where we use  Lemma~\ref{lem:bar-on-crashes}.
For each extended epoch~$\cD_j$  the respective set~$G_{l(j)}$ is well defined,  so in retrospect we can consider $G_{l(j)}$ as bar for crashes when $\cD_j$ ends.


\begin{lemma}
\label{lem:bar-on-crashes}

For any known $f<p$, if a bar for crashes is specified at the beginning of a main extended epoch~$\cD_j$ of an execution of algorithm \aRP, then $\cD_j$  is not productive with probability at most~$\exp(- c_j/16)$.
\end{lemma}

\begin{proof}
Let us consider a specific extended epoch~$\cD_j$.
If the inequality 
\begin{equation}
\label{eqn:s-and-u}
s_{l(j-1)} \le \frac{3u_{l(j-1)}}{4}
\end{equation}
is satisfied, then the following estimates
\[
 u_{l(j-1)} - s_{l(j)} \ge u_{l(j-1)} - s_{l(j-1)}\ge u_{l(j-1)}- \frac{3u_{l(j-1)}}{4} = \frac{u_{l(j-1)} }{ 4}
\] 
hold as well.
This implies that the inequality~\eqref{eqn:def-productive}, which defines productive extended epochs, holds with certainty.
Next consider the case when~\eqref{eqn:s-and-u} does not hold.
Let $b$ denote $\frac{\min\{u_{l(j-1)},c_j\}}{4}$; this  is the right-hand side of~\eqref{eqn:def-productive}.
What needs to be shown is that the inequality $u_{l(j-1)}-s_{l(j)}\ge b$ holds with the claimed probability.

Let us order one by one the $c_j$ random selections of tasks by the core processors in~$\cD_j$.
These selections determine a sequence of operations performed, as specified by the algorithm.
The adversary does not crash any of the core processors in the course of~$\cD_j$, as each of them belongs to~$G_{l(j)}$.
Let $S^{(0)}$ denote~$S_{l(j-1)}$, and let $S^{(k)}$ be the set of tasks from~$S_{l(j-1)}$ that are not completed by the  $k$th selection, in this one-by-one ordering.

We define the random variables $X_k$ as follows, for $1\le k\le c_j$:
\[
X_k=\left\{
\begin{array}{ll}
1 & \text{ if either \ } u_{l(j-1)}-|S^{(k)}|\ge b \text{\ \  or \ } 
		S^{(k)}\varsubsetneq S^{(k-1)} \ \\
0 & \text{ otherwise. }
\end{array}
\right .
\]
The intended meaning of this definition is that $X_k$ equals~$1$ when either 
\begin{enumerate}
\item[(a)] 
the target inequality $u_{l(j-1)}-s_{l(j)}\ge b$ is already satisfied in the round when the $k$th task  is selected, or 

\item[(b)] 
the $k$th task selected from $S_{l(j-1)}$ increases the set of the tasks already selected to be performed in this extended epoch,
that is, this task has not been selected  to be performed in this extended epoch prior to this $k$th selection.
\end{enumerate}

We define $X=\sum_{1\le k\le c_j} X_k$.
This random variable has the property that $X\ge b$ if and only if  $u_{l(j-1)}-s_{l(j)}\ge b$.
The random variables $X_1,\ldots,X_k$ are well-defined and mutually independent, as long as  $u_{l(j-1)}-|S^{(k)}| < b$.
This is because these random variables are determined by random and mutually independent selections of tasks by the core processors, while no such processor may crash in~$\cD_j$, due to the assumed restriction on the adversary.

Suppose that processor~$v$ is to perform its $k$th operation in~$\cD_j$.
We want to estimate how many tasks in~\texttt{Tasks$_v$} are outstanding.
Regarding the contents of~\texttt{Tasks$_v$} at the end of the extended epoch~$\cD_{j-1}$, at most $u_{l(j-1)}-|S^{(k-1)}|$ tasks in \texttt{Tasks$_v$} are performed during the first $k-1$ operations, since the size of  \texttt{Tasks$_v$} at the end of~$\cD_{j-1}$ is at most~$u_{l(j-1)}$.
Let us recall two inequalities now.
One is $s_{l(j-1)} > \frac{3u_{l(j-1)}}{4}$; it holds because the inequality \eqref{eqn:s-and-u} does not. 
The second is $b\le \frac{u_{l(j-1)}}{4}$; it holds because $b$ denotes $\frac{\min\{u_{l(j-1)},c_j\}}{4}$.
By combining these two inequalities, we obtain that if $u_{l(j-1)}-|S^{(k-1)}|<b$, then 
\begin{eqnarray*}
s_{l(j-1)}-\big(u_{l(j-1)}-|S^{(k-1)}|\big) &>& s_{l(j-1)}-b\\
& >& \frac{3u_{l(j-1)}}{4} -  \frac{u_{l(j-1)}}{4} \\
&=& \frac{u_{l(j-1)}}{2}  \ .
\end{eqnarray*}
It follows that at least half of the tasks that are in the list \texttt{Tasks$_v$}, in the round when the $k$th selection is made, have not been performed yet.
Therefore the probability that $v$ chooses an outstanding task in the $k$th selection is at least~$\sfrac{1}{2}$, and this outcome is independent from the preceding choices.
Selecting an outstanding task to perform in the $k$th operation results in $S^{(k)}\subsetneq S^{(k-1)}$, which means that $X_k=1$. 

Consider a sequence $\langle Y_k\rangle_{1\le k\le c_j}$ of independent Bernoulli trials, with $\Pr(Y_k=1 ) =\sfrac{1}{2}$. 
Define $Y=\sum_{1\le k\le c_j} Y_k$.
Now we claim $X$  statistically dominates  $Y$, in the sense that, for any $d>0$: 
\begin{equation}
\label{eqn:statistic-dominance}
\Pr (X \le d )\le \Pr(Y \le d )\ .
\end{equation}
This is because of how the constituent random variables $X_k$ and $Y_k$ contribute to~$X$ and~$Y$, respectively.
Namely, the random variables~$X_k$ of~$X$ have the property that  each $X_k=1$ with probability at least $\sfrac{1}{2}$ independently from the prior values, as long as either $k=1$ or  
\[
\sum_{j=1}^{k-1}X_j<u_{l(j-1)}-s_{l(j-1)}-b\ ,
\]
while~$X_k=1$ when the following is satisfied
 \[
 \sum_{i=1}^{k-1}X_i \ge u_{l(j-1)}-s_{l(j-1)}-b\ .
 \] 
Comparing $X$ to $Y$, the constituent random variables~$Y_i$ of~$Y$  are mutually independent Bernoulli trials, in which $Y_i=1$ with probability~$\sfrac{1}{2}$, for all $1\le i\le c_j$.

We resort to~\eqref{eqn:statistic-dominance} and  $b\le \frac{c_j}{4}$ to estimate the probability $\Pr(X\le b)$. 
These inequalities yield 
\[
\Pr(Y \le b)
\le
\Pr\Bigl( Y\le \frac{c_j}{4}\Bigr) 
=
\Pr\Bigl( Y\le \Bigl(1-\frac{1}{2}\Bigr) \cdot c_j \cdot \frac{1}{2}\Bigr) \ .
\]
The right-hand side of this bound  is at most $\exp(- c_j/16)$, by the Chernoff bound~\eqref{eqn:chernoff}, in which we substitute $n=c_j$, $\varepsilon=\sfrac{1}{2}$, and $r=\sfrac{1}{2}$.
\end{proof}


\Paragraph{Adversaries and $f$-chains.}

We make the following definitions towards finalizing the randomized analysis.

An infinite sequence of sets of processors $\cC=\langle V_k \rangle_{k\ge 0}$ is called an \emph{$f$-chain} if, for each $k\ge 0$:
\begin{enumerate}
\item[(a)] 
$V_0=V$ is the set of all processors and $V_{k+1}\subseteq V_k$, 
\item[(b)] 
$|V_k |\ge p-f$, 
\item[(c)]
either  $|V_k|=p-f$ or $|V_k| = \lceil  p/2^i \rceil$, for some $i\ge 0$.
\end{enumerate}

Recall that we denote by~$G_{l(j)}$ the set of non-faulty processors at the end of~$\cD_j$, for $j\ge 0$.
An adversary is said to be \emph{$f$-chain-constrained} when any execution of the algorithm, in which the adversary controls the timing of crashes, has the property that the sequence $\langle G_{l(j)}\rangle_{j\ge 0}$ of sets of processors is an $f$-chain.

The following Lemma~\ref{lem:probability-f-chain} extends the estimate on the probability of one key extended epoch to be productive, as given in Lemma~\ref{lem:bar-on-crashes} in the case when one extended epoch is restricted by a bar on crashes, to the estimate on the probability of an execution to include only productive key extended epochs, in the case when the adversary is $f$-chain constrained.


\begin{lemma}
\label{lem:probability-f-chain}

For any known $f<p$ and for sufficiently large~$p$, the probability that some key extended epochs in an execution of algorithm \aRP\ are not productive is~$\cO( p^{-1/6} \log p)$ when the algorithm is executed against the $f$-chain-constrained adversary. 
\end{lemma}

\begin{proof}
When the adversary is $f$-chain-constrained, then the sequence of sets of processors $\langle G_{m(j)}\rangle_{j\ge 0}$ that stay non-faulty through the end of an extended epoch $\cD_j=\langle \cE_{m(j)},\ldots, \cE_{l(j)}\rangle$ make an $f$-chain.
This means that the numbers in the corresponding sequence $\langle |G_{m(j)}|\rangle_{j\ge 0}$  are of a restricted form.
We partition key extended epochs into groups~$C_k$, for $0\le k\le \lceil \lg \frac{p}{p-f}\rceil$, determined by these numbers.
Namely, group~$C_k$ consists of these key extended epochs~$\cD_j$  for which $|G_{m(j)} |=\lceil p/2^k\rceil$, for $0\le k < \lceil \lg \frac{p}{p-f}\rceil$, and $| G_{m(j)} |=p-f$ in the case of  $k = \lceil \lg \frac{p}{p-f}\rceil$.

The core graph~$Z_j$, for an extended epoch~$\cD_j$ in~$C_k$, has at least $\frac{p}{7 \cdot 2^{k+1}}$ processors, by Lemma~\ref{lem:subgraph-property}.
It follows that the core number~$c_j$ is at least $\frac{p g(p)}{7 \cdot 2^{k+1}}$, even when the extended epoch~$\cD_j$ consists of just a single regular epoch.
The probability that an extended epoch~$\cD_j$ with a specific set~$G_{l(j)}$ is not productive is at most $\exp (-p g(p)/(7 \cdot 16\cdot 2^{k+1}))$, by  Lemma~\ref{lem:bar-on-crashes}.

If $\cD_j$ is in~$C_k$, for $0\le k\le \lceil \lg \frac{p}{p-f}\rceil$, then there are at most $2^{p/2^k}$ sets~$G_{l(j)}$.
By the union bound, the probability that an extended epoch in~$C_k$ is not productive is thus at most 
\begin{eqnarray*}
\label{eqn:union-bound}
2^{p/2^k} \exp \Bigl(-\frac{p g(p)}{7 \cdot 16\cdot 2^{k+1}}\Bigr)
&=&
\exp\Bigl(-\frac{p}{2^k}\Bigl(\frac{g(p)}{224} - \ln 2\Bigr)\Bigr)\\
&\le&
\exp(- g(p) /225)
\ ,
\end{eqnarray*}
for sufficiently large~$p$.
There are at most $1+\log p$ groups $C_k$, since $k \le  \lceil \lg \frac{p}{p-f}\rceil$. 
We can sum the right-hand sides of~\eqref{eqn:union-bound}, over all these~$k$, to obtain the following estimate
\begin{eqnarray*}
\cO(\log p) e^{- g(p) /225} 
&=& 
\cO(\log p) \, p^{- 30 / (226\cdot \ln 2)} \\
&<&
\cO(\log p) \, p^{-1/6}
\ ,
\end{eqnarray*}
that holds for sufficiently large~$p$.
This is a bound on the probability that some key extended epoch is not productive.
\end{proof}

\subsection{General adversaries and deterministic computation}

\label{sec:deterministic}

Next we show that the generic algorithm can be instantiated with fixed permutations at processors such that $\cO(t+p\log^2 p)$ is the worst-case bound on work.


\begin{lemma}
\label{lem:deterministic-work-f-chain-constrained-adversary}

Algorithm \aDP\  can be instantiated, for any known $f<p$, to have  work $\cO (t +p\log^2 p)$ against the $f$-chain-constrained adversary during main extended epochs.
\end{lemma}

\begin{proof}
Let us consider algorithm \aRP\ executed against the $f$-chain-constrained adversary.
The probability of the event that some key extended epochs in an execution of algorithm \aRP\ are not productive is bounded above by Lemma~\ref{lem:probability-f-chain}.

That bound  is less than~$1$ for sufficiently large~$p$, and then the probability that all key extended epochs are productive is positive.
By the principle of the probabilistic method, we obtain that there exist such permutations that, when algorithm \aDP\ is instantiated with these permutations, then all key extended epochs in any execution of this instantiation against the $f$-chain-constrained adversary are productive.
This allows to make use of Lemma~\ref{lem:key-epochs-productive} to obtain that the work  accrued by such an instantiation of algorithm \aDP\ against the $f$-chain-constrained adversary during main extended epochs  is~$\cO(t+p\log^2 p)$.
\end{proof}


\begin{lemma}
\label{lem:approximation-by-f-chain-constrained adversary}

Suppose that a deterministic instantiation of algorithm \aDP\  is such that if $p$ processors aim to perform $t$ tasks against an $f$-chain-constrained adversary, for a known $f<p$, then this is accomplished with at most $C(p,t)$ work and within $D(p,f,t)$ rounds during main regular epochs, for some functions $C(p,t)$ and~$D(p,f,t)$.
If this instantiation is executed against the general $f$-bounded adversary, then its work is $\cO(C(p,t)+p\log p)$ and time is $D(p,f,t)$ during main regular epochs.
\end{lemma}

\begin{proof}
Let us consider an execution $\cS$ of the given deterministic instantiation of algorithm \aDP\  against the general $f$-bounded adversary.
This adversary is adaptive and constrained by $f$ only as the bound on the number of crashes.
We will compare the main work performed in $\cS$ to the work performed during a specific  execution~$\cT$ of the algorithm against the $f$-chain-constrained adversary.

Recall that the $f$-chain-constrained adversary is restricted such that extended epochs~$\cD_j$ begin with the set of processors $K_{m(j)}$ whose sizes have the property that $| K_{m(j)} |=\lceil p/2^k\rceil$, for $0\le k < \lceil \lg \frac{p}{p-f}\rceil$, or $| K_{m(j)} |=p-f$.
For a given execution in which crashes are controlled by an $f$-chain-constrained adversary, 
let us call a natural number \emph{$f$-critical} (or just \emph{critical} if $f$ is clear from the context)
if it is the cardinality of a set of processors $K_{m(j)}$ in the execution.

The execution $\cT$ is obtained from~$\cS$ in such a way that processors in $\cT$ are crashing  in the same order as in $\cS$ and not later than in $\cS$, but some processors that crash in $\cS$ crash possibly earlier in~$\cT$ than in~$\cS$.
We maintain two invariants regarding timing of crashes.
One is that a processor that crashes in $\cT$ does it at the very end of some epoch.
The other is that the number of processors that crash by the end of each epoch of $\cT$ is a critical number.

We consider the consecutive regular epochs of~$\cS$ one by one to determine crashes that determine the corresponding epochs of~$\cT$.
We denote by $\cE_i$ and by $\cE'_i$ the $i$th epoch of $\cS$ and $\cT$, respectively.
Similarly, we denote by $G_i$ and $G_i'$ the sets of processors that are non-faulty at the end of epoch $\cE_i$ and~$\cE'_i$, respectively.
The first epoch $\cE_0$ of~$\cS$ consists of initialization, and this same epoch starts~$\cT$ as $\cE_0'$.
Suppose we have considered the epochs $\cE_0$ through $\cE_{j}$ of $\cS$ to determine the corresponding epochs  $\cE_0'$ through $\cE_{j}'$ of~$\cT$.
Now consider $\cE_{j+1}$.
If some processors crash in $\cE_{j+1}$ and they have not crashed in $\cS$ by the end of $\cE_{j}$ then we crash all these processors at the very end of $\cE_{j}'$ and additionally as many additional processors as needed so that the size of the set $G_j'$ is a critical number, but the minimum such a number.

All the processors terminate in the execution $\cS$ not later than in $\cT$ because in each round the processors that are non-faulty in $\cT$ make a subset of the processors that are non-faulty in $\cS$.
So if $D(p,f,t)$ is the bound on rounds in executions against $f$-chain-constrained adversary then it is a bound for the general adversary constrained by $f$, the bound on the number of crashes.

Moreover, if  epoch~$\cE_i$ is calm, then in any round of~$\cE_i$ the ratio of the cardinalities 
of these sets is at most~$4$; this is because the number of non-faulty processors decreases by at most a factor of~$2$ in~$\cE_i$ and the number of non-faulty processors in~$\cE_i$ drops below a critical number in a round of~$\cE_i$ then this contributes another factor of~$2$.
We conclude that the amount of work performed in~$\cS$ during calm epochs is at most four times 
the corresponding amount of work in~$\cT$.
The stormy epochs contribute $\cO(p\log p)$ work to any execution, by Lemma~\ref{lem:stormy-epochs}.
It follows that the main work in~$\cS$ is at most $4C(p,t) + \cO(p\log p)$.
\end{proof}

The following Theorem summarizes the properties of algorithm \aDP.


\begin{theorem}
\label{thm:deterministic-permutations}

Algorithm \aDP\ is a deterministic nonconstructive solution for  the \DA\ problem that, for any known numbers of $p$ processors, $t$ tasks, and bound on the number of crashes $f<p$, can be instantiated to have work $\cO (t +p\log^2 p)$.
\end{theorem}

\begin{proof}
Let us take an instantiation of the algorithm as provided by 
Lemma~\ref{lem:deterministic-work-f-chain-constrained-adversary}.
Its work against the $f$-chain-constrained adversary  during main extended epochs is always 
$\cO (t +p\log^2 p)$, as stated in Lemma~\ref{lem:deterministic-work-f-chain-constrained-adversary}.
By Lemma~\ref{lem:approximation-by-f-chain-constrained adversary}, we obtain that the amount of 
work accrued by this instantiation during main extended epochs against the general $f$-bounded adversary, constrained only by $f$ as a bound on the number of crashes, is also always $\cO (t +p\log^2 p)$.

Let us consider an execution against the general $f$-bounded adversary and its extended epochs after the main ones. 
We may restrict our attention only to calm extended epochs, as stormy extended epochs contribute $\cO(p\log p)$ work, by Lemma~\ref{lem:work-stormy-extended-epochs}.
Let $\cD_j$ be the first such extended epoch that is not main.
This means that some core processor starts executing procedure \texttt{Closing} during this extended epoch.
There are two cases, depending on whether $\cD_j$ is task-poor or task-rich.
We discuss them next.

If $\cD_j$ is task-poor, then it consists of just one regular epoch, say epoch $\cE_k$, by Lemma~\ref{lem:few-tasks-extended-epoch}.
All the core processors switch to performing closing phases by the end of $\cE_{k+1}$, by Lemma~\ref{lem:one-epoch-switch}.
These two epochs $\cE_k$ and $\cE_{k+1}$ contribute $\cO(p\log p)$ to work.
If $\cD_j$ is task-rich, then the inequality $c_j < u_{l(j-1)}$ holds.
The number $c_j$ is the amount of work spent by the core processors in $\cD_j$.
This work may be charged to the number of tasks $t$, because $u_{l(j-1)}\le t$.
All the core processors switch to performing closing phases by the end of the first regular epoch after the extended epoch~$\cD_j$, by Lemma~\ref{lem:one-epoch-switch}.
This one regular epoch contributes $\cO(p\log p)$ work.

Starting from this point, we consider contributions to work from the three kinds of  regular epochs.
We can conclude now, by the discussion above, that the total work during main regular epochs is $\cO (t +p\log^2 p)$.

Two consecutive mixed epochs are sufficient for every core processor to learn that there are no outstanding tasks, by Lemma~\ref{lem:one-epoch-switch}. 
We conclude that the total work during mixed epochs is $\cO(p\log p)$.

Next, we  estimate the work performed during closing regular epochs by applying Lemma~\ref{lem:closing-is-as-efficient-as-main}. 
We need first to take the amount obtained by substituting~$p$ for~$t$ in the estimate $\cO(t+p\log^2 p)$ on the work during main regular epochs, which gives $\cO(p\log^2 p)$.
Next, we add $\cO(p\log p)$ to it, which does not affect the asymptotic bound.
We conclude that the total work performed during closing regular epochs is $\cO(p\log^2 p)$. 

Overall, the contributions from all three kinds of regular epochs are  therefore $\cO (t +p\log^2 p)$.
\end{proof}

Observe that the bound on work given in Theorem~\ref{thm:deterministic-permutations} does not
depend on~$f$,  as it is also the case in Theorem~\ref{thm:work-balance-load}. 
These algorithms use the respective overlay graphs that depend on~$f$.
Hence the amount of communication the algorithms generate depends on~$f$, as estimated in Lemma~\ref{lem:communication-from-work} in Section~\ref{sec:design-algorithms}.


\begin{corollary}
\label{cor2}

Algorithm \aDP\ is a deterministic nonconstructive solution for  the \DA\ problem 
with $p$ processors and $t$ tasks that can be instantiated, for a number of crashes $f$ bounded by $f\le cp$,  
for any known constant $0<c < 1$, to have effort $\cO (t +p\log^2 p)$.
\end{corollary}

\begin{proof}
The argument is similar to that used in the proof of Corollary~\ref{cor1} in Section~\ref{sec:constructive}.
The difference is in referring to the bound in Theorem~\ref{thm:deterministic-permutations}.
\end{proof}

\section{An Algorithm Optimized for Effort} 

\label{sec:effort}

In this section we present an algorithm, called \aEP, that is optimized for effort complexity against the unbounded adversary.
This algorithm is structured to trade work for effort efficiency, and so it is not as work-efficient as \aDP, but it curbs the amount of communication when the number of crashes becomes large.
Recall that the communication complexity of \aDP\ is estimated using Lemma~\ref{lem:communication-from-work}, where the expression in the big-O bound can be as high as  $\omega(p^3)$ for $f$ sufficiently close to~$p$.

In optimizing effort, it is natural to balance the bounds on work and communication.
To do so, we employ a hybrid approach to algorithm design, where we first use \aDP\ for some duration, then we use some other algorithm to complete the remaining work.
We use  \aDP\ for as long as the number of failures~$f$ is sufficiently small, 
namely $f<f_1$ for a certain~$f_1<p$ that is used by the code.
Additionally, \aDP\ is modified so that its communication is comparable to its work, which is accomplished by having processors send messages only after  sufficiently many tasks have been performed.
We use~$T_1$ to denote the number of rounds that is  sufficient for algorithm \aDP\ to complete all tasks if there are at most~$f_1$ crashes in an execution; this number~$T_1$ is used in the code.
If first running \aDP\ for~$T_1$ rounds results in completion of all tasks, then we are done.
Otherwise we set stage for transitioning to another algorithm
by renaming the processors that are still non-faulty.
Specifically, after renaming we invoke the algorithm of De~Prisco et al.~\cite{PriscoMY94}, referred to as algorithm~\aDMY.


\begin{fact}[\cite{PriscoMY94}] 
\label{fact:DMY}

Algorithm \aDMY\ has work $\cO(p(f+1) + t)$ and message complexity $\cO(p(f+1))$.
\end{fact}

Renaming the non-faulty processors is necessary in order to use
the bounds given in Fact~\ref{fact:DMY} because algorithm \aDMY{}
assumes that the processors commencing the computation
are consecutively numbered (e.g., if there are  $k$ processes,
then they have unique identifiers from the interval $[1,k]$).

Specifying algorithm \aEP\ involves determining how $f_1$ and $T_1$ depend on~$p$ and~$t$. 
These parameters are chosen do that the work  in the invocation of \aDP\ is approximately 
the same as the work  in the invocation of \aDMY.
This balancing is done by means of the bounds given 
in Theorem~\ref{thm:deterministic-permutations} and Fact~\ref{fact:DMY}.
What we will obtain is that the parameter~$f_1$ needs to be close to~$p$, 
namely, $f_1=p(1-p^{-c})$, for a constant $0<c<1$ 
(the details will be given in Section~\ref{sec:justifying-initialization}).
It turns out that with as few as $p^{1-c}$ processors remaining non-faulty
at the invocation of algorithm~\aDMY{}  we can afford to restore the lists of tasks 
to their original form, thus redoing some of the tasks, 
without compromising the performance bounds.

\subsection{The design of  algorithm \aEP{}}

\label{sec:design-effort-priority}

Algorithm \aEP\ is structured in terms of three parts, named  \PartOne, \PartTwo, and \PartThree.
\PartOne\ executes a suitably modified algorithm \aDP{}.
\PartTwo\ provides a transition from to algorithm \aDMY{},
its purpose is to coordinate the renaming of the presumed active processes.
\PartThree\ executes algorithm \aDMY, with the (presumed) active processes 
referred to by their new names.
The three parts are given in more detail as follows.

\PartOne\ invokes the modified version of \aDP\ for up to $T_1$ rounds.
We call this modified algorithm \aMDP, and it is obtained as follows.
The algorithm uses the overlay graph~$L(p)^\ell$ with the threshold~$f_1$ playing the role of the number of crashes and with $\ell$ as determined by Lemma~\ref{lem:degree-power-of-Lp}; this means that $G(p,f_1)$ is used as the overlay graph.
We use parameter~$\Delta_1$ to denote the maximum degree of~$G(p,f_1)$.
The input set of tasks is partitioned into \emph{chunks}, with each chunk consisting of~$\Delta_1$ tasks.
Each such chunk is now treated as a single ``task,'' in the sense that messages are sent precisely after a chunk of tasks is performed.
After $T_1$ rounds of \aMDP, the algorithm continues its execution with \PartTwo.

\PartTwo\ begins with each processor~$v$ having only its list \texttt{Processors$_v$}  as 
the current approximation of the set of active processors.
The active processors need to agree on a set of processors that is simultaneously
a subset of the processors that finish \PartOne\ as active, and a superset of the set of 
processors that are still active when \PartThree\ begins; 
we call such a set of processors \emph{checkpoint}, following \cite{PriscoMY94}.
Finding a checkpoint is the goal of \PartTwo, which is accomplished as follows.

At the start of \PartTwo, every processor~$v$ that has not halted sends a message 
to each processor in its list \texttt{Processors$_v$}, including itself.
We refer to this round, and the messages circulated in it, as \emph{preparatory}.
Upon receiving the messages sent in the preparatory round, 
processor~$v$ removes from list \texttt{Processors$_v$} the identifiers of the processors 
from whom preparatory messages did not arrive.
Each processor~$v$ creates a new private list \texttt{Coordinators$_v$} and initializes 
it to \texttt{Processors$_v$} once the preparatory messages have been processed.
Each private copy of \texttt{Coordinators} is ordered in the increasing order of processor identifiers.
The next $4(p-f_1)$ rounds are structured in groups of four rounds; 
we call each such group the \emph{checkpointing phase}.
The pseudocode for the preparatory round and a checkpointing phase is given in Figure~\ref{fig:PartTwo}.


\begin{figure}[thp]
\rule{\textwidth}{0.75pt}
\begin{center}
\begin{minipage}{\pagewidth}

\begin{description}
\item[\sc Preparatory Round:]
Each processor $v$ sends a preparatory message to all processors in its list
\texttt{Processors}$_v$, including itself.

Upon receiving preparatory messages, processor $v$ removes from 
\texttt{Processors}$_v$ the identifiers of processors from whom
such messages are not received. The list \texttt{Coordinators}$_v$
is set to \texttt{Processors}$_v$, sorted in the order
of  identifiers.
\end{description}

\rule{\textwidth}{0.75pt}

~

{\sc Checkpointing Phase}

\begin{description}
\item[\sc Round 1:]
If $v$'s name occurs as the first entry in \texttt{Coordinators}$_v$ then 
$v$ sends a message to each processor in
 \texttt{Processors}$_v$, including itself, proposing to become coordinator.

\item[\sc Round 2:]

If $v$ does not receive coordinator proposal messages from Round~1, then this Round~2 is void.

Otherwise, let $v_1, \ldots v_k$ be the identifiers of the processors from which $v$ received coordinator proposals, sorted according to the identifiers, i.e., $v_1< \ldots <v_k$.
First, processor~$v$ responds to~$v_1$ by sending its  \texttt{Processors}$_v$ list.
Next, processor~$v$ adds any identifiers from $\{v_1, \ldots v_k\}$ that are not already in \texttt{Coordinators}$_v$ to \texttt{Coordinators}$_v$, maintaining the sorted order of entries.

If $v$ proposed itself in Round~1 as the coordinator, but $v$ is not the smallest identifier in \texttt{Coordinators}$_v$, then $v$ no longer considers itself coordinator for the current phase.

\item[\sc Round 3:]

If $v$ still considers itself coordinator for the phase then, upon receiving response messages sent in Round~2, 
$v$ removes from its  \texttt{Processors}$_v$ list the identifiers that are missing from any list \texttt{Processors}$_w$ just received.
Next $v$ sends its updated  \texttt{Processors}$_v$ list to each processor in \texttt{Processors}$_v$.

\item[\sc Round 4:]

Upon receiving list \texttt{Processors}$_w$ in Round~3 from some coordinator~$w$, 
processor~$v$ replaces its  \texttt{Processors}$_v$ list
with list \texttt{Processors}$_w$.

Processor~$v$ then removes the first entry in its  \texttt{Coordinators}$_v$ list.
If there are still some identifiers in \texttt{Coordinators}$_v$ that are smaller than 
both $v$ and the identifier of the processor that served as coordinator in Round~3, 
then $v$ removes these identifiers from \texttt{Coordinators}$_v$ as well. 
\end{description}
\end{minipage}

\FFF

\rule{\textwidth}{0.75pt}

\parbox{\captionwidth}{\caption{\label{fig:PartTwo}
Preparatory round and Checkpointing phase of \PartTwo\ of algorithm \aEP; the code for  processor~$v$.}}
\end{center}
\end{figure}

When a checkpointing phase begins, each processor~$v$ that is ranked first in its list \texttt{Coordinators}$_v$ broadcasts its proposal to serve as the phase coordinator to each processor in \texttt{Processors}$_v$.
It may happen that multiple processors send such messages in Round~1 of the phase, 
while we want to have only one coordinator.
If processor~$v_1$ proposes itself as coordinator in Round~1, but receives a proposal from processor~$v_2$, such that $v_1>v_2$, it does not consider itself coordinator for this phase;
otherwise $v_1$ stays as the coordinator for the whole phase.
A processor that acts as a coordinator through the whole phase is called the \emph{successful coordinator} for this phase.

In Round~2 of a phase, any processor~$v$ adds the identifiers of all proposers 
in the previous round to its list \texttt{Coordinators}$_v$, in case these identifiers are missing, and responds to the successful coordinator by sending a copy of its list \texttt{Processors}$_v$.

In Round~3, the successful coordinator~$v$ removes from \texttt{Processors}$_v$
any identifiers that are missing from any list \texttt{Processors} just received from other processors.
In case the successful coordinator~$v$ does not receive a message from some processor~$w$ whose name occurs in \texttt{Processors}$_v$ in the beginning of Round~3, processor $v$ does not remove $w$ from \texttt{Processors}$_v$ as long as $w$ occurs in each list \texttt{Processors}  received by~$v$ in this round.
This is done by the coordinator to ``stabilize'' the lists \texttt{Processors} to the eventual checkpoint, so that the ``stabilized'' lists are not modified even when additional crashes occur.
The successful coordinator~$v$ disseminates its updated list \texttt{Processors}$_v$ at the end of 
Round~3 among all processors with identifiers in \texttt{Processors}$_v$.

In Round~4, each recipient of a copy of the list \texttt{Processors} from the coordinator 
replaces its list \texttt{Processors} with this copy.
Next, each processor deletes the first entry on its list \texttt{Coordinators}.
The round may include additional deletions from these lists as follows.
It may happen that processor~$v$, that has not acted as a successful coordinator yet, 
receives in the first round of the phase a message from a would-be coordinator~$w$ such that $w>v$.
This indicates a discrepancy between the lists \texttt{Coordinators}$_v$ and \texttt{Coordinators}$_w$; 
in particular, processor~$w$ has already removed $v$ from \texttt{Coordinators}$_w$, 
while $w$ is still listed behind~$v$ in \texttt{Coordinators}$_v$. 
In this situation, processor~$v$ removes any identifier in \texttt{Coordinators}$_v$ ahead of its own identifier~$v$ to become the first on the list, subsequently proposing itself 
as the coordinator in the next phase.
The rationale for this is to help converge the lists \texttt{Coordinators}, while
deleting from them the identifiers of crashed processors.
This works because any active processor~$y$, that is ahead of $v$ in Round~1,
and whose identifier is deleted by~$v$ in Round~4 of this phase, 
will propose itself as coordinator along with~$v$ in Round~1 of the next phase, 
so both identifiers $v$ and~$y$ will be restored in the lists \texttt{Coordinators} 
in Round~2 of the next phase.

After  \PartTwo\ concludes, the processors proceed to  \PartThree, 
where algorithm \aDMY\  is invoked (given in~\cite{PriscoMY94}).
In \PartThree\ the algorithm reverts to using individual tasks instead of chunks of tasks.
Processor~$v$ cooperates only with the processors whose identifiers are in the list \texttt{Processors}$_v$  at the conclusion of  \PartTwo.
However, instead of using the original processor identifiers, each processor~$v$ refers to any active processor~$w$ in the list \texttt{Processors}$_v$ by $w$'s rank in the list.  
The lists \texttt{Processors} are not used by algorithm \aDMY; these lists only serve to
provide new processor identifiers.
Finally, the lists of tasks are reverted to their original contents of $t$ tasks.
This is sufficient to obtain the desired complexity bounds and  simplifies the analysis.
\nopagebreak
 
This concludes the specification of algorithm  \aEP.

\subsection{Correctness of algorithm \aEP{}}

We now address the correctness of the algorithm.
We first show that the set of processor identifiers in at least one list \texttt{Coordinators} 
shrinks in each checkpointing phase of \PartTwo.


\begin{lemma}
\label{lem:shrinking-coordinators}

The least identifier that occurs in some list \texttt{Coordinators} in  Round~1 of a phase is removed from all lists where it is present in Round~4 of the phase.
\end{lemma}

\begin{proof}
Consider a phase, and let $v$ be the least identifier occurring in any list \texttt{Coordinators} 
in the beginning of the phase; thus $v$ is the first entry in any list  
\texttt{Coordinators} it occurs in.
Suppose first that no processor announces itself coordinator in Round~1 of the phase.
Then no new entries are added to the lists during the phase and the first entries are removed 
from each list \texttt{Coordinators}; this eliminates~$v$ from any list it occurs in.
Now consider the case when some processors announce themselves coordinators in Round~1 of the phase.
If $v$ is not among them, then the proposed names of coordinators are greater than $v$, 
so no entry is added in Round~2 of the phase as new first entry to any list in which $v$ occurs.
Therefore all occurrences of $v$ remain first in the respective lists, 
and so they are removed in Round~4.
If, on the other hand, $v$ announces itself coordinator, then $v$ is added in Round~2 
as the new first entry to any list \texttt{Coordinators} where $v$ does not occur,  
only to be removed in Round~4 of the phase.
\end{proof}


\begin{lemma}
\label{lem:coordination-correct}

If there are fewer than $p-f_1$ active processors at the preparatory round of \PartTwo,  
then the lists \texttt{Processors} of all active processors become identical  
at some checkpointing phase in \PartTwo, and remain invariant through the remaining checkpointing phases.
\end{lemma}

\begin{proof}
We deal only with the processors that are active at the preparatory round of \PartTwo.
Let us assume that the list \texttt{Coordinators}$_v$ of any such processor~$v$ is modified  in Round~4 of each phase according to the specification in Figure~\ref{fig:PartTwo},  even after processor~$v$ crashes.
These modifications do not affect the course of execution, as processors do not send any messages  after crashing, but this simplifies the structure and exposition of the argument.

Let $v$ be a processor that stays non-faulty through the end of \PartTwo.
We claim that there is a phase in which $v$ acts as a successful coordinator.
To show this, observe the following three facts:
\begin{enumerate}
\item[(1)] the set of identifiers occurring in lists \texttt{Coordinators} shrinks with 
each checkpointing phase per Lemma~\ref{lem:shrinking-coordinators};
\item[(2)] there are fewer than $p-f_1$ identifiers in all lists \texttt{Coordinators} 
when checkpointing phases begin, by the assumption;
\item[(3)] there are $p-f_1$ checkpointing phases in total, by the code of  \PartTwo.  
\end{enumerate}
Combining these three facts  we get that eventually a phase~$P$ occurs such that the identifier of processor $v$ is removed from its own list \texttt{Coordinators}$_v$ in Round~4 of~$P$.
If $v$ is first on its list \texttt{Coordinators}$_v$ in Round~4 of~$P$ then $v$ was also first in Round~1 of~$P$, since no entries are removed from lists \texttt{Coordinators} in the first three rounds of a phase.
Therefore, processor~$v$ announces itself to be coordinator in Round~1 of phase~$P$.
Processor~$v$ is least among all processors that announce themselves coordinators for phase~$P$, as otherwise a smaller identifier would be added to \texttt{Coordinators}$_v$ in Round~2 of phase~$P$, and so $v$ would not be removed from \texttt{Coordinators}$_v$ in Round~4 of~$P$.
Therefore $v$ stays as coordinator for the entire phase~$P$.

Thus there exists a phase in which some processor acts as a successful coordinator.
Consider the first such phase, and let processor $v$ be the successful coordinator in this phase.
Every active processor~$w$ receives list \texttt{Processors}$_v$ from~$v$ in Round~4 of the phase.
Now, $w$ adopts the obtained list as its new list \texttt{Processors}$_w$, per Round~4 specification  (see Figure~\ref{fig:PartTwo}).
From this round on, all lists \texttt{Processors} of active processors stay invariant through the last checkpointing phase, as no element is ever removed from them, per Rounds~3 and~4 specification (see Figure~\ref{fig:PartTwo}).
\end{proof}


\begin{lemma}
\label{lem:correctness-effort-priority}

For any $f_1$ and $T_1$, such that \PartOne\ completes all tasks in  $T_1$ rounds 
when at most $f_1<p$ processors crash, algorithm \aEP\ that uses $f_1$ and $T_1$ solves the \DA\ problem against  the unbounded adversary.
\end{lemma}

\begin{proof}
The specification of Round~4 in Figure~\ref{fig:PartTwo}  refers to receiving a message from ``some coordinator,'' allowing for the possibility of multiple coordinators.
We show that there is at most one coordinator in Round~3 of any checkpointing phase 
in any execution of \PartTwo.
Suppose that there are two such coordinating processors with identifiers~$v$ and~$w$, and let $v<w$.
These two processors broadcast in Round~1 of the phase.
This results in processor $w$ restoring identifier $v$ in its list \texttt{Coordinators}$_w$ 
ahead of its own identifier~$w$, thereby losing the status of coordinator 
for the current phase, per Round~2 of Figure~\ref{fig:PartTwo}.

Next, we must ensure that the identifiers of processors and tasks are consistent among participating processors at the invocation of algorithm \aDMY.
This algorithm is invoked only when more than~$f_1$ processors crash while executing \PartOne, by the selection of~$T_1$.
Therefore when the checkpointing phases are over, the lists \texttt{Processors} 
are identical among the active processors per Lemma~\ref{lem:coordination-correct}.
This ensures that each active processor has the same rank in each such sorted list,
providing consistent new identifiers to all processors that execute algorithm \aDMY.

The lists of tasks are restored to their original form of $t$ entries at each processor just before algorithm \aDMY\ is invoked; this provides consistent referencing of the tasks  by the processors executing algorithm \aDMY.
Since algorithm \aDMY\ performs all tasks (Fact~\ref{fact:DMY}) in \PartThree, algorithm \aEP\ solves the \DA\ problem as required.
\end{proof}

\subsection{Specifying algorithm parameters}

\label{sec:justifying-initialization}

Next we specify the numeric parameters that are present in the code of algorithm \aEP. 
The idea is to select the parameter $f_1$ and the corresponding~$T_1$ so as to make the work of \PartOne\ approximately equal to the work of \PartThree.
Let $\rho = 27/2$  and~$\Delta_0=74$ be the constants used in the specification of the overlay graphs in Section~\ref{sec:graphs}. 
Choose $a$ such that $0<a<1$ and 
\begin{equation}
\label{eq:a-gives-work}
a<\frac{1}{1+2\log_\rho \Delta_0}\ .
\end{equation}
We set $f_1=p-p^{1-a}$ as the bound on the number of crashes. 
From this we get  $p/(p-f_1)= p^a$.
By Lemma~\ref{lem:degree-power-of-Lp}, the degree of the corresponding overlay graph~$L(p)^\ell$ satisfies
\begin{equation}
\label{eqn:Delta-for-Mod-Perm}
\Delta_1
=
\cO\Bigl(\Bigl(\frac{p}{p-f_1}\Bigr)^{2\log_\rho \Delta_0}\Bigr)
=
\cO(p^{2a\log_\rho \Delta_0})\ .
\end{equation}
We instantiate $T_1$ to be sufficiently large for Lemma~\ref{lem:correctness-effort-priority} 
to be applicable, but such that
\begin{equation}
\label{eqn:T_1}
T_1=\cO\Bigl(\frac{t+p}{p-f_1} + \Delta_1\log^2 p\Bigr)\ .
\end{equation}
We next show that it is possible to set $T_1$ so that bound~\eqref{eqn:T_1} is satisfied.
This is ultimately accomplished in Lemma~\ref{lem:justification-of-T1}.

The first step is computing the time needed to perform tasks when the key extended epochs are productive.
To this end, we revisit the proof of Lemma~\ref{lem:key-epochs-productive} and verify its time aspects.
Recall the following definitions given in Section~\ref{sec:extended-epochs}.
The integers $j$ that are indices of main extended epochs~$\cD_j$ make a contiguous interval~$\cJ$ that is represented in the form $\cJ= J_1\cup J_2\cup J_3\cup J_4\cup J_5$.
Set $J_1$ consists of indices $j$ of stormy extended epochs~$\cD_j$, set $J_2$  consists of these $j\in \cJ$  for which $u_{\ell(j-1)} \ge 11p^2 g(p)$, set $J_3$ consists of the index of the extended epoch $\cD_j$ such that $u_{m(j)} \ge 11p^2g(p)$ and $u_{l(j)} < 11p^2g(p)$, set $J_4$ consists of these $j \in \cJ \setminus (J_1\cup J_2\cup J_3)$ for which the (key) extended epoch~$\cD_j$ is task-rich, and set $J_5$ consists of these $j\in \cJ \setminus (J_1\cup J_2\cup J_3)$ for which the (key) extended epoch~$\cD_j$ is task-poor.


\begin{lemma}
\label{lem:from-productivity-to-time}

For any known $f<p$, all  processors executing the generic algorithm halt after performing $\cO\bigl(\frac{t}{p-f}  +  \log^2 p\bigr)$  rounds during the main extended epochs~$\cD_j$, if only all key extended epochs among them are productive.
\end{lemma}

\begin{proof}
We need to sum up the durations of~$\cD_j$, for~$j\in \cJ$.
We estimate this sum by considering the indices of the extended epochs that fall into the four subsets $J_1$, $J_2$, $J_3$, $J_4$, and $J_5$.

First we consider the set~$J_1$.
There are $\cO(\log p)$ extended epochs with indices in~$J_1$ because each of them consists of a single regular stormy epoch per Lemma~\ref{lem:stormy-extended-epoch}.
So their contribution to the total time is $\cO(\log^2 p)$. 
Next consider extended epochs~$\cD_j$ for $j\in J_2$.
The time spent in the extended epochs whose regular epochs~$\cE_i$ satisfy $u_i \ge 11p^2g(p)$ 
is $\cO(\frac{t}{p-f})$ by Lemma~\ref{lem:no-overlap-in-tasks}.

The key extended epochs~$\cD_j$ have their indices~$j$ in~$J_4\cup J_5$.
We assume that each of them is productive.
For the extended epochs~$\cD_j$ such that $j\in J_4$ we use the estimate \eqref{eqn:many-tasks}
that is derived in the proof of Lemma~\ref{lem:key-epochs-productive} (see Section~\ref{sec:extended-epochs}).
It follows that the decrease in the number of tasks in two consecutive extended epochs is proportional to the work expended by the core processors, this work being~$c_j$, and so this takes $\cO(\frac{t}{p-f})$ time. 
There are  $\cO(\log p)$ extended epochs $\cD_j$ for $j\in J_5$, by~Lemma~\ref{lem:J-4}, 
since each such $\cD_j$ is productive.
Each such extended epoch $\cD_j$ for $j\in J_5$ consists of a single regular epoch
by Lemma~\ref{lem:few-tasks-extended-epoch}.
The time contribution of such extended epochs is~$\cO(\log^2 p)$.

There is at most one extended epoch~$\cD_j$ such that $u_{m(j)} \ge 11p^2g(p)$ and $u_{l(j)} < 11p^2g(p)$; if it exists then $J_3=\{j\}$.
Observe that the duration of this extended epoch is at most the time of all the extended epochs whose indices are in $J_1\cup J_2\cup J_4\cup J_5$.
This means that it suffices to double the time derived for the extended epochs whose indices are in $J_1\cup J_2\cup J_4\cup J_5$ to cover the case of this extended epoch.

The time of the main extended epochs~$\cD_j$ such that $j\in J_1\cup J_2\cup J_4\cup J_5$ is $\cO\bigl(\frac{t}{p-f}  +  \log^2 p\bigr)$, if only all key extended epochs among them are productive.
It follows that $\cO\bigl(\frac{t}{p-f}  +  \log^2 p\bigr)$ is the bound on time of all the main extended epochs with this property.
\end{proof}

The next step is computing the time needed to perform the closing work.
To this end, we revisit the proof of Lemma~\ref{lem:closing-is-as-efficient-as-main} (see Section~\ref{sec:epoch-core-processors}) to extract time bounds.


\begin{lemma}
\label{lem:time-of-closing}

Suppose an instantiation of the generic algorithm has the property that if $p$ processors aim to perform $p$ tasks and up to $f$ of them may fail, for a known~$f<p$, then this is accomplished within $\cO(T(p,f))$ main rounds, for some function~$T(p,f)$.
Then the algorithm terminates after $\cO(T(p,f)+ \log^2 p)$ closing rounds.
\end{lemma}

\begin{proof}
There are $\cO(\log p)$ stormy regular epochs, so they contribute $\cO(\log^2 p)$ to time.
Next we consider only calm epochs.
By the definition of~$T(p,f)$, core processors perform~$\cO(T(p,f))$ closing rounds 
while each core processor~$v$ has a nonempty \texttt{Busy}$_v$ list.
When some processor~$v$, that is core for~$\cE_k$, obtains an empty list \texttt{Busy}$_v$ in~$\cE_k$, then $v$ sends copies of this empty list to its neighbors, and halts.
All  processors that are core for epoch~$\cE_{k+1}$ halt by the end of~$\cE_{k+1}$.
These two epochs, $\cE_k$ and~$\cE_{k+1}$, contribute $\cO(\log p)$ to time.
It follows that the core processors halt within $\cO(T(p,f)+ \log^2 p)$ closing rounds.

It may happen that after all core processors halt, there are still compact processors with nonempty \texttt{Busy} lists.
In this case we may carve compact graphs from the subgraph of $G(p,f)$ induced by these nodes.
The nodes of any such a graph terminate within time bounds we obtained for the core processors, because the notion of core processor is only conceptual, used for the sake of arguments, while no processor knows whether it is core or not.
Let us choose such a compact subgraph of $G(p,f)$ that stays compact longest; these could be core processors, without loss of generality.
The processors in this subgraph halt within $\cO(T(p,f)+ \log^2 p)$ closing rounds.
Afterwards, at most two epochs occur such that by their end no processor considers itself compact, 
thus all processors halt in these epochs.
This again contributes only $\cO(\log p)$ additional time.
\end{proof}

We revisit the proof of Theorem~\ref{thm:deterministic-permutations} to examine time performance, for the same instantiation that provides good work performance.


\begin{lemma}
\label{lem:time-deterministic-permutations}

For any known $f<p$, algorithm \aDP\ can be instantiated so that its work is $\cO(t+p\log^2 p)$, with all processors halting in time $\cO\bigl(\frac{t+p}{p-f}  +  \log^2 p\bigr)$.
\end{lemma}

\begin{proof}
We take an instantiation of the algorithm as provided by Lemma~\ref{lem:deterministic-work-f-chain-constrained-adversary}.
The bound on work  given in Theorem~\ref{thm:deterministic-permutations} applies.
Next we consider  time complexity.
The algorithm instantiation has the property that all key extended epochs in any execution against the $f$-chain-constrained adversary are productive.
This lets us use Lemma~\ref{lem:from-productivity-to-time} that gives 
$\cO\bigl(\frac{t}{p-f}  +  \log^2 p\bigr)$ as the time spent by such an instantiation of algorithm \aDP{} during main extended epochs  against the $f$-chain-constrained adversary.
By Lemma~\ref{lem:approximation-by-f-chain-constrained adversary}, it takes $\cO\bigl(\frac{t}{p-f}  +  \log^2 p\bigr)$ time for this instantiation to complete the main epochs against the general $f$-bounded adversary.
We use Lemma~\ref{lem:time-of-closing} to estimate the number of closing rounds.
This contribution is obtained by adding $\cO(\log^2 p)$ to the amount obtained by substituting~$p$ for~$t$ in the estimate $\cO\bigl(\frac{t}{p-f}  +  \log^2 p\bigr)$ that we have  obtained for the main rounds.
These two contributions together make up $\cO\bigl(\frac{p}{p-f}  +  \log^2 p\bigr)$ closing rounds.
\end{proof}

The complexity bounds of algorithm \aMDP\ can be obtained by resorting to the bounds already derived for algorithm \aDP.
To this end we apply the following \emph{scaling}.
We substitute the number of chunks for the number of tasks and modify the time scale by 
extending the time of each round by the time sufficient to perform one chunk. 
Observe first that the number of chunks is $\cO(t/\Delta_1)$ and a chunk is performed in~$\Delta_1$ rounds. 
This is used in the proof of the next  lemma.


\begin{lemma}
\label{lem:justification-of-T1}

Algorithm \aEP\ can be initialized with $T_1$ satisfying Equation~\eqref{eqn:T_1} so that all processors halt by round $T_1$ while executing  \PartOne\, provided that the number of crashes in the execution is at most~$f_1$.
\end{lemma}

\begin{proof}
We estimate the time by which all tasks are performed when there are $f_1$ crashes.
To this end we apply scaling of tasks to chunks.
Partitioning the tasks into chunks of size $\Delta_1$ scales down the term $t+p$ in Lemma~\ref{lem:time-deterministic-permutations} to $(t+p)/\Delta_1$.
We substitute this term for $t+p$ in the bound on time in Lemma~\ref{lem:time-deterministic-permutations}, and simultaneously extend the time bound by a factor of~$\Delta_1$ that is 
due to the time sufficient to complete a chunk, to obtain 
\[
\cO\Bigl(\Delta_1\Bigl(\frac{t+p}{\Delta_1(p-f_1)}  +  \log^2 p\Bigr)\Bigr) 
=
\cO\Bigl(\frac{t+p}{p-f_1} + \Delta_1 \log^2 p\Bigr)
\]
as a bound on the number of rounds by which all processors halt.
\end{proof}

This concludes the justification that $T_1$ can be initialized such that 
Lemma~\ref{lem:correctness-effort-priority} is applicable and estimate~\eqref{eqn:T_1} is satisfied.

\subsection{The bound on effort}

\label{sec:bound-effort}

We proceed with the analysis of effort for algorithm \aEP.
The parameters in the code of the algorithm are set as specified 
in the beginning of Section~\ref{sec:justifying-initialization}.


\begin{lemma}
\label{lem:effort-part-one}

The effort of \PartOne\  is $\cO(t+\Delta_1  p\log^2 p)$.
\end{lemma}

\begin{proof}
Each processor executing the algorithm sends $\Delta_1$ messages precisely after completing $\Delta_1$ tasks.
Thus  communication complexity is the same as work complexity,
and so it suffices to estimate work.
We consider two cases depending on the number of processors staying non-faulty until round~$T_1$. 

\noindent
Case 1: the number of such non-faulty processors is at least~$p-f_1$. 
All  processors  halt by round~$T_1$ executing  \PartOne{} per Lemma~\ref{lem:justification-of-T1}.
Their work is as estimated in Theorem~\ref{thm:deterministic-permutations}, 
by Lemma~\ref{lem:time-deterministic-permutations}, 
suitably adjusted by scaling.
Applying the scaling, we obtain that the work during \PartOne\ is  
\begin{equation} 
\label{eqn:chunks}
\cO\Bigl(\Delta_1\Bigl(\frac{t}{\Delta_1}+p\log^2 p\Bigr)\Bigr)=\cO(t+\Delta_1 p\log^2 p) \ .
\end{equation}

\noindent
Case 2: the number of such non-faulty processors is less than $p-f_1$. 
Since there are few surviving processors, algorithm \aMDP\ used in \PartOne\ 
may not complete all tasks.
Using the same derivation as for bound~\eqref{eqn:chunks},
the effort is $\cO\bigl(t+\Delta_1 p\log^2 p\bigr)$ up to the round when the number of active processors
falls below $p-f_1$.
The contribution to work of the remaining rounds of \PartOne{} is 
\begin{eqnarray*}
\cO((p-f_1) \, T_1)
&=&
\cO\Bigl((p-f_1) \Bigl(\frac{t+p}{p-f_1} + \Delta_1 \log^2 p\Bigr) \Bigr)\\
&=&
\cO(t+p+\Delta_1\, (p-f_1)\log^2 p)\\
&=&
\cO(t+\Delta_1 \, p\log^2 p)
\end{eqnarray*}
by equality~\eqref{eqn:T_1} and Lemma~\ref{lem:justification-of-T1}.
\end{proof}

Next we deal with  message complexity for \PartTwo.
We begin by discussing a situation in which a processor announces itself coordinator in a checkpointing phase and next abandons this status within the same phase.
This generates $\cO(p)$ messages each time for each such processor.
In these scenarios, some processors $v_1,\ldots,v_k$, where $k>1$ and $v_1 < \ldots < v_k$, propose, in Round~1 of some checkpointing phase~$P$, that each of them intends to become coordinator for the phase.
Then $v_1$ remains as the successful coordinator for phase~$P$, while $v_2,\ldots,v_k$ give up the status of coordinator for phase~$P$.
Observe that each processor~$v_i$, among $v_1,\ldots, v_k$, has the identifiers $v_1,\ldots, v_k$ in its list \texttt{Coordinators}$_{v_{i}}$ just after the preparatory round, 
because these processors are still non-faulty later in checkpointing phase~$P$.
Still, each processor~$v_i$ appears first on its own list \texttt{Coordinators}$_{v_{i}}$ 
in the first round of phase~$P$.
An explanation of what must have occurred prior to phase~$P$ is as follows, where, for simplicity, we assume that no crashes after the preparatory round affect the considered scenario.

Each processor~$v_i$, for $1<i\le k$, must  delete $v_1,\ldots,v_{i-1}$ from list \texttt{Coordinators}$_{v_{i}}$ prior to phase~$P$ to become first in  \texttt{Coordinators}$_{v_{i}}$ in Round~1 of phase~$P$, as $v_j <v_i$ for $1\le j<i$.
Additionally, any processor~$v_j$, for $j<i$, deletes some $i-j$ entries from its list \texttt{Coordinators}$_{v_j}$ (to become first in this list) that $v_i$ does not need to delete (to become first in its own list \texttt{Coordinators}$_{v_i}$),  to make up for the entries $v_j, \ldots, v_{i-1}$ that $v_j$ does not need to delete (in order  to become first in its list \texttt{Coordinators}) as compared to~$v_i$.
Specializing this argument to $i=k$ and $j=1$, we obtain that processor~$v_1$ deletes $k-1$ processors that $v_k$ does not delete, prior to phase~$P$.
If $x_1,\ldots, x_{k-1}$ are some $k-1$ processors that $v_1$ deletes by the first round of phase~$P$, but $v_k$ does not, then $x_i<v_1$, for $1\le i\le k-1$, as $v_1$ removes $x_i$ from \texttt{Coordinators}$_{v_1}$ prior to announcing itself coordinator. 
Therefore, these are the discrepancies among their lists \texttt{Coordinators} that make  processors $v_1,\ldots,v_k$ broadcast simultaneously in Round~1 of phase~$P$.
Next we explain how such disparities could have occurred.

The preparatory messages sent by each $x_i$, for $1\le i\le k-1$, are received by~$v_1$, since $v_1$ removes~$x_i$ from \texttt{Coordinators}$_{v_1}$, so $x_i$ is placed in list \texttt{Coordinators}$_{v_1}$ at the end of the preparatory round.
On the other hand, the preparatory messages sent by each $x_i$ are not received  by~$v_k$.
This is because otherwise $x_i$ would be in the list \texttt{Processors}$_{v_k}$ after the preparatory round, and therefore also in list \texttt{Coordinators}$_{v_k}$, so that $v_k$ would need to remove $x_i$ eventually from \texttt{Coordinators}$_{v_k}$ to become first in this list, as $x_i<v_1<v_k$, while we know that $v_k$ does not do it.
It follows that each processor~$x_i$, for $1\le i\le k-1$, crashes in the preparatory round, 
as some of the messages that $x_i$ was to send in that round are not received.
We refer to such a scenario of broadcasts in Round~1 of phase~$P$ by saying that 
\emph{processor~$v_1$ is delayed with respect to~$v_k$ by $k-1$ crashes}, where these crashes occur in the preparatory round.
We also say that the broadcasts of processors $v_2,\ldots,v_k$ in Round~1 of checkpointing phase~$P$ are \emph{futile broadcasts} in this phase.
When the identifiers of processors $v_2,\ldots,v_k$ are added in Round~4 of a phase to each list \texttt{Coordinators}, 
in which any of them are absent at that point, then we say that the crashes of processors $x_1,\ldots, x_{k-1}$ become \emph{neutralized} thereby.
 
Now, we are ready to tackle the message complexity of \PartTwo.
This complexity is estimated by careful accounting of futile broadcasts caused by delays due to crashes on one hand, and by neutralizing the same number of crashes, on the other.

Let $p_2$ denote the number of processors that stay non-faulty and do not halt by round~$T_1$, that is, they start executing \PartTwo.


\begin{lemma}
\label{lem:communication-prior-DMY}

There are $\cO(p\,(p-f_1) )$ messages sent during \PartTwo.
\end{lemma}

\begin{proof}
It is sufficient to consider only the case when more than $f_1$ processors crash during \PartOne, as otherwise \PartTwo{} is not invoked at all.  
We have $p_2 < p-f_1$, so that $p_2\cdot p=\cO(p\,(p-f_1))$.

The preparatory round contributes $\cO(p_2\, p)$ messages, 
since at most $p_2$ processors broadcast messages, each to at most $p$ recipients.
Next we consider the messages generated in the checkpointing phases.
Any processor~$v$ acts as a successful coordinator at most once, 
since afterwards $v$  removes its identifier from list \texttt{Coordinators}$_v$ 
to never announce itself as a coordinator again.
A processor that acts as a successful coordinator for a phase broadcasts messages twice, 
in Round~1 and Round~3 (rounds refer to Figure~\ref{fig:PartTwo}), 
and collects messages once, in Round~2.
Therefore such activity contributes $\cO(p_2 \cdot p_2)=\cO(p_2\, p)$ messages.

To account for futile broadcasts, we charge these broadcasts to some crashes that occur either in the preparatory round, or in subsequent checkpointing phases.
For this accounting to work, it is sufficient to assure that at most one futile broadcast is charged to any one crash.
The next part of the proof deals with showing this fact.
The argument is broken into two major parts, named Claims~1 and~2, 
based on whether crashes occur only in the preparatory round or also after it.

We define \emph{configuration} to be a snapshot of the contents of  all lists \texttt{Coordinators} just after the preparatory round.
A configuration is determined by crashes in the preparatory round and the messages that are to be sent by the processors that crash in the preparatory round; 
each such message may be sent and delivered, or sent and not delivered, or not sent at all.
We consider executions of \PartTwo{} in which no  crashes occur, except in the preparatory round. 
For any configuration~$\cC$ that may come into existence at the end of the preparatory round as the result of such an execution, we define $\cE(\cC)$ as the fragment of this execution during the rest of \PartTwo{}.
Since no new crashes occur in \PartTwo{}, execution~$\cE(\cC)$ is unique with respect to configuration~$\cC$.

\begin{quote}
\textbf{Claim 1:} 
For any configuration $\cC$, each futile broadcast in~$\cE(\cC)$ can be charged to a unique crash in the preparatory round.
\end{quote}
If there are  futile broadcasts in~$\cE(\cC)$, then let $P_\cC$ be the first round in which simultaneous broadcasts by multiple processors occur in~$\cE(\cC)$.
We consider two cases for execution~$\cE(\cC)$: 
(1)~there is a single coordinator that broadcasts in Round~1 of the previous phase~$P_\cC-1$, and
(2)~there is no coordinator in phase~$P_\cC-1$.
The case when $P_\cC$ is the very first checkpointing phase in~$\cE(\cC)$, 
i.e., there is no checkpointing phase~$P_\cC-1$, 
falls under case (2) for~$\cE(\cC)$.
Next, we consider these cases in detail.

\noindent
{\bf Case 1:}
For execution~$\cE(\cC)$ in this case, let $v$ be the coordinator of phase~$P_\cC-1$. 
The processors that broadcast simultaneously in Round~1 of the next phase~$P_\cC$ 
include all processors whose identifiers are less than~$v$, per Round~4, and possibly processors with identifiers greater than $v$.
Let $w_1,\ldots, w_\ell$ be those processors with identifiers less than $v$, 
where $w_1<\ldots<w_\ell$, and let $v_1,\ldots,v_k$ be those processors with
identifiers greater than~$v$.
Processor~$w_1$ acts as coordinator for this phase, 
and in the following $\ell -1$ phases processors $w_2,\ldots,w_\ell$ act as coordinators one by one.
Following this, all processor identifiers in lists \texttt{Coordinators} are greater than~$v$.
Processors $v$ and $v_1,\ldots,v_k$, where $v<v_1<\ldots<v_k$, may not
constitute a segment of consecutive identifiers of processors that are still active.
 If there are some additional $j$ active processors whose identifiers are between $v$ and $v_k$, then this means that $w_1$ is delayed with respect to~$v_k$ by $k+\ell+j-1$ crashes.
It is also possible that some processors delete identifiers greater than~$v_k$ from their lists \texttt{Coordinators} prior to phase~$P_\cC$.
If $m$ is the maximum number of such identifiers deleted from the list of some processor~$z$, then $w_1$ is delayed by at least $k+\ell+j+m-1$ crashes with respect to this processor~$z$.
We have already used $k+\ell-1$ crashes to account for futile broadcasts 
by $w_2,\ldots,w_\ell$ and by $v_1,\ldots,v_k$.
The identifiers of $w_2,\ldots,w_\ell$ and $v_1,\ldots, v_k$ are restored in phase~$P_\cC$ to any \texttt{Coordinators} list in which some of them were missing, per Round~2 specification, neutralizing these $k+\ell-1$ crashes.
After this restoration, processor~$w_1$ is still delayed by at least $j+m$ crashes with respect to processor~$z$. 
This means that at most $j+m$ futile broadcasts may occur in~$\cE(\cC)$ after phase~$P_\cC$, and any such broadcast can be charged to a unique crash in the preparatory round.

\noindent
{\bf Case 2:}
For execution~$\cE(\cC)$ in this case, there is no coordinator in phase~$P_\cC-1$.
Let processors $v_1,\ldots, v_k$ broadcast together in Round~1 of phase~$P_\cC$. 
Processor~$v_2$ also broadcasts in Round~1 of phase~$P_\cC+1$.
Some other processors $w_1,\ldots,w_\ell$, where $w_1<\ldots<w_\ell < v_1$, 
may become first in their lists \texttt{Coordinators} in phase~$P_\cC$, 
per Round~4 specification.
These processors also broadcast in Round~1 of phase~$P_\cC+1$.
If this is the case, then $w_1$ is delayed with respect to crashes of~$v_k$ by $k+\ell -1$.
In the next $\ell -1$ phases of execution~$\cE(\cC)$, processors $w_2,\ldots, w_\ell$,  
that perform futile broadcasts in phase~$P_\cC+1$, act as coordinators.
Following this, all identifiers in lists \texttt{Coordinators} are greater than~$v_1$.
The identifiers of processors $v_1,\ldots,v_k$, where $v_1<\ldots<v_k$, may not form a segment of consecutive identifiers.
If there are some additional $j$ active processors whose identifiers are between $v_1$ and $v_k$, then this means that processor $w_1$ is delayed with respect to $v_k$ by $k+\ell+j-1$ crashes.
It is also possible that some processors  delete identifiers greater than $v_k$ 
from their lists \texttt{Coordinators} prior to phase~$P_\cC$.
If $m$ is the maximum number of such identifiers deleted from the list of some processor~$z$, then $w_1$ is delayed by at least $k+\ell+j+m-1$ crashes with respect to this processor~$z$.
We have already used $k-1$ crashes to account for the futile broadcasts by  $v_2,\ldots, v_k$, and $\ell$ crashes for similar broadcasts by processors $w_2,\ldots,w_\ell$ and~$v_2$.
The identifiers of processors $v_2,\ldots,v_k$ are restored in all lists \texttt{Coordinators} 
in phase~$P_\cC$, and the identifiers of processors $w_2,\ldots,w_\ell$ and~$v_2$ 
are restored in all lists \texttt{Coordinators} in phase~$P_\cC+1$, per Round~2 specification, which neutralizes $k+\ell-1$ crashes.
After this restoration, processor~$w_1$ is still delayed by at least $j+m$ crashes with respect to processor~$z$. 
This means that at most $j+m$ futile broadcasts may occur in~$\cE(\cC)$ after phase~$P_\cC$, and any such broadcast can be charged to a unique crash in the preparatory round.
This completes the proof of Claim~1.

Let us associate  \emph{credits} with configuration~$\cC$.
We define the number of credits to be equal to $k+\ell+j+m-1$, where these integers are determined by $\cE(\cC)$.
We will use credits to account for futile broadcasts.
In particular, $j+m$ credits that are still available after phase~$P_\cC$, when execution~$\cE(\cC)$ 
falls under Case 1, or after phase~$P_\cC+1$ in Case 2, respectively, are sufficient to account for futile broadcasts after phase~$P_\cC$ or $P_\cC+1$, respectively, in execution~$\cE(\cC)$, as at most $j+m$ of them may possibly  occur.

In a general scenario, crashes after the preparatory round may also contribute to multiple simultaneous broadcasts.
Crashes during checkpointing phases matter only when they occur in Round~1 of a checkpointing phase, because discrepancies among lists \texttt{Coordinators} may be created when some recipients receive messages sent in these rounds and some not. 
In what follows we restrict our attention only to crashes of some processors~$x$ when these crashes have one of the following properties: 
(1) processor~$x$ crashes  while broadcasting messages in Round~1 of some phase, and
(2) some recipients of a message sent by~$x$ receive the message and some do not.
Now, consider arbitrary executions of \PartTwo.
\begin{quote}
\textbf{Claim 2:} 
The following invariant is maintained through all checkpointing phases~$P$: 
the number of credits available in the beginning of~$P$ is sufficiently large to account for 
all futile broadcasts in~$P$ and thereafter if crashes occur only prior to phase~$P$.
\end{quote}
The proof of the invariant is by induction on the phase number~$P$.
Let an execution of \PartTwo{} begin with configuration~$\cC$.
The basis of the induction deals with the case when $P=1$, which means that no crashes occur after the preparatory round; this is taken care of by Claim~1.
Next we show the inductive step.
Let $P$ be a phase number, where $P>1$, and assume that the invariant is established by phase $P-1$.
If there are no crashes in~$P$, then the invariant extends automatically, so suppose at least one crash occurs in~$P$.
There are two cases, depending on whether there is a single or multiple broadcasts  in Round~1 of phase~$P$.

\noindent
{\bf Case 1, single broadcast:}
There is one coordinator~$v$ that broadcasts in Round~1 of~$P$, and $v$ crashes 
in this round so that some processors receive a message from~$v$ and some do not.
There are two sub-cases: 
(1.a) a successful broadcast by $v$ results in some processors removing at least two entries from their lists \texttt{Coordinators} in Round~4, and 
(1.b) this does not happen.

The latter sub-case (1.b) is such that there is no difference whether $v$ crashes or not, as every processor removes just the first entry from its list \texttt{Coordinators}, so the inductive assumption is applicable directly.
It remains to consider the former sub-case (1.a) that may result in multiple simultaneous broadcasts in phase~$P+1$.

Suppose that if $v$ does not crash, then $w_1,\ldots, w_\ell$ and $v_1,\ldots,v_k$ broadcast simultaneously in Round~1 of phase~$P+1$, where the identifiers $w_1,\ldots, w_\ell$ are less than~$v$, and the identifiers $v_1,\ldots,v_k$ are greater than~$v$.
(This situation is similar to the one we considered before for executions of 
the form $\cE(\cC)$, when there is a single coordinator in phase $P_\cC-1$ prior 
to phase $P_\cC$ with multiple broadcasts.)
By the inductive assumption, there are enough credits to pay for the futile broadcasts among those by the processors $w_1,\ldots, w_\ell$ and $v_1,\ldots,v_k$.
Each of the processors $v_1,\ldots,v_k$ removes just the first entry from its respective 
list \texttt{Coordinators}, which is not the result of the crash by~$v$, 
so we account for all their futile broadcasts using $k$ of the available credits.
The processors $w_1,\ldots, w_\ell$ are those that remove from \texttt{Coordinators} 
some identifiers that are less than both~$v$ and their own name, per Round~4 specification.
A crash of $v$ may result in some among $w_1,\ldots, w_\ell$ not receiving 
the message of~$v$ in phase~$P$, and so they do not shorten their lists \texttt{Coordinators}, 
and do not broadcast at all in phase~$P+1$.
But if this occurs, then we will use correspondingly fewer credits to account 
for the broadcasts of processors among $w_1,\ldots, w_\ell$ in the next phase $P+1$, 
so the accounting breaks even.

\noindent
{\bf Case 2, multiple broadcasts:}
Here some processors perform futile broadcasts in phase~$P$ and some of those that broadcast crash in the process.
(This situation is similar to the one we considered before for executions 
of the form $\cE(\cC)$, when there were multiple broadcasts by would-be coordinators 
in phase $P_\cC$, while there was no coordinator in the immediately preceding phase~$P_\cC-1$.)
Suppose that $v_1,\ldots, v_k$ broadcast together in Round~1 of phase~$P$, 
and some of them crash in the process.
If the crashes did not occur, then processor~$v_2$ would broadcast in Round~1 of phase~$P+1$, and some other processors $w_1,\ldots,w_\ell$ might become first in their lists \texttt{Coordinators} in phase~$P$, per Round~4 specification, and they will broadcast in Round~1 of phase~$P+1$.

We consider how the crashes of some of the processors $v_1,\ldots,v_k$ affect the accounting, in which multiple broadcasts are charged to previous crashes,  by relating to the situation without these crashes.
By the inductive assumption, there are enough credits to account for the futile broadcasts by any processors $w_1,\ldots, w_\ell$ and $v_1,\ldots,v_k$.

First, consider futile broadcasts by processors $w_1,\ldots, w_\ell$.
It could be the effect of the crashes that some among $w_1,\ldots,w_\ell$ 
do not broadcast in phase $P+1$; but if this occurs, 
then we use correspondingly fewer credits to account for the broadcasts 
of processors among $w_1,\ldots, w_\ell$ that actually take place in~$P+1$, 
saving them for possible future futile broadcasts.

Next, consider futile broadcasts by processors $v_1,\ldots,v_k$.
When a message from processor~$v_i$ is received in Round~1 of  a phase, 
for $1\le i\le k$, then identifier $v_i$ is restored in any list \texttt{Coordinators} 
from which it was deleted before.
Such a restoration neutralizes the crash that contributes to the discrepancies among the lists \texttt{Coordinators} that in turn results in delays that caused $v_i$ to  broadcast along with other processors.
A crash by such~$v_i$ may prevent this neutralization to occur, so that $v_i$ appears in some lists \texttt{Coordinators} while it does not in others.
The accounting is based on the principle that when we associate a credit with a futile broadcast,  then the underlying crash is simultaneously neutralized by restoring the names of all such broadcasters, per Round~2 specification.
It follows that the inductive assumption is not immediately applicable in this situation.
Instead, we charge the futile broadcast by~$v_i$ to its very own crash.
Thereby we save one credit available to pay for the futile broadcast by~$v_i$ if it does not crash. 
Discrepancies among the lists \texttt{Coordinators} created by the presence or absence of $v_i$ 
in them may still occur, which may result in one more futile broadcast after phase~$P+1$, but one credit was saved, which means that the accounting breaks even.
This concludes the proof of Claim~2.

To complete accounting for futile broadcasts, consider the very last phase~$P$ of \PartTwo.
By Claim~2, there are enough credits through phase~$P$ if no crashes occur in~$P$.
But crashes in~$P$ may only diminish the number of messages sent in~$P$, 
while they have no effect on future phases, as there are no phases after~$P$. 
This implies that each futile broadcast proclaiming intent to  become a coordinator 
can be charged to a unique processor that crashed during \PartTwo.
As fewer than $p_2$ processors may crash in \PartTwo, 
these broadcasts contribute $\cO(p_2\cdot p_2)=\cO(p_2\, p)$ messages.

Summing up all messages sent during \PartTwo{} as considered above, 
we obtain that in total there are $\cO(p_2 \, p)=\cO(p\,(p-f_1))$ messages.
\end{proof}

Next we calculate the effort complexity of algorithm \aEP{}.


\begin{lemma}
\label{lem:effort-priority-effort}

The effort of algorithm \aEP\ is $\cO(t+p^{2-a})$.
\end{lemma}

\begin{proof}
The effort of the algorithm is comprised of the contributions by \PartOne, \PartTwo, and \PartThree.

The effort of \PartOne, by Lemma~\ref{lem:effort-part-one} and Equation~\eqref{eqn:Delta-for-Mod-Perm}, is 
\begin{equation}
\label{eqn:first-effort-bound}
\cO(t+\Delta_1  p\log^2 p)
=
\cO\Bigl( t +  p^{1+2a \log_\rho\Delta_0}\log^2 p \Bigr) 
\ .
\end{equation}
To take care of the factor $\log^2 p$ in~\eqref{eqn:first-effort-bound}, we use the identity $\lg^2 p=p^{2\lg\lg p/\lg p}$.
Substituting this into the right-hand side of~\eqref{eqn:first-effort-bound}, we obtain that the effort is
\begin{equation}
\label{eq:effort}
\cO(t+p^{a2\log_\rho\Delta_0}p\log^2 p) 
=
\cO(t+p^{1+a 2\log_\rho\Delta_0+\frac{2\log\log p}{\log p}}) \ .
\end{equation}
By the property $2\lg\lg p/\lg p=o(1)$ and 
by the specification of the parameter~$a$ in~\eqref{eq:a-gives-work}, 
making $2 a \log_\rho \Delta_0< 1- a$, the following inequality holds for sufficiently large $p$:
\[
a 2\log_\rho\Delta_0+\frac{2\lg\lg p}{\lg p}<1-a \ .
\]
This shows that the effort from bound~\eqref{eq:effort}, and so of \PartOne, is $\cO(t+p^{2-a})$.

Next we consider \PartTwo{} and \PartThree.
Let $p_2$ be the number of processors that begin \PartTwo.
There are two cases depending on the magnitude of~$p_2$.

\noindent
Case 1 is for $p_2\ge p-f_1$.
This means that  all chunks of tasks are performed, and the non-faulty processors halt 
in the execution of \aMDP, by Lemma~\ref{lem:justification-of-T1}.
Therefore there is no contribution to the effort by both \PartTwo{} and \PartThree{} in this case.

\noindent
Case 2 is for $p_2< p-f_1$.
Now algorithms \PartTwo{} and \PartThree{} are executed. 

We calculate work complexity first.
\PartTwo{} contributes 
\[
p_2 (p-f_1)<(p-f_1)^2=\cO(p^{2-2a})
=
\cO(t+p^{2-a})
\] 
to  work, as $p-f_1=p^{1-a}$.
Let $f_2$ be the number of crashes that occur during \PartTwo{} and \PartThree. 
By Fact~\ref{fact:DMY}, along with $p-f_1=p^{1-a}$,  we obtain that the work accrued during \PartThree\ is 
\[
\cO(t+p_2(p_2+1))
=
\cO(t+(p-f_1)^2)
=
\cO(t+p^{2-2a})
=
\cO(t+p^{2-a})
\ . 
\]

Next we calculate the message complexity of \PartTwo{} and \PartThree.
The contribution of \PartTwo{}  is $\cO(p\,(p-f_2))$ by Lemma~\ref{lem:communication-prior-DMY}.
From Fact~\ref{fact:DMY}, the message complexity of \PartThree{} is 
\[
\cO((f_2+1) p_2) = \cO(p_2(p_2+1))\ . 
\]
Summing up the two contributions we obtain message complexity that is
\[
\cO(p\, (p-f_1))+\cO(p_2(p_2+1))
=
\cO(p\,(p-f_1))
=
\cO(p^{2-a})\ .
\]
Therefore the effort of \PartTwo\ and \PartThree{} together is $\cO(t+p^{2-a})$.
\end{proof}

Finally, we give the main result for algorithm \aEP{}.


\begin{theorem}
\label{thm:optimized-for-effort}

Algorithm \aEP\ can be instantiated, for any $p$ processors and $t$ tasks, to be a deterministic nonconstructive solution for the \DA\ problem, for any unknown $f<p$,
with effort $\cO (t+p^{1.77})$ .
\end{theorem}

\begin{proof}
We initialize algorithm~\aEP\  with the numeric parameters as specified in the beginning 
of Section~\ref{sec:justifying-initialization}, and such that Lemma~\ref{lem:correctness-effort-priority} 
is applicable.
The work and communication of the algorithm against the unbounded adversary is  $\cO(t+p^{2-a})$, 
by  Lemma~\ref{lem:effort-priority-effort}.
The parameter~$a$ was selected in~\eqref{eq:a-gives-work} so as to satisfy 
the inequality $a<(1+2\lg_\rho \Delta_0)^{-1}$.
We can directly verify that $0.23<(1+2\lg_\rho \Delta_0)^{-1}$, for $\rho = 27/2$  and~$\Delta_0=74$.
Hence the parameter~$a$ can be set equal to~$0.23$, yielding $1.77$ as the numeric value of the exponent.
\end{proof}

\section{Conclusion and Discussion}

\label{sec:discussion}

We presented a constructive algorithm for the \DA\ problem that achieves $\cO(n^{3/2}\text{ polylog } n)$ work, for $n=\max\{p,t\}$, against the unbounded adversary. 
No deterministic algorithms for the \DA\ problem with work complexity~$o(n^2)$, for~$n=\max\{p,t\}$, were known to exist prior to this work, even against linearly-bounded adversaries.
We also presented a nonconstructive solution, giving an algorithm with work $\cO(t+p \textrm{ polylog }p)$, against arbitrary, but known adversaries.

Developing \DA\ solutions that are efficient in terms of effort, defined as the work plus communication complexities, was among the ultimate goals of this work.
We have shown that \DA\ can be solved with effort $\cO(t+p^{1.77})$  against the unbounded adversary.

At the time of its original announcement in~\cite{ChlebusGKS-DISC02}, this was the first known algorithm
simultaneously achieving both work and communication worst-case complexities sub-quadratic in~$p$.
Subsequently Georgiou et al.~\cite{GeorgiouKS05} developed an algorithm with effort $\cO(t+p^{1+\varepsilon})$, for any constant~$\varepsilon>0$, that uses $\varepsilon$ in the algorithm code.
On the other hand,  $\Omega(t+p\log p/\log\log p)$ is the best known general lower bound on work
(the slightly higher corresponding lower bound on work is $\Omega(t+p\log p)$ for shared-memory models).
Thus a gap remains between the lower and upper bounds.

Our nonconstructive algorithm, optimized for work against arbitrary, but known, adversaries 
attains effort $\cO(t+p \text{ polylog } p)$ against arbitrary  known linearly-bounded adversaries. 
It is an open problem whether there exists a constructive \DA\ solution with effort 
$\cO(t+p \text{ polylog } p)$  against the unbounded adversary.

We envision future directions of work on problems related to \DA\ spanning both the static and dynamic 
variants of the problem.
For the static variants, it is worth investigating the settings where the processors have
possibly partial information about the set of tasks.
Recently, Drucker et al.~\cite{DruckerKO12} studied the communication complexity of a distributed task allocation problem in which each processor receives as input a subset of all tasks, where the goal is to assign each task to a unique processor.
For the  dynamic variants, research may pursue the definition of suitable adversarial models, 
the issues of task generation, and the analysis of competitive performance.
Alistarh et al.~\cite{AlistarhABGG14} recently proposed a formalization 
of a dynamic version  of \DA\ in asynchronous shared-memory systems and gave a solution that is work-optimal within polylogarithmic factors.
Georgiou and Kowalski~\cite{GeorgiouK11} proposed a competitive-performance framework to study dynamic task arrivals in a message-passing system with processors prone to crashes and restarts.



\bibliographystyle{abbrv}

\bibliography{doing-it-all}

\end{document}